\documentclass[pra,aps,twocolumn,superscriptaddress,longbibliography]{revtex4-1}
\usepackage[colorlinks=true, citecolor=red, urlcolor=blue ]{hyperref}
\usepackage{graphicx}
\usepackage{bm}
\usepackage{bbm}
\usepackage{amsmath,amsfonts}
\usepackage{amsthm}
\usepackage{hyperref}
\usepackage{xcolor}
\hypersetup{
    urlcolor=blue,
    citecolor=blue
           }

\usepackage{makecell}
\usepackage{derivative}
\usepackage{braket}
\theoremstyle{plain}
\usepackage{color}
\usepackage{amssymb}
\usepackage{amsthm}
\usepackage{amsfonts}
\usepackage{float}
\usepackage{tabularx}
\usepackage{graphicx}
\usepackage[utf8]{inputenc}

\usepackage{mathtools}
\usepackage{esvect}
\usepackage{wrapfig}
\usepackage{amsthm}
\usepackage{verbatim}
\usepackage{bbm}
\usepackage[normalem]{ulem}

\usepackage{enumitem}
\usepackage{fmtcount}
\usepackage{booktabs}
\usepackage{csquotes}
\usepackage{epsfig}

\usepackage{tabularx}
\usepackage{graphicx}
\usepackage{amsmath}
\usepackage{braket}
\usepackage{latexsym}
\usepackage{bm}
\usepackage{graphics,epstopdf}
\usepackage{enumitem}
\usepackage{fmtcount}
\usepackage{booktabs}
\usepackage{csquotes}
\usepackage{epsfig}

\theoremstyle{plain}

\def\bea{\begin{eqnarray}}
\def\eea{\end{eqnarray}}
\def\ba{\begin{array}}
\def\ea{\end{array}}

\def\beq{\begin{equation}}
\def\eeq{\end{equation}}

\usepackage[normalem]{ulem}
\usepackage{float}
\usepackage{graphicx}
\usepackage{dcolumn}
\usepackage{amssymb}
\usepackage{appendix}
\usepackage{physics}
\usepackage{mathtools}
\usepackage{esvect}
\usepackage{wrapfig}
\usepackage{amsthm}
\usepackage{verbatim}
\usepackage{bbm}
\usepackage{soul}

\usepackage[mathscr]{euscript}

\def\({\left(}
\def\){\right)}
\def\[{\left[}
\def\]{\right]}

\newtheorem{corollary}{Corollary}

\newtheorem{lemma}{Lemma}

\newtheorem{proposition}{Proposition}

\begin{document}

\title{Homomorphic Quantum Error Correction}

\author{Kornikar Sen}
\affiliation{Departamento de Física Teórica, Universidad Complutense de Madrid, 28040 Madrid, Spain}

\author{Miguel A. Martin-Delgado}
\affiliation{Departamento de Física Teórica, Universidad Complutense de Madrid, 28040 Madrid, Spain}
\affiliation{CCS-Center for Computational Simulation, Campus de Montegancedo Universidad Politecnica de Madrid (UPM), 28660 Boadilla del Monte, Madrid, Spain}

\begin{abstract}
We study the algebraic compatibility between quantum homomorphic encryption and stabilizer quantum error correction in an honest-but-curious model, restricting the encryption to the transversal block-Pauli mask $(X^aZ^b)^{\otimes n}$ controlled by two secret key bits $a$ and $b$. We prove a necessary and sufficient criterion for an $[[n,1,d]]$ stabilizer code to remain compatible with this mask, namely that the transversal $X$ and $Z$ operators commute with every stabilizer generator, and we specialize it to a classical-code criterion for Calderbank--Shor--Steane codes. Combining code-space preservation with the logical action of the mask, every stabilizer code falls, with respect to each of the two transversal mask operators, into exactly one of three classes: incompatible, compatible but trivially acting, or compatible with a nontrivial logical Pauli action; the two operators need not fall into the same class. We verify this classification by exact symplectic computation for rotated and planar surface codes, toric codes, triangular and toric color codes, the $[[15,1,3]]$ code and several small codes, and find that the class is fixed by the lattice size and the boundary presentation rather than by the topological phase: the surface-code family realizes all three classes, whereas color codes realize only the two compatible ones. For encrypted computation we show that the Shor code does not implement the logical $T$ gate by a transversal $T^{\otimes 9}$, whereas suitable triorthogonal codes do so up to a Clifford correction. We formulate logical-Pauli masking as a code-independent algebraic template. Syndrome extraction and recovery are key-independent whenever the mask preserves the code space, so that the client only needs the block size and the class of the code to decrypt.
\end{abstract}


\maketitle

\section{Introduction}
\label{sec1}
Quantum homomorphic encryption (QHE) allows a client to delegate a quantum computation to a remote server without revealing the input, the output or, in the strongest variants, the computation itself \cite{PhysRevLett109150501,Liang_2013,Tan_2016,Broadbent_2015,Dulek_2018}. In the setting considered here, the client sends quantum data to the server, the server carries out the computation, and the client receives the final output. The basic primitive is the quantum one-time pad: the client masks every qubit of the quantum data with a random Pauli operator $X^{a_j}Z^{b_j}$, the server applies the requested circuit, and the client undoes the mask with keys that have been updated according to the gates performed. Clifford gates always map Pauli operators to Pauli operators under conjugation, so the corresponding key updates are classical; non-Clifford gates such as $T$ produce byproduct operators that can be removed by gate teleportation with the help of entangled states supplied by the client, as shown in Ref. \cite{Liang30}. Information-theoretically secure QHE for arbitrary circuits is impossible without exponential resources \cite{nogo1,nogo2}, which is why practical schemes trade generality for efficiency and target a fixed universal gate set. The scheme of Ref.~\cite{Liang30} has been used to encrypt Grover's algorithm \cite{Grover}, the recursive Bernstein--Vazirani algorithm \cite{BV} and semiclassical Szegedy walks \cite{SQW}, and it underlies a ciphertext dimension-reduction method \cite{Grover2}. Several QHE protocols have been demonstrated experimentally \cite{exp1,exp2}. 

Quantum computation \cite{NielsenChuang2000,GalindoMartinDelgado2002RMP} is now offered as a cloud service, but the hardware on which it runs is noisy. Superconducting \cite{DiVincenzo_2000,superqubits,PhysRevB77024522,qc,Wendin_2017,Gambetta_2017,GoogleColorCode2025}, trapped-ion \cite{PhysRevLett.74.4091,trapped_ion,PhysRevX.13.041052,Nigg2014TopologicalQubit} and neutral-atom \cite{cold_atom1,Bluvstein_2023,Evered2025NeutralAtomArchitecture} processors are now operated as remote services, and the same devices are the ones on which quantum error correction (QEC) \cite{QEC1,Roffe_2019} is being demonstrated with stabilizer codes \cite{stabilizer}: the Shor code \cite{Shor}, the five-qubit code \cite{fivequbit,fivequbit_exp}, Calderbank--Shor--Steane (CSS) codes \cite{CSS1,CSS2}, toric and surface codes \cite{toriccode1,toriccode2,toriccode3} and color codes \cite{colorcode,BombinMartinDelgado2007PRL,BombinMartinDelgado2007PRB} to reduce the impact of noise on the quantum data. A client who uploads encrypted quantum data to such a server therefore needs two guarantees at once: that the server cannot read the data, and that the server can still detect and correct errors on it. The two requirements interact. Encryption acts on the physical qubits, error correction is defined by the code space, and a mask that takes an encoded state out of the code space is not merely useless but harmful, because the syndrome bits of the generators that anticommute with the mask are then shifted by the key bits, so that the server learns the key and its recovery tends to undo the encryption. 

The interplay between QHE and QEC has been examined before. Ouyang, Tan and Fitzsimons constructed QHE schemes directly from quantum codes, using the code structure itself as the source of secrecy \cite{OuyangTanFitzsimons2018}; permutational-key schemes with homomorphic error correction \cite{OuyangRohde2022} and QHE schemes built on CSS codes \cite{Sohn2024ErrorCorrectable} have also been proposed. These works design new encryption schemes with error-correcting features. The questions addressed here are different and, in a sense, prior to them. Given an $[[n,1,d]]$ standard stabilizer code and the standard Pauli-mask encryption of Ref.~\cite{Liang30}, restricted to the transversal block mask $U_{\rm enc}(a,b)=(X^aZ^b)^{\otimes n}$ controlled by two randomly selected key bits $(a,b)$, we ask:
\begin{enumerate}
    \item[(a)] When does the encrypted encoded state remain a codeword?
    \item[(b)] What logical operation does the mask implement on it?
    \item[(c)] Which gates can the server evaluate on the encrypted encoded block?
\end{enumerate}
We stress from the outset that this block mask is not an independent physical quantum one-time pad on each of the $n$ qubits. In the special scenarios, where $X^{\otimes n}$ and $Z^{\otimes n}$ act as single-logical-qubit Pauli operators, the encryption reduces to a logical one-time pad of the encoded qubit. Our statements concern code-space compatibility and protocol correctness in an honest-but-curious model, meaning that if the server follows the protocol the client recovers the intended logical state. Security against a malicious server, and fault tolerance of the resulting circuits, are external assumptions that we do not establish here.

Our contributions are the following. First, to address (a), we prove that the transversal block mask preserves the code space of an $[[n,1,d]]$ stabilizer code if and only if $X^{\otimes n}$ and $Z^{\otimes n}$ commute with every stabilizer generator. We reduce this condition, for a CSS code built from classical codes $C_2\subset C_1$, to the statement that the all-ones classical word belongs to $C_1$ and $C_2^{\perp}$. Second, to examine (b), combining code-space preservation with the logical action of the mask, we show that every stabilizer code falls into one of the three classes with respect to each mask operator: (i) the mask does not preserve the code space, (ii) the mask acts trivially on the code space, (iii) the mask acts nontrivially on the code space but still preserves it. We give a closed-form parity rule for CSS codes that decides the class from the weights of the stabilizer generators and of the logical operators. We implement the criterion as a short symplectic routine and classify rotated and planar surface codes, toric codes, triangular and toric color codes, the $[[15,1,3]]$ code, and several small codes. The surface-code family realizes all three classes, whereas color codes realize only classes (ii) and (iii), because all their stabilizer generators have even weight; in both families the class is fixed by the lattice size and by the boundary presentation and not by the topological phase, so that the class of the code presentation, and not only the code family, has to be specified in the protocol. Third, to find (c), i.e., the possibility of computation on encrypted data, we show that a transversal $T^{\otimes 9}$ does not implement the logical $T$ gate on the \textit{encrypted} Shor code, whereas suitable \textit{encrypted} triorthogonal codes do so up to a known Clifford correction. We formulate logical-Pauli masking as an algebraic template that is compatible with any stabilizer code at the price of code-dependent logical gates. Finally, we observe that whenever the mask preserves the code space the syndrome extraction and the recovery are key-independent, so that under Pauli noise the logical error rate of the encrypted memory is the same as that of the plaintext memory up to the two additional transversal Pauli layers introduced by encryption and decryption.

The paper is organized as follows. Section~\ref{sec2} fixes notation and summarizes the considered QHE scheme on which the rest of the paper is built. Section~\ref{sec3} treats encrypted storage: the bit-flip and Shor codes are analyzed as examples, the general stabilizer and CSS criteria are proved, the three-class classification is stated and verified by exact symplectic computation, and logical Pauli masking and the consequences for noise are discussed. Section~\ref{sec4} treats the implementation of logical non-Clifford gates on encrypted encoded states. Section~\ref{sec5} summarizes the scope of the results. The appendices collect supporting material: Appendix~\ref{A1} reviews gate teleportation, the key-update rules and a two-qubit worked example of the QHE scheme; Appendix~\ref{A2} gives the encrypted amplitudes of the bit-flip code and the explicit code-space preservation lemma for the Shor code; Appendix~\ref{A3} checks the CSS conditions for the Steane code; Appendix~\ref{X} verifies the compatibility condition stated in Proposition \ref{thm:stabilizer-compatibility} for $[[15,1,3]]$ code and its extension to tetrahedral three-dimensional color codes; and Appendix~\ref{A4} constructs explicit unitary representatives of the logical gates for the Shor code.

\section{Homomorphic encryption with Pauli masks}\label{sec2}
This section summarizes the QHE scheme of Ref.~\cite{Liang30} in the form used {in the rest of the paper}; the derivations, the key-update table and a worked example are given in Appendix~\ref{A1}. Throughout the paper, we consider the universal gate set as $\mathcal U=\{X,Z,S,H,CNOT,T,T^\dagger\}$, whose Clifford part is $\{X,Z,S,H,CNOT\}$ and whose non-Clifford part is $\{T,T^\dagger\}$. 

Consider a client who wants a server to apply a circuit, $\boldsymbol{\mathcal{C}}$, decomposed into $n_u$ gates from $\mathcal U$ of which $n_T$ are non-Clifford, to an $n$-qubit state $\ket\psi$. The protocol has three stages. In the encryption stage the client draws $2n$ independent random bits $\{(a_j,b_j)\}_{j=1}^n$, i.e., the secret key, prepares $n_T$ Bell pairs $\ket{\Phi}_{s_ic_i}$, sends the server the state $\ket{\psi^{\rm enc}}=\bigotimes_j X^{a_j}Z^{b_j}\ket\psi$ together with $n_T$ qubits, $s_i$, keeping the qubits $c_i$. In the evaluation stage, the server applies the gates in the prescribed order. For a single-qubit Clifford gate $G$, involved in $\boldsymbol{\mathcal{C}}$, acting on encrypted qubit $i$, the identity
\begin{equation}\label{eq:keyupdate}
G\,X^{a}Z^{b}=X^{f^{G}_x(a,b)}Z^{f^{G}_z(a,b)}\,G
\end{equation}
holds for classical functions $f^G_x,f^G_z:\{0,1\}^2\to\{0,1\}$ that depend only on $G$, so the server simply records $(f^G_x,f^G_z)_i$ while applying $G$ on the encrypted qubit $i$ and the client updates the key pair of qubit $i$ classically; for $CNOT$ acting on encrypted qubits $i$ and $j$ the key pairs, $\{(a_i,b_i),(a_j,b_j)\}$, corresponding to this pair of qubits become $(a_i,b_i\oplus b_j)$ and $(a_i\oplus a_j,b_j)$. The complete list of key update rules is reproduced in Fig.~\ref{fig2} and Table~\ref{table1} of Appendix~\ref{A1}. Non-Clifford gates cannot be handled in this way. A $T$ gate on an encrypted qubit $i$ satisfies
\begin{equation}\label{eq:Tbyproduct}
T\,X^{a_i}Z^{b_i}\ket\psi\simeq (S^\dagger)^{a_i}X^{a_i}Z^{a_i\oplus b_i}\,T\ket\psi ,
\end{equation}
and $T^\dagger$ gate satisfies the analogous relation with $S$ in place of $S^\dagger$, where $\simeq$ denotes equality up to a global phase. The key-dependent byproduct $(S^\dagger)^{a_i}$ can neither be removed by the server, who does not know $a_i$, nor be absorbed into a Pauli key. The server therefore swaps the affected qubit, $i$, with the Bell-pair qubit $s_i$ it received from the client, and the byproduct is removed at the decryption stage. At the end of the evaluation stage, the server sends all the qubits back to the client along with the key update rules. In the final step, the decryption stage, following the sequence in which Clifford and non-Clifford gates appear in $\boldsymbol{\mathcal{C}}$, the client either updates the keys following the classical key update functions for each application of Clifford gates or measures the qubit pair $c_is_i$ it holds in the $S^{a_i}$-rotated Bell basis, which teleports the corrected state and produces two further key bits (Appendix~\ref{A1}). After all $n_T$ such measurements and $n_u$ key updates corresponding to application of each of the gates the client obtains final keys $(a_i^f,b_i^f)$ and applies $\bigotimes_i X^{a_i^f}Z^{b_i^f}$ to recover the output state of the circuit.

Two features of this scheme matter for what follows. The first is that everything the server does is independent of the key: the Clifford gates are applied as they are, the $T$ gates are followed by a swap that does not depend on $\{(a_i,b_i)\}$, and the classical update functions are fixed by the gate sequence. The second is that the correctness of the scheme rests on the fact that Clifford gates normalize the Pauli group under conjugation, Eq.~\eqref{eq:keyupdate}, while $T$ maps a Pauli operator to a Pauli times a Clifford, Eq.~\eqref{eq:Tbyproduct}. 

In the rest of the paper, operators acting on physical qubits are written without decoration ($X$, $Z$, $H$, $S$, $T$, $CNOT$) and operators acting on logical qubits carry a bar ($\bar X$, $\bar Z$, $\bar T$). A logical state is denoted as $\ket{\bar\psi}$. The logical Pauli operators, $(\bar X, \bar Z)$, by construction satisfy $\bar X\ket{\bar 0}=\ket{\bar 1}$, $\bar X\ket{\bar 1}=\ket{\bar 0}$, $\bar Z\ket{\bar 0}=\ket{\bar 0}$ and $\bar Z\ket{\bar 1}=-\ket{\bar 1}$, which physically corresponds to multi-qubit operations on the encoded block. The qubits $w_1,\dots,w_n$ are the physical qubits of a code block, encoding a single qubit, and the mask is restricted to $a_j=a$ and $b_j=b$ for all $j$. This restricted block-Pauli mask, or the logical-Pauli mask that replaces it in Sec.~\ref{sec4}, is not an independent physical quantum one-time pad on every physical qubit. Accordingly, the results below concern code-space compatibility and protocol correctness rather than information-theoretic security in the full physical quantum one-time pad sense.

\section{Quantum error correction on encrypted data}\label{sec3}
In this section, the client stores encrypted data on the server, and the server is only asked to keep it error-free: the encryption stage involves encoding as well as encryption, and the evaluation stage consists of syndrome extraction and recovery, but no logical gates are applied. The question is whether the transversal block mask $U_{\rm enc}(a,b)=(X^aZ^b)^{\otimes n}$ is compatible with a given stabilizer code, in the sense that every encrypted codeword is again a codeword and the server can correct errors without learning $(a,b)$. Two examples, the bit-flip code and the Shor code, are worked out first to fix ideas; the general criterion for stabilizer codes and its CSS specialization follow; the section closes with a classification of stabilizer codes with respect to its transformation under the mask and with a discussion of noise.

\subsection{Bit-flip code}
\label{secIIIA}
The purpose of this example is to show, in the simplest setting, that a code can correct an encrypted state without knowing the key. The three-qubit bit-flip code encodes $\ket{\bar 0}=\ket{000}$ and $\ket{\bar 1}=\ket{111}$, so that $\ket{\bar\psi}=c_0\ket{000}+c_1\ket{111}$, detects a single bit flip by measuring the stabilizer generators $Z_1Z_2$ and $Z_2Z_3$ on the physical qubits and corrects it applying $X$ to the qubit indicated by the two-bit syndrome.
To store this state, $\ket{\bar \psi}$, in the server, the client draws two bits $(a,b)$ and applies $(X^aZ^b)^{\otimes 3}$ on $\ket{\bar \psi}$. $X^{\otimes 3}$, being the logical $\bar{X}$ gate, exchanges $\ket{000}$ and $\ket{111}$, and similarly, $Z^{\otimes 3}$, being the logical $\bar{Z}$ gate, multiplies $\ket{111}$ by $-1$. As a result, the encrypted state is again of the form $c_0'\ket{000}+c_1'\ket{111}$, where $(c_0',c_1')$ is obtained from $(c_0,c_1)$ by the logical Pauli action $\bar X^a\bar Z^b$; the four cases, corresponding to four possible keys, $(a,b)$, are listed in Table~\ref{T1} of Appendix~\ref{A2}.
The server measures the same two stabilizer generators, obtains the same syndrome that it would obtain on the plaintext codeword affected by the same error, and applies the same correction. Nothing in this procedure depends on $(c_0',c_1')$, hence nothing depends on the key, and no key update is required. When the client retrieves the block it applies $(Z^b)^{\otimes 3}(X^a)^{\otimes 3}$ or $\bar Z^b\bar X^a$ and recovers $\ket{\bar\psi}$. The phase-flip code is treated identically with the roles of $X$ and $Z$ interchanged.

\subsection{Shor code}\label{sec:shor}
\label{secIIIB}
The purpose of this example is to exhibit a code for which the transversal mask preserves the code space but implements a logical operator different from the one suggested by its physical label. The Shor code encodes one logical qubit into nine physical qubits arranged in three blocks of three:

\begin{eqnarray}
|{{0}}\rangle\rightarrow \ket{\bar{0}}=(\ket{000}+\ket{111})^{\otimes 3}/\left(2\sqrt{2}\right),\label{myeq8}\\
    |{{1}}\rangle\rightarrow \ket{\bar{1}}=(\ket{000}-\ket{111})^{\otimes 3}/\left(2\sqrt{2}\right),\label{myeq9}
\end{eqnarray}
and corrects any single-qubit error \cite{NielsenChuang2000}. Bit flips are detected by the six parity checks $Z_1Z_2$, $Z_2Z_3$, $Z_4Z_5$, $Z_5Z_6$, $Z_7Z_8$, $Z_8Z_9$ inside the blocks, phase flips by the two checks $X_1X_2X_3X_4X_5X_6$ and $X_4X_5X_6X_7X_8X_9$ that compare the signs of neighboring blocks, and the faulty qubit is restored by the corresponding $X$ or $Z$ gate. We assume independent single-qubit errors with probability $p\ll 1$.

To store $\ket{\bar\psi}=c_0\ket{\bar 0}+c_1\ket{\bar 1}$ in the cloud, the client applies the same physical Pauli operator to every qubit of the block,
\begin{equation}
\ket{\bar\psi^{\rm enc}}=(X^aZ^b)^{\otimes 9}\ket{\bar\psi}.
\end{equation}
With the basis of Eqs.~\eqref{myeq8}--\eqref{myeq9}, $X^{\otimes 3}$ leaves $\ket{000}+\ket{111}$ invariant and multiplies $\ket{000}-\ket{111}$ by $-1$, so $X^{\otimes 9}$ acts on the code space as the logical phase flip $\bar Z$, whereas $Z^{\otimes 3}$ exchanges the two block states, $\ket{\bar{0}}$ and $\ket{\bar{1}}$, and $Z^{\otimes 9}$ acts as the logical bit flip $\bar X$. The transversal operators must therefore not be identified with the logical operators of the same name: for the Shor code in this basis,
\begin{equation}\label{eq:shor-swap}
X^{\otimes 9}\simeq\bar Z,\qquad Z^{\otimes 9}\simeq\bar X\qquad\text{on the code space.}
\end{equation}
In either case the encrypted state is $c_0'\ket{\bar 0}+c_1'\ket{\bar 1}$ with $(c_0',c_1')$ obtained from $(c_0,c_1)$ by a key-dependent logical Pauli transformation; the explicit statement and proof are given as Lemma~\ref{lem:shor} in Appendix~\ref{A2}. 

The server receives the encrypted state, $c_0'\ket{\bar 0}+c_1'\ket{\bar 1}$, from the client, which is a valid codeword. After receiving, it measures the eight stabilizer generators, obtains a syndrome that depends only on the error and neither on the logical amplitudes, $(c_0,c_1)$, nor on the keys, $(a,b)$, and applies the corresponding recovery. In the ideal stabilizer-code model this is the whole evaluation stage. A fault-tolerant implementation would in addition require specifying syndrome-extraction circuits and tracking how the encrypted Pauli frame propagates through ancilla-assisted measurements; at this stage the result establishes compatibility between the restricted mask and ideal QEC, not a security proof against arbitrary server deviations. When desired, the client applies $(X^aZ^b)^{\otimes 9}$ again and recovers the plaintext codeword.

\subsection{Stabilizer codes and homomorphic compatibility}
The two examples share the property that the mask maps every codeword to a codeword. By virtue of this property, homomorphic encryption remains compatible with error correction, thereby enabling the server to perform error correction on the encrypted physical qubits without disclosing the encryption keys.

We now characterize this property for general stabilizer codes. Let $\mathcal S\subset G_n$ be an abelian subgroup of the $n$-qubit Pauli group not containing $-I$, with $n-k$ independent generators $\{g_i\}$, and let $C(\mathcal S)=\{\ket\psi:\,g\ket\psi=\ket\psi\ \forall g\in\mathcal S\}$ be the $2^k$-dimensional code space of the corresponding $[[n,k,d]]$ code \cite{stabilizer}. A set of Pauli errors $\{E_j\}$ is correctable if $E_j^\dagger E_\ell\notin N(\mathcal S)\setminus\mathcal S$ for all $j,\ell$, where $N(\mathcal S)$ is the normalizer of $\mathcal S$ in $G_n$. Thus, every product of two errors either belongs to $\mathcal S$ or anticommutes with some generator and is detected.
Logical Pauli operators are elements of $N(\mathcal S)\setminus\mathcal S$. For $k=1$, one fixes representatives $\bar Z,\bar X\in N(\mathcal S)$ satisfying $\bar X\bar Z=-\bar Z\bar X$, then defines the logical basis states $\ket{\bar 0}$ and $\ket{\bar 1}$ via $\bar Z\ket{\bar x}=(-1)^x\ket{\bar x}$. The relations $\bar X\ket{\bar 0}=\ket{\bar 1}$ and $\bar X\ket{\bar 1}=\ket{\bar 0}$ follow because $\bar X\ket{\bar 0}$ is the unique state stabilized by all $g_i$ and by $-\bar Z$. 

Encryption is compatible with the code if $g\,U_{\rm enc}\ket{\bar\psi}=U_{\rm enc}\ket{\bar\psi}$ for all $g\in\mathcal S$ and all codewords, that is, if $U_{\rm enc}$ preserves the stabilizer eigenvalues. For simplicity we state the result for $k=1$; it extends verbatim to code-space preservation for $k>1$.

\begin{proposition} \textbf{Homomorphic compatibility for stabilizer codes}
\label{thm:stabilizer-compatibility}
\\
Let $C(\mathcal{S})$ be an $[[n,1,d]]$ stabilizer code with stabilizer group $\mathcal{S}=\langle g_1,\dots,g_{n-1}\rangle$ generated by independent commuting Pauli operators $\{g_1,\dots,g_{n-1}\}\subseteq {G}_n$. Consider the transversal Pauli encryption
\begin{equation}
    U_{\mathrm{enc}}(a,b) = (X^a Z^b)^{\otimes n}, \quad (a,b)\in\{0,1\}^2,
\end{equation}
acting on the $n$ physical qubits encoding a single logical qubit. Then $U_{\mathrm{enc}}(a,b)$ is homomorphically compatible with $C(\mathcal{S})$ (i.e., maps the code space into itself for all keys $(a,b)$) if and only if
\begin{equation}\label{eq:stab-comm}
[X^{\otimes n}, g_i] = 0 \text{ and } [Z^{\otimes n}, g_i] = 0 \text{ for all } i=1,\dots,n-1.
\end{equation}
\end{proposition}

\begin{proof}
Write $\boldsymbol X:=X^{\otimes n}$ and $\boldsymbol Z:=Z^{\otimes n}$ and let $\mathcal C$ be the code space. If Eq.~\eqref{eq:stab-comm} holds and $\ket\phi\in\mathcal C$, then $g_i\boldsymbol X\ket\phi=\boldsymbol Xg_i\ket\phi=\boldsymbol X\ket\phi$ for every $i$, and likewise $g_i\boldsymbol Z\ket\phi=\boldsymbol Zg_i\ket\phi=\boldsymbol Z\ket\phi$, so $\boldsymbol X\ket\phi,\boldsymbol Z\ket\phi\in\mathcal C$. As a result, since $U_{\rm enc}(a,b)$ is a product of $\boldsymbol X$ and $\boldsymbol Z$, it preserves $\mathcal C$ for all keys. 

Conversely, suppose $U_{\rm enc}(a,b)$ preserves $\mathcal C$ for all $(a,b)$; the keys $(1,0)$ and $(0,1)$ show that $\boldsymbol X$ and $\boldsymbol Z$ individually preserve $\mathcal C$. If $[\boldsymbol X,g_i]\neq 0$ for some $i$, then, both being Pauli operators, they anticommute, and for any codeword $g_i\boldsymbol X\ket\phi=-\boldsymbol Xg_i\ket\phi=-\boldsymbol X\ket\phi$, so $\boldsymbol X\ket\phi$ is a $-1$ eigenstate of $g_i$ and lies outside $\mathcal C$, which is a contradiction. Hence ${\boldsymbol X}$ must commute with all the generators, $g_i$, for $U_{\rm enc}(a,b)$ to preserve $\mathcal C$. The same argument applies to $\boldsymbol Z$.
\end{proof}

Proposition~\ref{thm:stabilizer-compatibility} turns compatibility into a parity count. The generators $Z_1Z_2$ and $Z_2Z_3$ of the bit-flip code have even weight and commute with both masks. The eight generators of the Shor code, six $Z$-type checks of weight two and two $X$-type checks of weight six, all have even weight and again commute with $X^{\otimes 9}$ and $Z^{\otimes 9}$, which is the structural reason behind Lemma~\ref{lem:shor}. The proposition also explains why incompatibility is fatal rather than inconvenient: if some $g_i$ anticommutes with $\boldsymbol X$ (or $\boldsymbol{Z}$), the syndrome bit of $g_i$ measured by the server is shifted by the key bit $a$ (or $b$), so that in the absence of errors it equals the key bit, the server reads the key from the syndrome, and the recovery it applies tends to undo the mask.

\subsection{Compatibility condition for CSS codes}
CSS codes \cite{CSS1,CSS2} are stabilizer codes whose generators split into purely $Z$-type checks derived from a classical code $C_1$ and purely $X$-type checks derived from a classical code $C_2\subset C_1$; examples are the Shor code \cite{Shor}, the Steane code \cite{CSS1}, and the toric \cite{toriccode2} and color \cite{colorcode} codes. Because the two families of checks are decoupled, the commutation conditions of Proposition~\ref{thm:stabilizer-compatibility} become two conditions on the classical codes.

\begin{corollary} \textbf{Compatibility condition for CSS codes}
\label{thm:css-compatibility}\\
Let $C_1$ and $C_2$ be classical binary linear codes of length $n$ and dimensions $k_1$ and $k_2$, respectively, with $C_2 \subset C_1$ and $k_1>k_2$. The associated CSS code $\mathrm{CSS}(C_1,C_2)$ encodes $k_1-k_2$ logical qubits into $n$ physical qubits with computational basis states
\begin{equation}\label{eq:css-basis}
|x + C_2\rangle = \frac{1}{\sqrt{|C_2|}} \sum_{y \in C_2} |x \oplus y\rangle,
\end{equation}
where $x \in C_1$ labels a coset in $C_1/C_2$, and $\oplus$ denotes bitwise addition modulo $2$. Consider the transversal Pauli encryption
\begin{equation}
    U_{\mathrm{enc}}(a,b) = (X^a Z^b)^{\otimes n}, \quad (a,b)\in\{0,1\}^2.
\end{equation}
Let $e=(1,\dots,1)\in\mathbb{F}_2^n$. Then $\mathrm{CSS}(C_1,C_2)$ is compatible with this homomorphic encryption, in the sense that $U_{\mathrm{enc}}(a,b)$ preserves the CSS code space for all keys and maps encoded basis states to encoded basis states up to phases, if and only if
\begin{enumerate}
    \item
    \begin{equation}\label{eq:css-cond-X}
    e\in C_1,
    \end{equation}
    equivalently, $x\oplus e\in C_1$ for every $x\in C_1$.
    \item
    \begin{equation}\label{eq:css-cond-Z-even}
    e\in C_2^\perp ,
    \end{equation}
    equivalently, all words in $C_2$ have even Hamming weight. This even-weight condition is precisely what makes the phase $(-1)^{w(x\oplus y)}$, where $w(\cdot)$ denotes Hamming weight, independent of the representative $y\in C_2$ for each fixed $x\in C_1$.
\end{enumerate}
Under these equivalent conditions, $X^{\otimes n}$ and $Z^{\otimes n}$ preserve the CSS code space, and encrypted states remain correctable by the corresponding CSS code.
\end{corollary}

\begin{proof}
Write $\boldsymbol X:=X^{\otimes n}$ and $\boldsymbol Z:=Z^{\otimes n}$. Distinct CSS basis states have disjoint computational-basis support, and $\boldsymbol X$ shifts the support of $|x+C_2\rangle$ from the coset $x+C_2$ to $x\oplus e+C_2$ while $\boldsymbol Z$ keeps it inside $x+C_2$. Hence $\boldsymbol X$ preserves the code space if and only if $\boldsymbol X|x+C_2\rangle=|x\oplus e+C_2\rangle$ is again a basis state, that is, $x\oplus e\in C_1$ for all $x\in C_1$; since $C_1$ is linear and contains $0$, this is condition \eqref{eq:css-cond-X}. For $\boldsymbol Z$ one has $\boldsymbol Z|x\oplus y\rangle=(-1)^{w(x\oplus y)}|x\oplus y\rangle$ and $w(x\oplus y)\equiv w(x)+w(y)\pmod 2$, so
\begin{equation}
\boldsymbol Z|x+C_2\rangle=\frac{(-1)^{w(x)}}{\sqrt{|C_2|}}\sum_{y\in C_2}(-1)^{e\cdot y}|x\oplus y\rangle .
\end{equation}
If $e\in C_2^\perp$ every sign $(-1)^{e\cdot y}$ equals $1$ and $\boldsymbol Z|x+C_2\rangle=(-1)^{w(x)}|x+C_2\rangle$. If $e\notin C_2^\perp$, the character $y\mapsto(-1)^{e\cdot y}$ is nontrivial on the group $C_2$, so $\sum_{y\in C_2}(-1)^{e\cdot y}=0$ and $\boldsymbol Z|x+C_2\rangle$ is orthogonal to $|x+C_2\rangle$; having support only on the coset $x+C_2$, it is orthogonal to every other basis state as well, hence to the whole code space. Condition \eqref{eq:css-cond-Z-even} is therefore necessary and sufficient for $\boldsymbol Z$, and since every $U_{\rm enc}(a,b)$ is a product of powers of $\boldsymbol X$ and $\boldsymbol Z$ up to phase, the two conditions together are necessary and sufficient for all keys.
\end{proof}

The Steane code is an example of a CSS code. One can easily check from the classical code words, $C_1$ and $C_2$, of the Steane code that both conditions in the above corollary are satisfied in this case  (see Appendix \ref{A3} for more detail). In the CSS/Hamming-code presentation of Steane code used here, the all-$X$ operator exchanges the two logical codewords, while the all-$Z$ operator distinguishes their parity; hence $X^{\otimes 7}$ and $Z^{\otimes 7}$ provide valid logical Pauli representatives \cite{colorcode}. Thus, on the code space, the transversal block-Pauli masking considered here coincides with a logical Pauli masking of the encoded qubit.

\subsection{Classification of stabilizer codes under transversal Pauli masking}\label{sec:classification}

Proposition~\ref{thm:stabilizer-compatibility} and Corollary~\ref{thm:css-compatibility} decide whether the transversal encryption operators $\boldsymbol{X}$ and $\boldsymbol{Z}$ preserve the stabilizer and CSS code space. For a homomorphic encryption scheme, however, preservation of the code space is not the end of the story: the client also needs to know how the encrypted encoded state is related to the plaintext encoded state, because this relation is what the server-side key update and the final decryption have to undo. Combining both requirements, every stabilizer code $C(\mathcal S)$ falls into exactly one of three classes with respect to each of the two masking operators $P\in\{\boldsymbol{X},\boldsymbol{Z}\}$:
\begin{enumerate}[label=(\roman*)]
    \item $P\notin N(\mathcal S)$. The mask anticommutes with at least one stabilizer generator and takes the code space to an orthogonal syndrome space. Transversal Pauli encryption is incompatible with the code. Worse, the syndrome bits of the generators that anticommute with the mask are shifted by the key bits, so the server learns the key, and its recovery would tend to undo the encryption.
    \item $P\in\mathcal S$ up to a sign. The mask preserves the code space but acts as the identity, up to a global phase, on every codeword. The corresponding key bit is therefore not encrypted at all and can be dropped from the key without loss.
    \item $P\in N(\mathcal S)\setminus\mathcal S$. The mask preserves the code space and acts on it as a nontrivial logical Pauli operator $\bar P$. This is the only case in which the transversal mask implements a genuine one-time pad on the encoded data, and the identity of $\bar P$ fixes the code-dependent key-update rule of Section \ref{sec4}.
\end{enumerate}
The classes are decided by binary linear algebra. Writing an $n$-qubit Pauli operator as a vector $(x|z)\in\mathbb F_2^{2n}$, membership in $N(\mathcal S)$ is the vanishing of the symplectic products of $X^{\otimes n}$ and $Z^{\otimes n}$
with the stabilizer generators, and membership in $\mathcal S$ up to a sign is membership of $X^{\otimes n}$ and $Z^{\otimes n}$ in the $\mathbb F_2$-row space of the generator matrix, which separates class (ii) from class (iii). Once a mask is known to be in class (iii), its logical action is read off from the symplectic products with a fixed symplectic basis of logical operators $\bar X_j,\bar Z_j$, $j=1,\dots,k$. For CSS codes the classification can be stated in closed form.

\begin{proposition}\label{prop:parity}
Let $C(\mathcal S)$ be an $[[n,k,d]]$ CSS code and let $\{\bar X_j,\bar Z_j\}_{j=1}^{k}$ be a symplectic basis of logical operators in which every $\bar X_j$ is of $X$ type and every $\bar Z_j$ is of $Z$ type; such a basis always exists for CSS codes. Then $\boldsymbol{X}=X^{\otimes n}\in N(\mathcal S)$ if and only if every $Z$-type stabilizer generator has even weight, and in that case
\begin{equation}\label{eq:parity-X}
    \boldsymbol{X}\;\simeq\;\prod_{j\,:\,w(\bar Z_j)\ \mathrm{odd}}\bar X_j
    \qquad\text{on the code space,}
\end{equation}
where $w(\cdot)$ denotes the Hamming weight and $\simeq$ denotes equality up to a phase and an element of $\mathcal S$. In particular $\boldsymbol{X}\in\mathcal S$ if and only if all $\bar Z_j$ have even weight. The same statement holds with the roles of $X$ and $Z$ interchanged.
\end{proposition}
\begin{proof}
An $X$-type operator commutes with a $Z$-type operator if and only if their supports overlap on an even number of qubits, and the support of $\boldsymbol{X}$ is the full set of qubits, so $\boldsymbol{X}$ commutes with a $Z$-type generator $g$ if and only if $w(g)$ is even. Being of $X$ type, $\boldsymbol{X}$ commutes with all $X$-type generators and with all $\bar X_j$ automatically. Hence $\boldsymbol{X}\in N(\mathcal S)$ if and only if all $Z$-type generators have even weight. If $\boldsymbol{X}\in N(\mathcal S)$, its class in $N(\mathcal S)/\mathcal S$ is determined by its commutation relations with the symplectic basis: the symplectic form is nondegenerate on $N(\mathcal S)/\mathcal S$, so every element of $N(\mathcal S)$ is, up to a phase and an element of $\mathcal S$, a unique product $\prod_{j=1}^{k}\bar X_j^{a_j}\bar Z_j^{b_j}$ with $a_j,b_j\in\{0,1\}$, and the exponents are read off from the commutation relations, $b_j=1$ if and only if the operator anticommutes with $\bar X_j$ and $a_j=1$ if and only if it anticommutes with $\bar Z_j$ \cite{QEC1,NielsenChuang2000}. In particular $\boldsymbol{X}$ acts trivially on the code space, that is $\boldsymbol{X}\in\mathcal S$ up to a phase, if and only if it commutes with every $\bar Z_j$. Now $\boldsymbol{X}$ commutes with every $\bar X_j$, so $b_j=0$ for all $j$, and it anticommutes with $\bar Z_j$ exactly when $w(\bar Z_j)$ is odd, so $a_j=1$ exactly for those $j$. This is the right-hand side of \eqref{eq:parity-X}. Since all $Z$-type stabilizers have even weight, the parity of $w(\bar Z_j)$ does not depend on the chosen representative.
\end{proof}
When no $\bar Z_j$ has odd weight the product in Eq.~\eqref{eq:parity-X} is empty and $\boldsymbol{X}\simeq I$, that is, $\boldsymbol{X}\in\mathcal S$ and the $X$ mask is in class (ii); the rotated surface code of even distance is an example. Otherwise $\boldsymbol{X}$ is a nontrivial logical operator and the mask is in class (iii).

The labels $\bar X_j$ and $\bar Z_j$ in Proposition \ref{prop:parity} refer to the pure-type basis, which need not coincide with the logical basis in which a code is customarily written. For the Shor code in the basis of Section \ref{sec:shor}, the $Z$-type logical operator is $\bar X$ and the $X$-type one is $\bar Z$, so that the proposition yields $\boldsymbol{X}\simeq\bar Z$ and $\boldsymbol{Z}\simeq\bar X$, in agreement with the direct computation of that section.

Proposition~\ref{prop:parity} refines Corollary~\ref{thm:css-compatibility}: condition \eqref{eq:css-cond-X} places $\boldsymbol X$ in the normalizer, whereas the stronger statement $e\in C_2$ would place it in the stabilizer and hence in class (ii).

For $k=1$ there is a parity rule that does not even require the logical operators. If all stabilizer generators have even weight and $n$ is odd, then $\boldsymbol X$ and $\boldsymbol Z$ belong to $N(\mathcal S)$ but cannot belong to $\mathcal S$, because a product of even-weight operators of the same type has even weight whereas $\boldsymbol X$ and $\boldsymbol Z$ have odd weight $n$; with a single logical qubit they must then act as $\bar X$ and $\bar Z$, and both masks are in class (iii). This single observation accounts for the Steane code, for the two-dimensional color codes on a triangular patch and for the $[[15,1,3]]$ code in Table~\ref{tab:classification}, where it is referred to as the parity rule, and it already shows that the class of a code can change under a change of presentation or of lattice size.

\begin{table*}[t]
\caption{Classification of stabilizer codes with respect to the transversal masks $\boldsymbol{X}=X^{\otimes n}$ and $\boldsymbol{Z}=Z^{\otimes n}$, according to the three classes (i)--(iii) of Section \ref{sec:classification}. The column ``Action'' lists the logical operators implemented by $\boldsymbol X$ and by $\boldsymbol Z$, in this order, in class (iii); in class (ii) the action is trivial. For the rectangular rotated code, $d_1=|\bar X|$ and $d_2=|\bar Z|$ are the two side lengths of the $d_1\times d_2$ patch, and the tested sizes are $3\times2$, $5\times2$ and $3\times4$. The column ``Derivation'' points to the argument that establishes each row: Proposition~\ref{prop:parity} applied to the stabilizer weights and to the weights of the pure-type logical operators, the explicit computations of the text and of the appendices, or a direct inspection of the generators. Every row was reproduced by the exact symplectic routine over $\mathbb F_2$ described in the text.}
\label{tab:classification}
\footnotesize
\begin{ruledtabular}
\begin{tabular}{llcccl}
Code & $[[n,k,d]]$ & Class of $\boldsymbol{X}$ & Class of $\boldsymbol{Z}$ & Action & Derivation \\
\hline
Bit-flip repetition code & $[[3,1,1]]$ & (iii) & (iii) & $\bar X$, $\bar Z$ & Sec.~\ref{sec3}, App.~\ref{A2} \\
Shor code \cite{Shor} & $[[9,1,3]]$ & (iii) & (iii) & $\bar Z$, $\bar X$ & Eq.~\eqref{eq:shor-swap}, Lemma~\ref{lem:shor} \\
Steane code \cite{CSS1,CSS2} & $[[7,1,3]]$ & (iii) & (iii) & $\bar X$, $\bar Z$ & App.~\ref{A3}; parity rule \\
Five-qubit code \cite{fivequbit} & $[[5,1,3]]$ & (iii) & (iii) & $\bar X$, $\bar Z$ & weight-4 generators \\
Triangular color code, $d=3,5,7,9$ \cite{colorcode} & $[[(3d^2+1)/4,1,d]]$ & (iii) & (iii) & $\bar X$, $\bar Z$ & parity rule, $k=1$, $n$ odd \\
Toric color code 6.6.6, $n=24,36,72$ \cite{colorcode} & $[[n,4,d]]$ & (ii) & (ii) & trivial & faces of one color \\
Tetrahedral 3D color code \cite{BombinMartinDelgado2007PRL,steane1996quantumreedmullercodes} & $[[15,1,3]]$ & (iii) & (iii) & $\bar X$, $\bar Z$ & App.~\ref{Aqrm}; parity rule \\
Rotated surface code, $d=3,5$ \cite{BombinMartinDelgado2007PRA,Horsman2012} & $[[d^2,1,d]]$ & (iii) & (iii) & $\bar X$, $\bar Z$ & Prop.~\ref{prop:parity} \\
Rotated surface code, $d=2,4$ \cite{BombinMartinDelgado2007PRA,Horsman2012} & $[[d^2,1,d]]$ & (ii) & (ii) & trivial & Prop.~\ref{prop:parity} \\
Rotated surface code, $d_1\times d_2$, $d_1$ odd, $d_2$ even & $[[d_1d_2,1,\min(d_1,d_2)]]$ & (ii) & (iii) & trivial, $\bar Z$ & Prop.~\ref{prop:parity} \\
Planar surface code, $L=2,\dots,5$ \cite{toriccode1} & $[[L^2+(L-1)^2,1,L]]$ & (i) & (i) & --- & weight-3 generators \\
Toric code, $L=3$ \cite{toriccode2} & $[[2L^2,2,L]]$ & (iii) & (iii) & $\bar X_1\bar X_2$, $\bar Z_1\bar Z_2$ & Prop.~\ref{prop:parity} \\
Toric code, $L=2,4$ \cite{toriccode2} & $[[2L^2,2,L]]$ & (ii) & (ii) & trivial & Prop.~\ref{prop:parity} \\
$[[4,2,2]]$ code \cite{Vaidman1996,Grassl1997} & $[[4,2,2]]$ & (ii) & (ii) & trivial & generators \\
$[[8,3,2]]$ code \cite{Campbell2016,VasmerKubica2022} & $[[8,3,2]]$ & (ii) & (ii) & trivial & generators \\
\end{tabular}
\end{ruledtabular}
\end{table*}

To test the classification on concrete codes we implemented the symplectic criterion as a short exact routine over $\mathbb F_2$. Each code is specified by its stabilizer generators as rows of a binary matrix $(H_X|H_Z)$ over $\mathbb F_2$. The routine computes the symplectic products of $\boldsymbol{X}$ and $\boldsymbol{Z}$ with all generators, which decides class (i); it then tests membership of the mask in the row space of the generator matrix by Gaussian elimination over $\mathbb F_2$, which separates classes (ii) and (iii); finally, for class (iii) it computes the symplectic products with a symplectic basis of logical operators, which identifies the logical action. For CSS codes the logical operators are obtained from the kernels of $H_X$ and $H_Z$ modulo the corresponding row spaces, paired into a symplectic basis. The surface and toric codes were generated directly from their lattice definitions, the triangular color codes from the face structure of the 6.6.6 tiling as produced by the code generator of Stim \cite{Gidney2021}, the toric color codes from a brick-wall presentation of the same tiling with periodic boundaries, and the $[[15,1,3]]$ code from its description as a punctured quantum Reed--Muller code \cite{steane1996quantumreedmullercodes}. The results are collected in Table \ref{tab:classification}.

Three features of Table~\ref{tab:classification} deserve comment. First, all three classes occur within the surface-code family, and the class is a property of the presentation rather than of the topological phase. The rotated surface code has even-weight boundary stabilizers and a logical $\bar Z$ of weight $d$, so by Proposition~\ref{prop:parity} it is in class (iii) for odd $d$ and in class (ii) for even $d$; on a rectangular patch with $w(\bar X)$ odd and $w(\bar Z)$ even the two masks fall into different classes, $\boldsymbol X\in\mathcal S$ while $\boldsymbol Z\simeq\bar Z$, so that only one key bit encrypts. The planar code with smooth and rough boundaries has weight-three boundary stabilizers, so both masks anticommute with $2(L-1)$ generators and the code is in class (i) for every $L$. The toric code has only even-weight stabilizers and two logical qubits with $w(\bar Z_j)=L$, so for odd $L$ the masks act as $\bar X_1\bar X_2$ and $\bar Z_1\bar Z_2$: the two key bits then encrypt a single combined degree of freedom and leave the relative logical information unprotected. The color codes, which are CSS codes, behave differently in one respect: class (i) never occurs, because every face of a trivalent three-colorable lattice has an even number of vertices and every cell of a three-dimensional color code has an even number of vertices, so all stabilizer generators have even weight and both masks lie in the normalizer by Proposition~\ref{prop:parity}. The class is again fixed by the boundary presentation. On a closed surface the faces of any one color partition the qubits, so the product of the $X$-type (or $Z$-type) generators of that color is $\boldsymbol X$ (or $\boldsymbol Z$) and the code is in class (ii), as the toric 6.6.6 code of Table~\ref{tab:classification} illustrates. On a triangular patch with three colored boundaries the code has $k=1$ and odd $n$, so it is in class (iii) with $\boldsymbol X\simeq\bar X$ and $\boldsymbol Z\simeq\bar Z$. Second, the two-dimensional color codes and the $[[15,1,3]]$ code are in class (iii) for all distances tested, with the masks acting as $\bar X$ and $\bar Z$, which is why these codes support both the encrypted storage of this section and the encrypted non-Clifford computation of Sec.~\ref{sec4} without modification of the key-update rules. Third, the Shor code is in class (iii) with the logical action interchanged, $\boldsymbol X\simeq\bar Z$ and $\boldsymbol Z\simeq\bar X$, as found in Eq.~\eqref{eq:shor-swap} and used in Appendix~\ref{A4}. The classification therefore not only decides compatibility but also fixes the code-dependent map between the physical key $(a,b)$ and the logical Pauli frame that the client must track. For the client the practical consequence is that the class of the server's code presentation, and not only the code family, has to be specified in the protocol, on the same footing as the block size $n$: class (i) presentations are excluded, class (ii) presentations offer no confidentiality for the affected key bit and should be replaced by a class (iii) presentation of the same code, which for the surface and toric codes amounts to choosing an odd lattice size and for the color codes to using a patch with boundaries rather than a closed surface, and class (iii) presentations are the ones to which the protocols of this paper apply.

\subsection{Logical-Pauli masking and noise}
\label{sec:discussion}
Proposition~\ref{thm:stabilizer-compatibility} and Corollary~\ref{thm:css-compatibility} concern the fixed physical operators $X^{\otimes n}$ and $Z^{\otimes n}$. A more general but weaker statement holds if the client masks with logical Pauli representatives instead. For one logical qubit choose $\bar X,\bar Z\in N(\mathcal S)$ with
\begin{eqnarray}
&&\bar Z\ket{\bar 0}=\ket{\bar 0},\qquad \bar Z\ket{\bar 1}=-\ket{\bar 1},\\
&&\bar X\ket{\bar 0}=\ket{\bar 1},\qquad \bar X\ket{\bar 1}=\ket{\bar 0}.
\end{eqnarray}
Every such representative commutes with the stabilizer and maps the code space to itself, so
\begin{equation}
\bar X^a\bar Z^b\,(c_0\ket{\bar 0}+c_1\ket{\bar 1})=c_0'\ket{\bar 0}+c_1'\ket{\bar 1}
\end{equation}
is a codeword for every key, whatever the code. This is a statement about encoded states and logical operators, not a claim that the corresponding circuits are transversal, fault tolerant or cheap; for $k>1$ the mask $\prod_j\bar X_j^{a_j}\bar Z_j^{b_j}$ is treated in the same way. We nevertheless focus on the transversal physical mask for two reasons. Single-qubit Pauli gates are native on every hardware platform and require no compilation, whereas general logical operators must be synthesized code by code. And Proposition~\ref{thm:stabilizer-compatibility} and Corollary~\ref{thm:css-compatibility} give an explicit, checkable criterion for the physical mask, which for class (iii) codes coincides with a logical-Pauli mask anyway. Because the mask acts on the physical qubits, the client needs to know only the block size $n$ and the class of the code presentation, not the code itself.

\subsection{Effect of noise on compatible encryption}
The compatibility results also determine how encryption interacts with hardware noise, which is a question distinct from privacy. Consider a Pauli noise channel acting on the stored block, that is, a random Pauli error $E$ drawn from a key-independent distribution. Since $U_{\rm enc}(a,b)$ is itself a Pauli operator, $E\,U_{\rm enc}(a,b)=\pm U_{\rm enc}(a,b)\,E$, and when Eq.~\eqref{eq:stab-comm} holds $U_{\rm enc}(a,b)$ commutes with every generator, so the syndrome of $E\,U_{\rm enc}(a,b)\ket{\bar\psi}$ equals the syndrome of $E\ket{\bar\psi}$ and the recovery chosen by the decoder is the same. The decoder either succeeds or fails on the encrypted codeword exactly when it does so on the plaintext codeword, and the residual logical error, which is a Pauli operator, commutes with the mask up to a phase. Under Pauli noise, then, the logical error rate of the encrypted memory equals that of the plaintext memory, and encryption and decryption contribute only two extra transversal Pauli layers. Each layer applies $X$ to every physical qubit when $a=1$ and $Z$ when $b=1$, that is, $a+b$ noisy single-qubit gates per physical qubit and layer, $2(a+b)$ in total and at most four. For a distance-three code with physical error rate $p$ and single-qubit gate error rate $p_g$ this raises the effective error rate entering the first-order logical failure probability from $p$ to roughly $p+2(a+b)p_g$, so the ratio of encrypted to plaintext logical failure is $(1+2(a+b)p_g/p)^2$ at leading order, which is close to one whenever $p_g\ll p$. Since this small overhead depends on the key through $a+b$, an implementation that wishes to keep the error rate key-independent should apply both Pauli layers for every key, inserting identity gates of the same duration when a key bit is zero, at the cost of the maximal overhead $p\to p+4p_g$. The argument uses only two facts: the mask is a Pauli operator in the normalizer, and the noise is a Pauli channel. The first fact holds in every compatible presentation and makes the syndrome extraction and the recovery key-independent for arbitrary noise. The second is what makes the logical error rate key-independent: for a general noise channel $\mathcal E$, the encrypted run coincides with a plaintext run subject to the conjugated channel $\mathcal E_P(\rho)=P\,\mathcal E(P\rho P)\,P$ with $P=U_{\rm enc}(a,b)$, which equals $\mathcal E$ whenever $\mathcal E$ is a Pauli channel. The Pauli noise model is the standard one in the analysis of stabilizer-code memories and is the one adopted throughout this paper.

\section{Logical non-Clifford gates on encrypted encoded states}\label{sec4}
Storage is usually not enough: a cloud server must also compute on the encrypted encoded data. For Clifford gates this is unproblematic, because Clifford gates normalize the Pauli group and their effect is tracked by the key-update rules of Sec.~\ref{sec2}, both at the physical and, for codes with transversal Clifford gates, at the logical level. The difficulty is concentrated in the non-Clifford gate $T$, which together with the Clifford gates forms a universal set, typically dominates the fault-tolerance overhead, and in the homomorphic setting produces the key-dependent byproduct of Eq.~\eqref{eq:Tbyproduct}. On an encoded block the server can act only through physical operations, for instance a transversal $T^{\otimes n}$ or some other physical representative of the logical $\bar T$. The question of this section is whether such a physical operation, applied to the encrypted encoded block and followed by the homomorphic corrections, key updates, and decryption, leaves the client with the same logical state, $T\ket{\bar\psi}=c_0\ket{\bar 0}+e^{i\pi/4}c_1\ket{\bar 1}$, as if $\bar T$ had been applied to the plaintext logical state $\ket{\bar\psi}=c_0\ket{\bar 0}+c_1\ket{\bar 1}$. We first show that the answer is negative for the Shor code with a transversal $T^{\otimes 9}$, then positive for suitable triorthogonal codes, and finally we formulate the code-independent template based on logical-Pauli masking.

\subsection{Transversal encryption and transversal $T$ gates}
The purpose of the following computation is to show that code-space compatibility of the mask does not imply that a transversal non-Clifford gate implements the intended logical gate. In this regard, we focus on transversal application of $T$-gates. Consider $\ket{\bar\psi}=c_0\ket{\bar 0}+c_1\ket{\bar 1}$ in the Shor basis of Eqs.~\eqref{myeq8}--\eqref{myeq9} and let $\omega=e^{i\pi/4}$. Since $T\ket 0=\ket 0$ and $T\ket 1=\omega\ket 1$, on each three-qubit block
\begin{equation}
T^{\otimes 3}(\ket{000}\pm\ket{111})=\ket{000}\pm\omega^3\ket{111},
\end{equation}
so that $T^{\otimes 9}\ket{\bar\psi}=\tfrac{1}{2\sqrt 2}[c_0(\ket{000}+\omega^3\ket{111})^{\otimes 3}+c_1(\ket{000}-\omega^3\ket{111})^{\otimes 3}]$, whereas the target state is $\bar{T}\ket{\bar\psi}=c_0\ket{\bar 0}+\omega c_1\ket{\bar 1}=\tfrac{1}{2\sqrt 2}[c_0(\ket{000}+\ket{111})^{\otimes 3}+\omega c_1(\ket{000}-\ket{111})^{\otimes 3}]$. The transversal gate distributes a phase $\omega^3$ inside each block instead of a single phase $\omega$ on the logical $\ket{\bar 1}$ amplitude, and since $\omega^3\neq\pm 1$ each transformed block is a nontrivial superposition of the two block states $\ket{000}\pm\ket{111}$. The state $T^{\otimes 9}\ket{\bar\psi}$ therefore has components outside the Shor code space: $T^{\otimes 9}$ neither preserves the code space nor implements $\bar T$, and no subsequent key update can repair this, because the failure occurs already on the plaintext codeword.

Triorthogonal codes \cite{PhysRevA.86.052329,10735332,shi2024triorthogonalcodesselfdualcodes} are the CSS codes designed to avoid this obstruction. An $m\times n$ binary matrix $G$ is triorthogonal if any two distinct rows have even overlap, $\sum_j G_{ij}G_{kj}\equiv 0\pmod 2$, and any three distinct rows have even triple overlap, $\sum_j G_{ij}G_{kj}G_{\ell j}\equiv 0\pmod 2$. $G$ can be ordered as $G=\binom{G_1}{G_0}$, where $G_1$ contains the $k$ odd-weight rows and $G_0$ the even-weight rows. The rows of $G_0$ generate the $X$-type stabilizers whereas the rows $f_j$ of $G_1$ give logical representatives $\bar X_j=X(f_j)$. The reference codeword $\ket{G_0}$ is the uniform superposition over the row space of $G_0$, and the logical basis can be constructed as $\ket{\bar x}=\prod_j\bar X_j^{x_j}\ket{G_0}$. Pairwise orthogonality provides the CSS commutation relations, and triple orthogonality is what makes the transversal $T$ gate act as a logical non-Clifford gate: the triorthogonal family admits a code-dependent Clifford unitary $U$ such that, for every logical basis state,

\begin{equation}
U T^{\otimes n} \prod_{j=1}^{k} \bar{X}_{j}^{x_j} |G_0\rangle
= e^{i\pi/4 \sum_{j=1}^{k} x_j}
\prod_{j=1}^{k} \bar{X}_{j}^{x_j}  |G_0\rangle .\label{new00}
\end{equation}

Explicit members of the family exist that also satisfy Proposition~\ref{thm:stabilizer-compatibility}. The $[[15,1,3]]$ punctured quantum Reed--Muller code \cite{steane1996quantumreedmullercodes}, which is the minimal tetrahedral three-dimensional color code, is verified in Appendix~\ref{X}, where the same argument is extended to the tetrahedral family built on punctured $3$-colexes \cite{BombinMartinDelgado2007PRL,BombinMartinDelgado2007PRB}; Table~\ref{tab:classification} confirms that the $[[15,1,3]]$ code is in class (iii) with $\boldsymbol X\simeq\bar X$ and $\boldsymbol Z\simeq\bar Z$.

For such a code the homomorphic implementation of $\bar T$ on one encrypted encoded qubit proceeds as follows. The client encodes the logical state, $\ket{{\psi}}$, into the $n$-qubit physical state, $\ket{\bar{\psi}}$, prepares $n$ Bell pairs, one per physical qubit of the block, draws $(a,b)$, and sends the server $\ket{\bar\psi^{\rm enc}}=(X^aZ^b)^{\otimes n}\ket{\bar\psi}$ together with one qubit of each Bell pair; by Proposition~\ref{thm:stabilizer-compatibility} the encrypted state is a codeword, and since the code is in class (iii) the mask is a logical Pauli mask of the encoded qubit. The server applies $T^{\otimes n}$, which by Eq.~\eqref{eq:Tbyproduct} gives
\begin{equation}
T^{\otimes n}\ket{\bar\psi^{\rm enc}}=\bigl((S^\dagger)^aX^aZ^{a\oplus b}T\bigr)^{\otimes n}\ket{\bar\psi}.
\end{equation}
Then it swaps the block with the $n$ Bell-pair qubits to defer the removal of the $S$-type byproducts to the client, applies the Clifford correction $U$ of Eq.~\eqref{new00}, and returns all qubits together with the key-update functions of $U$. The client performs the $n$ rotated Bell measurements that remove the byproducts, composes the outcomes with the key-update functions of $U$ into final $n$ pair of keys $\{(a_i^f,b_i^f)\}_i$, and applies $\bigotimes_i Z^{b_i^f}X^{a_i^f}$ on the physical qubit encoding the single logical qubit, which leaves $UT^{\otimes n}\ket{\bar\psi}=c_0\ket{\bar 0}+e^{i\pi/4}c_1\ket{\bar 1}=\bar T\ket{\bar\psi}$ by Eq.~\eqref{new00}. The same steps with $S$ in place of $S^\dagger$ implement $\bar T^\dagger$. The resources are one encoded block of $n$ qubits, $n$ Bell pairs, $n$ physical $T$ gates and $O(n)$ Pauli and teleportation operations.

\subsection{Logical-Pauli masking: an algebraic template}
Transversal physical masking and transversal $T$ gates succeed only for code families with the structure just described. A code-independent formulation is obtained by replacing the physical mask with logical Pauli representatives $\bar X,\bar Z$  and the physical gates with representatives of the logical gates realized on the encoded block. Since $\bar X$ and $\bar Z$ preserve the code space by definition, the encrypted state
\begin{equation}
\ket{\bar\psi_{\rm enc}}=\bar X^{a}\bar Z^{b}\ket{\bar\psi}
\end{equation}
is a codeword for any stabilizer code, and the four possible ciphertexts are the logical states $c_0\ket{\bar 0}\pm c_1\ket{\bar 1}$ and $c_0\ket{\bar 1}\pm c_1\ket{\bar 0}$. This two-bit logical mask randomizes the encoded qubit at the logical level and must not be confused with an independent physical one-time pad on the $n$ qubits of the block.

To evaluate $\bar T$ homomorphically on the physical qubits encrypted with logical Pauli masking, the client prepares an encoded Bell pair $\ket{\bar\Phi}=(\ket{\bar 0}\ket{\bar 0}+\ket{\bar 1}\ket{\bar 1})/\sqrt 2$ and sends the server $\ket{\bar\psi_{\rm enc}}$ and one half of $\ket{\bar\Phi}$. The server needs physical representatives of $\bar S$ and $\bar T$ defined by $\bar S\ket{\bar 1}=i\ket{\bar 1}$ and $\bar T\ket{\bar 1}=e^{i\pi/4}\ket{\bar 1}$, with $\ket{\bar 0}$ invariant; on the code space the set $\{\bar S,\bar T,\bar X,\bar Z\}$ then obeys the same algebra as $\{S,T,X,Z\}$, and in particular the encoded version of Eq.~\eqref{eq:Tbyproduct}. The server applies $\bar T$ to the data block, swaps it with the encoded Bell half, and returns both blocks; the client measures its two logical qubits in the $\bar S^{\,a}$-rotated logical Bell basis, which is implemented by encoded operations on the underlying blocks, and applies the encoded Pauli correction dictated by the outcome, obtaining $\bar T\ket{\bar\psi}$. Appendix~\ref{A4} carries out this protocol for the Shor code, including an explicit unitary extension of $\bar S$ and $\bar T$ to the full nine-qubit Hilbert space.

The template must be read with two qualifications. On the code space the operators $\bar S$ and $\bar T$ are specified, but as physical gates they must be extended to the orthogonal complement of the code space; such extensions exist but are not unique, and their implementation cost is code-dependent, typically requiring magic-state injection, gauge fixing, code switching or non-transversal synthesis. And the construction is an algebraic compatibility statement, not a fault-tolerant fully homomorphic encryption scheme with a composable security proof. In terms of resources, both approaches use $3n$ physical qubits for one logical $\bar T$ on an $[[n,1,d]]$ code, one data block plus either $n$ physical Bell pairs or one encoded Bell pair of $2n$ qubits. The difference lies in the operational cost: $O(n)$ elementary gates for the transversal approach, versus $C(\bar X^a\bar Z^b)+C(\bar T)+C(\bar S^{\,a}\text{-rotated Bell measurement})+C(\bar Z^{b'}\bar X^{a'} )$ (for encryption, evaluation of $T$, application of measurement to remove byproduct, and decryption) for the logical approach, where $C(\bar U)$ denotes the physical cost of the encoded operator $\bar U$ and reduces to $O(n)$ only when the relevant logical gates are transversal. Transversal physical masking is therefore operationally cheaper whenever the code is in class (iii) and supports a transversal $T$-type logical gate, as tetrahedral color codes do \cite{colorcode,BombinMartinDelgado2007PRL,BombinMartinDelgado2007PRB,KubicaBeverland2015}, while logical-Pauli masking is universal in scope at a code-dependent price.

\section{Conclusion}
\label{sec5}
This work asks under what algebraic conditions quantum homomorphic encryption with a Pauli mask and stabilizer quantum error correction can coexist. Two results are established. First, code-space compatibility: the transversal block mask $(X^aZ^b)^{\otimes n}$ preserves the code space of an $[[n,1,d]]$ stabilizer code if and only if $X^{\otimes n}$ and $Z^{\otimes n}$ commute with every stabilizer generator, which for a CSS code built from $C_2\subset C_1$ reduces to $e\in C_1$ and $e\in C_2^\perp$ for the all-ones word $e$. Second, key independence of error correction: whenever the mask preserves the code space, syndrome extraction and recovery act identically on the encrypted and on the plaintext codeword, so the server corrects errors without learning the key and, under Pauli noise, without changing the logical error rate beyond the two extra Pauli layers of encryption and decryption. The classification of Table~\ref{tab:classification} adds a third element that a client needs in practice: the mask may leave the code space, act trivially, or act as a nontrivial logical Pauli operator, and the class depends on the presentation of the code, so that surface and toric codes must be used with odd lattice sizes, color codes with a bounded patch rather than a closed surface, planar codes with smooth and rough boundaries are excluded, and the Shor code requires an interchange of the physical and logical key labels.

For encrypted computation, the Shor code shows that code-space compatibility of the mask does not carry over to a transversal $T^{\otimes n}$. In contrast, triorthogonal codes, and among them the tetrahedral three-dimensional color codes, support both the transversal mask and a transversal $T$-type logical gate up to a Clifford correction. This enables homomorphic operation of the $T$ gate via transversal application of physical $T$ gates on the encrypted physical qubits encoding the logical state. Logical-Pauli masking removes the dependence on the code family at the level of code-space preservation, at the price of code-dependent logical gates whose physical realization may be non-transversal and resource intensive. The two approaches thus trade operational simplicity against universality.

We state explicitly what is not established here. We do not prove fault tolerance of the resulting circuits, which requires specifying syndrome-extraction circuits and tracking the encrypted Pauli frame through them, nor cryptographic security against a malicious server, which requires a composable analysis of the restricted logical mask; both are left as external assumptions or future work. Within these limits, the compatibility criteria and the classification give concrete guidance on which codes and which presentations can host encrypted encoded data, and the analytic noise argument suggests that the price of encryption in a Pauli-noise memory is small. Natural next steps are a formal security analysis of the restricted mask in the honest-but-curious and malicious models and an experimental demonstration of encrypted storage on a class (iii) code such as the distance-three color code on current hardware.

\acknowledgements
We thank P. Fern\'andez Ortiz for useful discussions during the early stages of this work. The authors acknowledge support from Spanish MICIN grant PID2021-122547NB-I00 and the ``MADQuantum-CM'' project funded by Comunidad de Madrid (Programa de acciones complementarias) and by the Ministry for Digital Transformation and of Civil Service of the Spanish Government through the QUANTUM ENIA project call --Quantum Spain project, and by the European Union through the Recovery, Transformation and Resilience Plan Next Generation EU within the framework of the Digital Spain 2026 Agenda, the CAM Programa TEC-2024/COM-84 QUITEMAD-CM. This work has been financially supported by  the project MADQuantum-CM, funded by Comunidad de Madrid (Programa de Acciones Complementarias), and by the Recovery, Transformation and Resilience Plan---funded by the European Union---(Next Generation EU, PRTR-C17.I1). KS acknowledges support from the project MADQuantum-CM, funded by Comunidad de Madrid (Programa de Acciones Complementarias) and by the Recovery, Transformation, and Resilience Plan---Funded by the EuropeanUnion---(NextGeneration EU, PRTR-C17.I1). M.A. M.-D. has been partially supported by the U.S. Army Research Office through Grant No. W911NF-14-1-0103, and Madrid ELLIS Unit CAM.

\appendix
\section{The homomorphic encryption scheme in detail}\label{A1}
This appendix collects the elements of the QHE scheme of Ref.~\cite{Liang30} that Sec.~\ref{sec2} uses: the gate set, the gate-teleportation identity, the key-update rules and a two-qubit worked example.

\subsection{Universal gate set}
Quantum algorithms are implemented via quantum circuits constructed from a universal gate set, which must contain at least one entangling gate and one non-Clifford gate.
The Clifford gates that are usually included in traditional universal gate sets comprise the Pauli gates
\begin{eqnarray}  X=\left(\begin{matrix}
0&1\\
1&0
\end{matrix}\right)\text{, } Z=\left(\begin{matrix}
1&0\\
0&-1
\end{matrix}\right),
\end{eqnarray}
the shift and Hadamard gates \begin{eqnarray}
  S=\left(\begin{matrix}
1&0\\
0&i
\end{matrix}\right)\text{, }
H=\frac{1}{\sqrt{2}}\left(\begin{matrix}
1&1\\
1&-1
\end{matrix}\right),
\end{eqnarray} and the entangling gate
\begin{eqnarray}
CNOT=\left(\begin{matrix}
1&0&0&0\\
0&1&0&0\\
0&0&0&1\\
0&0&1&0
\end{matrix}\right).  \end{eqnarray}
The non-Clifford gate, which is commonly added to the set to make the universal set of gates complete, is
\begin{eqnarray}
T=\left(\begin{matrix}
1&0\\
0&e^{i\pi/4}
\end{matrix}\right).
\end{eqnarray}

\subsection{Gate teleportation}
\label{GT}
Using a maximally entangled state shared between a sender and receiver, an unknown quantum state can be teleported from the former to the latter without physically transporting the quantum system \cite{PhysRevLett.70.1895}. In particular, to perfectly teleport an unknown qubit state, say $\ket{\psi}$, the sender and receiver must share a Bell state, say $\ket{\Phi}=\frac{1}{\sqrt{2}}(\ket{00}+\ket{11})$. Label the qubit to be teleported by $1$, the sender's half of the Bell pair by $2$ and the receiver's half by $3$. In that case, the joint initial state of the three qubits can be expressed as
\begin{equation} \ket{\psi}_1\otimes\ket{\Phi}_{23}=\frac{1}{2}\sum_{a,b}\ket{\Phi_{ab}}_{12}\otimes X^aZ^b\ket{\psi}_3,
\end{equation}
where $\ket{\Phi_{ab}}=Z^bX^a\otimes I\ket{\Phi}$, $a\in \{0,1\}$, and $b\in \{0,1\}$. Here $I$ denotes the two-dimensional identity operator. The set of states $\{\ket{\Phi_{ab}}\}_{a,b}$ forms a basis of the two-qubit Hilbert space. The sender performs a projective measurement in the Bell basis $\{\ket{\Phi_{ab}}\}_{a,b}$ on the
qubit to be teleported and their half of the shared Bell state. This measurement has
four possible outcomes, indexed by the classical bits $a, b\in \{0, 1\}$. Crucially, the
measurement outcome determines the state of the receiver's qubit: it will be projected
to $X^aZ^b\ket{\psi}$ (a deterministic function of the measurement result). The sender then
communicates the classical outcome $(a, b)$ to the receiver via a classical channel. Upon
receiving this information, the receiver applies the correction unitary $Z^bX^a$ to their
qubit, recovering the original state $\ket{\psi}$. In this way, the unknown quantum state
$\ket{\psi}$ has been transferred from sender to receiver using only 2-bit classical communication
(the bits $a$ and $b$) and a pre-shared entangled resource.

The same measurement-based protocol can be adapted to teleport not just quantum states,
but quantum gates themselves---a variant called gate teleportation. This extension is
crucial for implementing non-Clifford operations on encrypted data. To understand gate
teleportation, consider the following setup: sender and receiver share a Bell state
$\ket{\Phi}$, and the sender additionally holds a qubit prepared in state $\ket{\psi}$ and wishes to
apply a unitary gate $U$ to that state. In the gate teleportation protocol, the measurement
is performed in a ``$U$-rotated Bell basis" $\{\ket{\Phi(U)_{ab}}\}_{a,b}$ rather than the standard Bell basis, $\{\ket{\Phi_{ab}}\}_{a,b}$.
With the same labeling of the qubits, the joint three-qubit state can be decomposed as
\begin{equation}   \ket{\psi}_1\otimes\ket{\Phi}_{23}=\frac{1}{2}\sum_{a,b}\ket{\Phi(U)_{ab}}_{12}\otimes X^aZ^bU\ket{\psi}_3,
\end{equation}
where $\ket{\Phi(U)_{ab}}=U^\dag\otimes I\ket{\Phi_{ab}}$. Clearly $\{\ket{\Phi(U)_{ab}}\}_{a,b}$ also forms a basis. In this case, the performance of the projective measurement with the $U$-rotated Bell basis, $\{\ket{\Phi(U)_{ab}}\}$, by the sender on the two qubits that belong to the sender will project the receiver's state to $X^aZ^bU\ket{\psi}$. After getting the classical information about the exact measurement outcomes, $a$ and $b$, the receiver can accordingly apply $Z^bX^a$ on the available qubit and get the state, $U\ket{\psi}$. A schematic diagram describing the gate teleportation process is presented in Figure \ref{fig1}.

The relevance of gate teleportation will be realized soon in the next part, where we describe the homomorphic encryption scheme in detail.
\begin{figure}[h]
    \centering
\includegraphics[width=0.45\textwidth]{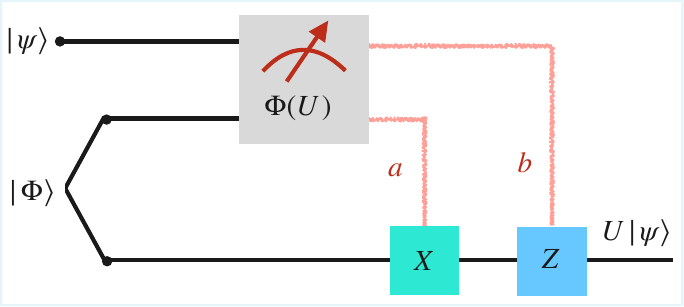}
    \caption{\textbf{Illustration of the process of gate teleportation.} The gray box represents the measurement performed on the two qubits that belong to the sender. Depending on the measurement outcome, $a$ and $b$, the receiver operates the gates $X^a$ (cyan box) and $Z^b$ (blue box) on the qubit that belongs to the receiver and gets the state, $U\ket{\psi}$.}
    \label{fig1}
\end{figure}

\subsection{Encryption, evaluation and decryption}
Consider a client, $C$, who wants to apply a quantum circuit on an $n$-qubit state $\ket\psi$ of qubits $\{w_i\}$ in the cloud controlled by a server, $S$. The circuit is decomposed into $n_u$ gates from $\mathcal U=\{X,Z,S,H,CNOT,T,T^\dagger\}$, of which $n_T$ are non-Clifford. The scheme operates in the following stages.

\begin{itemize}
    \item
\textit{Setup.---} To apply $n_T$ non-Clifford gates, the server and client must generate $n_T$ Bell states, $\ket{\Phi}_{s_ic_i}$. Here, $s_i$ and $c_i$ denote the qubits of the Bell state. We have used such notation because in the next step the client sends the qubits, $\{s_i\}$, to the server and keeps $\{c_i\}$.
\item
\textit{Encryption.---}
In the first step, the client encrypts the data and sends the encrypted quantum state to the server along with the qubits, $\{s_i\}$. To encrypt, since $\ket{\psi}$ is an $n$-qubit state, the client generates $n$ pairs of numbers, say $\{(a_j,b_j)\}_{j=1}^n$, where both $a_j$ and $b_j$ are randomly and independently chosen from the set $\{0,1\}$. Depending on this set of pairs of binary numbers, the client applies the Pauli gates, $X$ and $Z$, to the state of the $n$ qubits $\ket{\psi}$ and transforms it into $\ket{\psi^{\text{enc}}}=\bigotimes_j X^{a_j}Z^{b_j}\ket{\psi}$. The pairs $\{(a_j,b_j)\}$ are called secret keys, as this set is the key to the actual state $\ket{\psi}$. In particular, the original state $\ket{\psi}$ can always be recovered from the encrypted state $\ket{\psi^{\text{enc}}}$, if all secret keys are known by applying $X$ and $Z$ accordingly on $\ket{\psi^{\text{enc}}}$.
The encrypted state $\ket{\psi^{\text{enc}}}$ is finally sent to the server for cloud computing with the set of qubits $\{s_i\}$.
\item
\textit{Evaluation.---}
\begin{figure}[t]
    \centering
\includegraphics[width=0.43\textwidth]{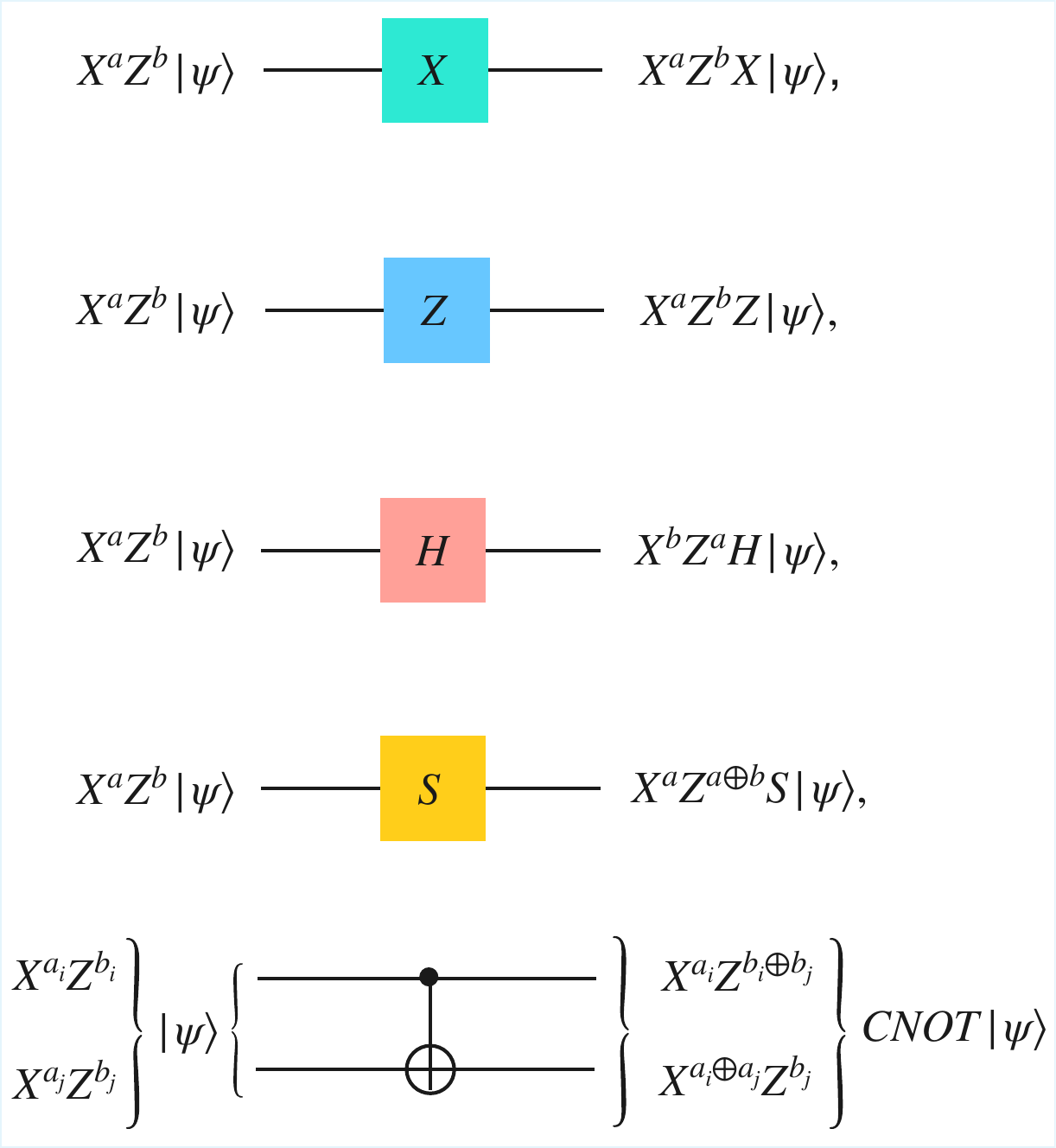}
    \caption{\textbf{Key updating rules for various Clifford gates.} The way in which the keys get updated on the action of (top to bottom in the figure) $X$, $Z$, $H$, $S$, and $CNOT$ gates is depicted in the above figure.}
    \label{fig2}
\end{figure}
\begin{table}[t]
    \centering \includegraphics[width=0.45\textwidth]{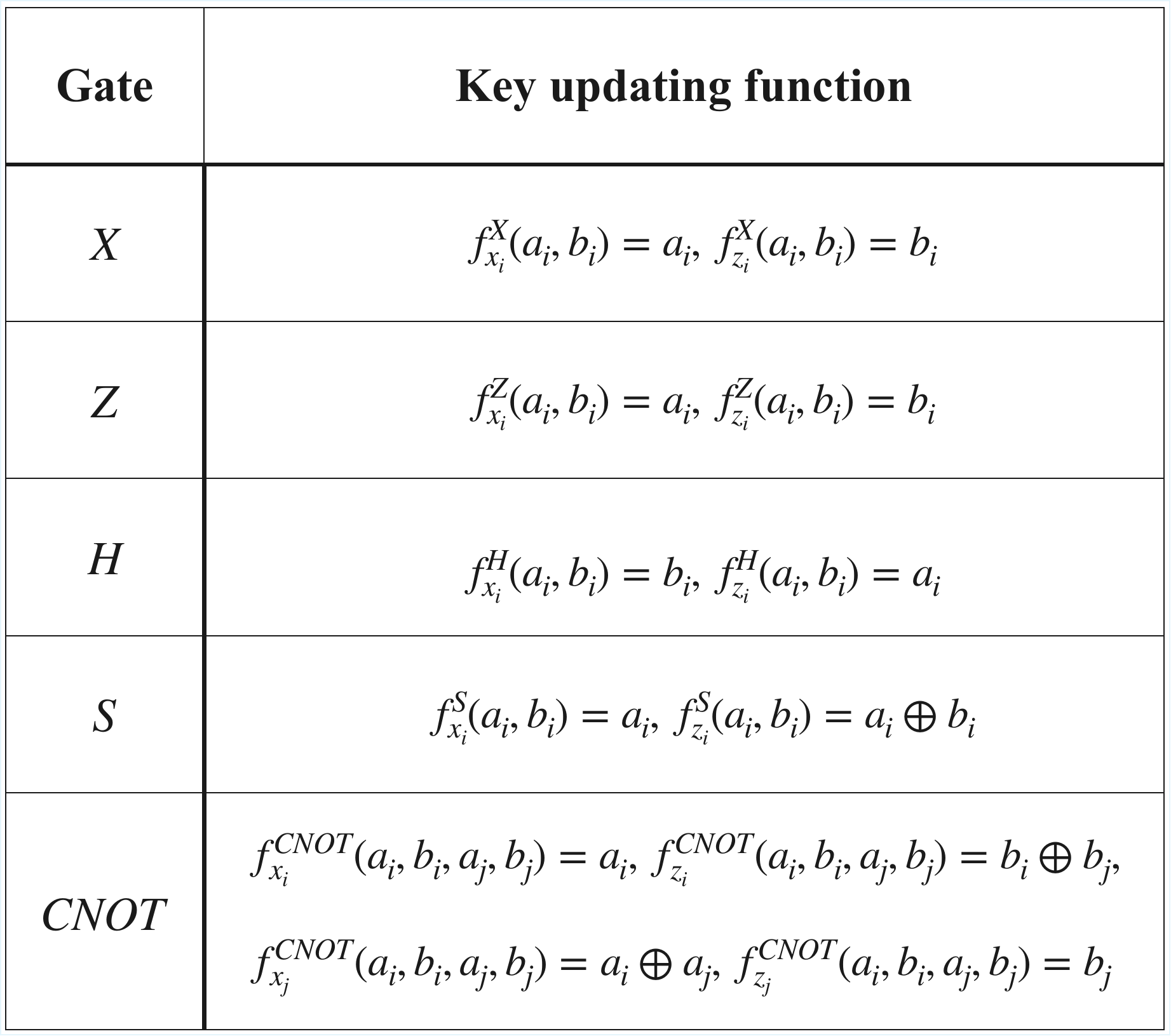}
    \caption{Key updating functions for different Clifford gates.}
    \label{table1}
\end{table}
After receiving the state $\ket{\psi^{\text{enc}}}$ from the client, the server needs to apply the gates in the required order for the application of the quantum circuit on the $n$-qubit system. However, since the state is encrypted, along with the application of the gates, the secret keys also need to be upgraded accordingly. Since for all single-qubit Clifford gates, $G$, we have $X^{a}Z^bG=GX^{f^{G}_x(a,b)}Z^{f^{G}_z(a,b)}$, for such gates, the situation is straightforward.
Specifically, when a single qubit Clifford gate $G$ is applied to the $i$th
qubit, the key pair $(a_i, b_i)$ is updated to a new pair, $\left(f^G_{x_i}(a_i, b_i),f^G_{z_i}(a_i, b_i)\right)$, through classical functions $f^G_{x_i}$ and
$f^G_{z_i}$, where each function takes the current bits $(a_i, b_i)$ as input and outputs single binary
numbers. In particular, when the server applies the gate $G$ on the $i$th qubit, it records the functions  $\{f_{x}^G,f_{z}^G\}_i$. These functions depend only on the gate $G$ and the current key values, making the
protocol practical and efficient. Due to similar reasons, when the server acts $CNOT$ gate to the $i$ and $j$th qubits, it notes the functions $\{f_x^{CNOT},f_z^{CNOT}\}_{i,j}$, where $f_x^{CNOT}$ and $f_z^{CNOT}$ act on $\{(a_i,b_i),(a_j,b_j)\}$ and provide two new pairs of keys, which are $(a_i, b_i\oplus b_j)$ and $(a_i\oplus a_j, b_j)$, for the $i$th and $j$th qubits, respectively. Figure \ref{fig2} and Table \ref{table1} present the complete key updating rules for all Clifford gates
in the universal gate set, $\mathcal{U}$. For each gate $G$, the table specifies the classical functions
$f^G_{x_i}$ and $f^G_{z_i}$ that determine the updated key pair for $i$th qubit when $G$ is applied.
These rules enable the server to evaluate Clifford circuits on encrypted data while
maintaining the homomorphic property.

In contrast, non-Clifford gates cannot be handled by such
simple key updates.
When a $T$ gate is applied on the $i$th qubit, it follows the relation $TX^{a_i}Z^{b_i}\ket{\psi}\simeq(S^\dagger)^{a_i} X^{a_i}Z^{a_i\oplus b_i}T\ket{\psi}$. For the adjoint gate one similarly has $T^\dagger X^{a_i}Z^{b_i}\ket{\psi}\simeq S^{a_i}X^{a_i}Z^{a_i\oplus b_i}T^\dagger\ket{\psi}$, where $\simeq$ denotes equality up to a global phase. These relations produce an $S$-type byproduct. To correct the error, the gate teleportation
method is utilized.

\textit{Algorithm for evaluating mixed Clifford--non-Clifford circuits.---}\\
To illustrate how Liang's scheme handles circuits with both Clifford and $T$ gates, consider
the following protocol:

\textit{Step 1: Clifford gates.---}\\
The server begins by applying all Clifford gates sequentially to the encrypted data qubits, until it finds a non-Clifford gate in the sequence of gates of the circuit.
For each Clifford gate, the server updates the encryption keys $(a_i,b_i)$ according to the
classical functions $f^G_{x_i}$ and $f^G_{z_i}$. The quantum state of the data qubits remains encrypted
throughout this phase.

\textit{Step 2: $T$ ($T^\dagger$) gate application.---}\\
When the server encounters a $T$ gate (or its adjoint $T^\dagger$) that must be applied to a data qubit
$w_i$, it employs the gate teleportation mechanism.  Specifically, the server applies the $T$ gate to $w_i$, then swaps the
states of $w_i$ and qubit $s_i$ from the pre-shared Bell
state $\ket{\Phi}_{s_ic_i}$.

\textit{Step 3: Interleaving Clifford and non-Clifford gates.---}\\
If additional gates follow the $T$ gate (whether Clifford or non-Clifford), the server
continues applying them in sequence:
\begin{itemize}
    \item For Clifford gates: standard key updates on the encrypted qubits.
    \item For additional $T$ or $T^\dagger$ gates: apply the gate on the required qubit and swap that qubit with the qubit of a fresh pre-shared Bell state.
    \end{itemize}

\textit{Step 4: Final output.---}\\
Once all gates have been applied, the server sends all qubits to the client, including
the qubits from the Bell states used for gate teleportation. The client applies the
final decryption operation using the accumulated encryption key information.

\item
\textit{Decryption process.---}
Upon receiving the qubits from the server, the client systematically applies classical operations on keys or quantum measurements on the qubit pairs $c_is_i$ depending on each gate application during the evaluation phase.

\textit{Step 1: Clifford gate key updates.---}\\
For each Clifford gate applied during server evaluation, the client uses the gate-dependent
key update functions to evolve the local keys. Since Clifford gates commute with Pauli
operators, these key updates are computed classically without requiring quantum measurements.

\textit{Step 2: Non-Clifford ($T$ and $T^\dagger$ gates) measurement and error correction. ---}\\
When the client encounters the $i$th application of a non-Clifford $T$ or $T^\dagger$ gate to qubit $w_j$, it
performs a specialized Bell measurement:
\begin{enumerate}
    \item
The client takes qubit $s_i$ (received from the server) and qubit $c_i$ (from the $i$th pre-shared Bell pair).
\item The client measures the pair of qubits $c_is_i$ in the \emph{$S^{a_j}$-rotated} (in case of application of $T$ gate) or \emph{$(S^\dagger)^{a_j}$-rotated} (for operation of $T^\dagger$ gate) Bell basis, where $a_j$ is the current value of
the key for qubit $j$. This basis is a rotation of the standard Bell basis by the unitary $S^{a_j}$.

\item The measurement yields two classical bits $(r_a, r_b)$ representing the measurement outcome.
\end{enumerate}

\textit{Why this measurement works.---} During server evaluation, the $T$ and $T^\dagger$ gates introduced a phase error
(the ``$S$-error"). The $S^{a_j}$-rotated and $(S^\dagger)^{a_j}$-rotated Bell measurements are specifically designed to detect and
eliminate these errors. The measurement outcome $(r_a, r_b)$ directly encodes the correction needed.

\textit{Step 3: Key update from measurement outcome.---}
Based on the measurement outcome $(r_a, r_b)$, the client updates the encryption keys for qubit $j$:
\begin{eqnarray}
&&a_j \rightarrow a_j \oplus r_a\text{, }b_j \rightarrow b_j \oplus (a_j \oplus r_b).
\end{eqnarray}
These updates ensure that the accumulated encryption keys reflect both the gates applied and
the measurement outcomes from non-Clifford operations.

\textit{Step 4: Final state recovery.---}\\
After processing all gates (updating keys for Clifford gates and performing measurements for
$T$ and $T^\dagger$ gates), let $(a^f_i, b^f_i)$ be the final encryption key pair for qubit $i$. The client
then applies the multi-qubit correction unitary
$\otimes_{i=1}^n X^{a^f_i} Z^{b^f_i}$
to all $n$ qubits $\{w_1, w_2, ..., w_n\}$, recovering the correct output state $\ket{\psi_{out}}$ of the original
quantum circuit.

\item
\textit{Summary.---} This four-step process ensures that all quantum information processed by the server
remains encrypted throughout evaluation, while the client retains the ability to recover the
correct output using only classical key information and Bell measurements.
\end{itemize}

\subsection{Two-qubit worked example}
\label{app:example}

\begin{figure}[h]
    \centering
\includegraphics[width=0.27\textwidth]{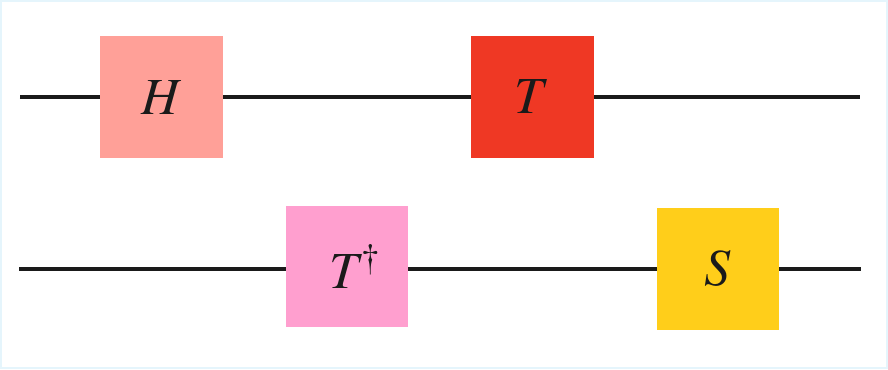}
    \caption{An example of a quantum circuit. The figure presents the structure of a two-qubit quantum circuit, which involves the action of an $H$ and a $T$ gate on the first qubit and $T^\dagger$ and $S$ gates on the second qubit. We discuss the process of implementing this circuit using the homomorphic encryption scheme in detail in Appendix \ref{A1}. }
    \label{fig3}
\end{figure}
Here we discuss an example of cloud computation with homomorphic encryption, where the encryption is done using Pauli gates, $X$ and $Z$. We consider the circuit shown in Figure \ref{fig3} which needs to be applied on two qubits, $w_1$ and $w_2$, initially acquiring the state $\ket{\psi}_{w_1w_2}$. The circuit involves the following gate operations:
\begin{enumerate}
    \item On the first qubit, $w_1$:
    \begin{enumerate}
        \item Application of a Clifford gate, $H$.
        \item Application of a non-Clifford gate, $T$.
    \end{enumerate}
     \item On the second qubit, $w_2$:
     \begin{enumerate}
        \item Application of a non-Clifford gate, $T^\dagger$.
        \item Application of a Clifford gate, $S$.
    \end{enumerate}
\end{enumerate}
The client wants to act the circuit on the initial two-qubit state, $\ket{\psi}_{w_1w_2}$, and prepare the final state $(TH\otimes ST^\dagger)\ket{\psi}_{w_1w_2}$ through cloud computation.
To evaluate the circuit on $\ket{\psi}_{w_1w_2}$ in the cloud, the client and server follow the steps described below.\\
\textit{Setup.---} The circuit includes the application of two non-Clifford gates, $T$ and $T^\dagger$, that implies $n_T=2$. Hence, the client generates two Bell states, $\ket{\Phi}_{s_1c_1}$ and $\ket{\Phi}_{s_2c_2}$, shared between the qubits $s_1-c_1$ and $s_2-c_2$. \\
\textit{Encryption.---}
To encrypt the two-qubit state, $\ket{\psi}_{w_1w_2}$, the client randomly picks two pairs of binary numbers, $(a^0_1,b^0_1)$ and $(a^0_2,b^0_2)$, applies the gates $X^{a^0_1}Z^{b^0_1}\otimes X^{a^0_2}Z^{b^0_2}$ on the state, $\ket{\psi}_{w_1w_2}$, and forms the encrypted state
\begin{equation}
\ket{\psi^{\text{enc}}}_{w_1w_2}=X^{a^0_1}Z^{b^0_1}\otimes X^{a^0_2}Z^{b^0_2}\ket{\psi}_{w_1w_2}.
\end{equation}
In this scheme, $\{(a^0_1,b^0_1),(a^0_2,b^0_2)\}$ is the set of secret keys, knowing which one can recover the original state, $\ket{\psi}_{w_1w_2}$, by applying $X$ and $Z$ gates appropriately on the encrypted state, $\ket{\psi^{\text{enc}}}_{w_1w_2}$.

Finally, after encryption, the client sends the encrypted state $\ket{\psi^{\text{enc}}}_{w_1w_2}$ to the server for the application of the circuit on the qubits, $w_1$ and $w_2$. The client also sends the qubits, $s_1$ and $s_2$, that are the parts of the two generated Bell states, to the server, which are needed to operate the two non-Clifford gates, $T$ and $T^\dagger$ (see Figure \ref{fig3}), involved in the circuit.\\
\textit{Evaluation.---}
After receiving $\ket{\psi^{\text{enc}}}_{w_1w_2}$, the server applies the gates in the given order and correspondingly keeps a note on the upgradation of the keys. In particular, the server takes the following actions sequentially:
\begin{enumerate}
    \item Application of $H$ on $w_1$: \begin{enumerate}
        \item It acts $H$ gate on the first qubit of $\ket{\psi^{\text{enc}}}_{w_1w_2}$ and prepares the state
        \begin{eqnarray*}
           &&(H\otimes I)(X^{a^0_1}Z^{b^0_1}\otimes X^{a^0_2}Z^{b^0_2})\ket{\psi}_{w_1w_2}\\
           =&&(X^{b^0_1}Z^{a^0_1}\otimes X^{a^0_2}Z^{b^0_2})(H\otimes I)\ket{\psi}_{w_1w_2}.
        \end{eqnarray*}
        \item The qubits, $c_1$, $s_1$, $c_2$, $s_2$, prepared in the Bell states, remain untouched. The complete state of the six qubits is
        \begin{equation*}          (X^{b^0_1}Z^{a^0_1}\otimes X^{a^0_2}Z^{b^0_2})(H\otimes I)\ket{\psi}_{w_1w_2}\ket{\Phi}_{s_1c_1}\ket{\Phi}_{s_2c_2}.
        \end{equation*}

        \item Server updates the first pair of keys as $a_1^1=b_1^0,b_1^1=a_1^0$, following the key updating functions, $f_{x_1}^H$ and $f_{z_1}^H$, and keeps the other pair of keys as it is, i.e., $a_2^1=a_2^0,b_2^1=b_2^0.$
    \end{enumerate}
    \item Application of $T$ on $w_1$:
    \begin{enumerate}
    \item It acts $T$ gate on the first qubit, $w_1$, and gets
    \begin{eqnarray*}
          && (T\otimes I)(X^{a^1_1}Z^{b^1_1}\otimes X^{a^1_2}Z^{b^1_2})(H\otimes I)\ket{\psi}_{w_1w_2}\\
           =&&((S^\dagger)^{a^1_1}X^{a^1_1}Z^{a^1_1\oplus b^1_1}\otimes X^{a^1_2}Z^{b^1_2})(TH\otimes I)\ket{\psi}_{w_1w_2}.
        \end{eqnarray*}
        \item Swaps the qubit $w_1$ with $s_1$ and keeps the qubits $w_2$ and $s_2$ as it is. The complete state of the six qubits at the end of this action is
        \begin{eqnarray*}
            &&((S^\dagger)^{a^1_1}X^{a^1_1}Z^{a^1_1\oplus b^1_1}\otimes X^{a^1_2}Z^{b^1_2})\\&&(TH\otimes I)\ket{\psi}_{s_1w_2} \ket{\Phi}_{w_1c_1}\ket{\Phi}_{s_2c_2}\\
            &&=\frac{1}{2}\sum_{r_a,r_b} \ket{\Phi(S^{a_1^1})_{r_ar_b}}_{s_1c_1}\\&&[X^{a_1^1\oplus r_a}Z^{a_1^1\oplus b_1^1\oplus r_b}
            \otimes X^{a_2^1}Z^{b_2^1}(TH\otimes I)]\ket{\psi}_{w_1w_2}\ket{\Phi}_{s_2c_2}.
        \end{eqnarray*}
        \item The corresponding key update, to be applied once the measurement outcomes are known, is $a_1^2=a_1^1\oplus r_a$, $b_1^2=a_1^1\oplus b_1^1\oplus r_b$, while the other pair of keys is unchanged, i.e., $a_2^2=a_2^1,b_2^2=b_2^1.$ For simplicity of notation we denote the first pair of keys as $a_1^2$ and $b_1^2$, but one must not forget about its dependency on $r_a$ and $r_b$.
        \end{enumerate}
    \item Application of $T^{\dagger}$ on $w_2$:
    \begin{enumerate}
        \item It acts $T^\dagger$ gate on qubit, $w_2$, and gets\begin{eqnarray*}
          && (I\otimes T^\dagger)(X^{a_1^2}Z^{b_1^2}
            \otimes X^{a_2^2}Z^{b_2^2})(TH\otimes I)\ket{\psi}_{w_1w_2}\\
           =&&(X^{a_1^2}Z^{b_1^2}
            \otimes S^{a^2_2}X^{a_2^2}Z^{a_2^2\oplus b_2^2})(TH\otimes T^{\dagger})\ket{\psi}_{w_1w_2}.
        \end{eqnarray*}
        \item Swaps the qubit $w_2$ with $s_2$. Keeps the qubits $w_1$ and $s_1$ as they are. Hence the complete state of the six qubits is
        \begin{eqnarray*}
            &&\frac{1}{2}\sum_{r_a,r_b}\ket{\Phi(S^{a_1^1})_{r_ar_b}}_{s_1c_1}(X^{a_1^2}Z^{b_1^2}
            \otimes S^{a^2_2}X^{a_2^2}Z^{a_2^2\oplus b_2^2})\\&&\hspace{3cm}(TH\otimes T^{\dagger})\ket{\psi}_{w_1s_2} \ket{\Phi}_{w_2c_2}\\&&
            =\frac{1}{4}\sum_{r_a,r_b,r'_a,r'_b} \ket{\Phi(S^{a_1^1})_{r_ar_b}}_{s_1c_1}\ket{\Phi((S^\dagger)^{a_2^2})_{r'_ar'_b}}_{s_2c_2} \\&&(X^{a_1^2}Z^{b_1^2}
            \otimes X^{a_2^2\oplus r'_a}Z^{a_2^2\oplus b_2^2\oplus r'_b})(TH\otimes T^{\dagger})\ket{\psi}_{w_1w_2}.
        \end{eqnarray*}
        \item The corresponding key update leaves the first pair unchanged, i.e., $a^3_1=a_1^2$, $b_1^3=b_1^2$, and maps the second pair to $a_2^3=a_2^2\oplus r'_a$, $b_2^3=a_2^2\oplus b_2^2 \oplus r'_b$. Again, readers must keep in mind that $a_2^3$ and $b_2^3$ depend on $r'_a$ and $r'_b$, respectively.
        \end{enumerate}
        \item Application of $S$ gate on $w_2$:
    \begin{enumerate}
        \item It acts $S$ gate on qubit, $w_2$, and gets\begin{eqnarray*}
          &&(I\otimes S) (X^{a_1^3}Z^{b_1^3}
            \otimes X^{a_2^3}Z^{b_2^3})(TH\otimes T^{\dagger})\ket{\psi}_{w_1w_2}\\
            &&=(X^{a_1^3}Z^{b_1^3}
            \otimes X^{a_2^3}Z^{a_2^3\oplus b_2^3})(TH\otimes ST^{\dagger})\ket{\psi}_{w_1w_2}.
        \end{eqnarray*}
        \item Does not change the state of the other qubits. The entire state of the six qubits is
        \begin{eqnarray*}        &&\frac{1}{4}\sum_{r_a,r_b,r'_a,r'_b} \ket{\Phi(S^{a_1^1})_{r_ar_b}}_{s_1c_1}\ket{\Phi((S^\dagger)^{a_2^2})_{r'_ar'_b}}_{s_2c_2}\\
        &&(X^{a_1^3}Z^{b_1^3}
            \otimes X^{a_2^3}Z^{a_2^3\oplus b_2^3})(TH\otimes ST^{\dagger})\ket{\psi}_{w_1w_2}.
        \end{eqnarray*}
        \item The server keeps the first pair of keys unchanged, i.e., $a_1^4=a_1^3$, $b_1^4=b_1^3$ and updates the second pair of keys following the rules of $f_{x_2}^S$ and $f_{z_2}^S$ as  $a_2^4=a_2^3,b_2^4=a_2^3\oplus b_2^3$.
        \end{enumerate}
\end{enumerate}
After the end of the application of all gates, the final joint state of all the qubits is
\begin{eqnarray*}
    &&\frac{1}{4}\sum_{r_a,r_b,r'_a,r'_b} \ket{\Phi(S^{a_1^1})_{r_ar_b}}_{s_1c_1}\ket{\Phi((S^\dagger)^{a_2^2})_{r'_ar'_b}}_{s_2c_2}\\
        &&(X^{a_1^4}Z^{b_1^4}
            \otimes X^{a_2^4}Z^{b_2^4})(TH\otimes ST^{\dagger})\ket{\psi}_{w_1w_2}.
\end{eqnarray*}
At this point, the server sends all the qubits to the client along with the key updating rules.
\\
\textit{Decryption.---}
 The client can now decrypt the data by using the key updating rules, performing rotated Bell measurements, and applying $X$ and $Z$ gates appropriately. The steps followed by the client for each application of the gates are described below.
 \begin{enumerate}
     \item Knowing the key updating functions, $f^H_{x_1}$ and $f^H_{z_1}$, the client finds $(a_1^1,b_1^1)$. The keys $(a_2^1,b_2^1)$ are known to be the same as $(a_2^0,b_2^0)$.
     \item The client performs a measurement on the $\left\{\ket{\Phi(S^{a_1^1})_{r_ar_b}}\right\}$ basis on $s_1c_1$ qubits. From the measurement outcomes, say $r_a$ and $r_b$, the client gets the keys, $a_1^2=a_1^1\oplus r_a$ and $b_1^2=a_1^1\oplus b_1^1\oplus r_b$. The keys of the second qubits remain the same, i.e., $a_2^2=a_2^1$ and $b_2^2=b_2^1$.
     \item Again the client performs a measurement in another rotated Bell basis, which is $\{\ket{\Phi((S^\dagger)^{a_2^2})_{r'_ar'_b}}\}$, on $s_2c_2$. Depending on the measurement outcomes, $r'_a$ and $r'_b,$ it updates the keys as $a_2^3=a_2^2\oplus r'_a$ and $b_2^3=a_2^2\oplus b_2^2 \oplus r'_b$. The keys corresponding to the first qubit are kept unchanged, i.e., $a_1^3=a_1^2$ and $b_1^3=b_1^2$.
     \item From the key updating functions, $f^S_{x_2}$ and $f^S_{z_2}$, the client evaluates $(a_2^4,b_2^4)$. The keys of the first pair of qubits are kept fixed at $a_1^4=a_1^3$, $b_1^4=b_1^3$.
 \end{enumerate}
At the end of the measurements, the state of the qubits, $w_1$ and $w_2$, collapses to $$\ket{\psi^{\text{enc}}_{final}}=(X^{a_1^4}Z^{b_1^4}
            \otimes X^{a_2^4}Z^{b_2^4})(TH\otimes ST^{\dagger})\ket{\psi}_{w_1w_2}.$$
The client finally gets the targeted state by applying $(Z^{b_1^4}X^{a_1^4}
            \otimes Z^{b_2^4}X^{a_2^4})$ on the state $\ket{\psi^{\text{enc}}_{final}}$, which is $$(Z^{b_1^4}X^{a_1^4}
            \otimes Z^{b_2^4}X^{a_2^4})\ket{\psi^{\text{enc}}_{final}}= (TH\otimes ST^{\dagger})\ket{\psi}_{w_1w_2},$$ that is the final desired state.

\section{Encrypted amplitudes and code-space preservation for the repetition and Shor codes}
\label{A2}

To protect a logical qubit in the state $\ket{\bar{\psi}}=c_0\ket{\bar{0}}+c_1\ket{\bar{1}}$ from bit-flip errors, the three-qubit bit-flip code encodes $\ket{\bar{0}}=\ket{000}$ and $\ket{\bar{1}}=\ket{111}$, so that $\ket{\bar{\psi}}=c_0\ket{{000}}+c_1\ket{{111}}$. To store this state in the cloud without revealing the logical amplitudes, the client chooses a random pair of binary keys, $(a,b)$, and applies the transversal Pauli mask $(X^aZ^b)^{\otimes 3}$. The encrypted state is then
\begin{equation}
\ket{\bar{\psi}^{\text{enc}}}
=(X^aZ^b)^{\otimes 3}\ket{\bar{\psi}}
=c'_0\ket{{000}}+c'_1\ket{{111}}.
\end{equation}
The exact form of the amplitudes, $c'_0$ and $c'_1$, depends on the keys $(a,b)$ as shown in Table \ref{T1}. In all four cases, the encrypted state remains in the bit-flip code space, which is spanned by $\{\ket{{000}},\ket{{111}}\}$.
\vspace{5mm}
\begin{table}[h]
    \centering
    \begin{tabular}{|c|c|c|}
    \hline
        $\boldsymbol{a}$, $\boldsymbol{b}$ &$\boldsymbol{c_0'}$&$\boldsymbol{c_1'}$\\
                  \hline
          $a=0$, $b=0$&$c_0$& $c_1$\\
          \hline
        $a=0$, $b=1$& $c_0$& $-c_1$\\
        \hline
         $a=1$, $b=0$ &$c_1$& $c_0$\\
         \hline
         $a=1$, $b=1$&$-c_1$&$c_0$\\
         \hline

    \end{tabular} \caption{Values of $c_0'$ and $c_1'$ for the four possible choices of $a$ and $b$.}
    \label{T1}
\end{table}

\subsection{Code-space preservation for the Shor code}
\begin{lemma}\label{lem:shor}

Let a logical qubit be encoded in the Shor code as
\begin{eqnarray}
   \ket{0}\rightarrow|\bar{0}\rangle = \frac{1}{2\sqrt{2}} (|000\rangle + |111\rangle)^{\otimes 3}, \\
\ket{1}\rightarrow|\bar{1}\rangle = \frac{1}{2\sqrt{2}} (|000\rangle - |111\rangle)^{\otimes 3},
\end{eqnarray}
and let $|\bar{\psi}\rangle = c_0 |\bar{0}\rangle + c_1 |\bar{1}\rangle$
be its encoded state. For any encryption keys $a,b \in \{0,1\}$, the homomorphic encryption operation
$U_{\text{enc}}(a,b) = (X^a Z^b)^{\otimes 9}$
preserves the Shor code space, i.e.,
$U_{\text{enc}}(a,b)\,|\bar{\psi}\rangle \in \mathrm{span}\{|\bar{0}\rangle, |\bar{1}\rangle\}.$
Equivalently,
\begin{equation}
   U_{\text{enc}}(a,b)\,|\bar{\psi}\rangle
=
c_0'(a,b)\,|\bar{0}\rangle + c_1'(a,b)\,|\bar{1}\rangle,
\end{equation}
where $c_0'$ and $c_1'$ are obtained from $c_0$ and $c_1$ by a key-dependent logical Pauli transformation.
\end{lemma}
\begin{proof}
We can write the encoded state as
\begin{equation}
    |\bar{\psi}\rangle
=
\frac{1}{2\sqrt{2}}
\left[
c_0 (|000\rangle + |111\rangle)^{\otimes 3}
+
c_1 (|000\rangle - |111\rangle)^{\otimes 3}
\right].
\end{equation}
We first analyze the action of $X^{\otimes 9}$ and $Z^{\otimes 9}$ on the Shor code basis states. Acting on a single three-qubit block,
\begin{equation}
  X^{\otimes 3} |000\rangle = |111\rangle,
\quad
X^{\otimes 3} |111\rangle = |000\rangle,
\end{equation}
so
\begin{equation}
  X^{\otimes 3} (|000\rangle \pm |111\rangle)
=
|111\rangle \pm |000\rangle
=
\pm (|000\rangle \pm |111\rangle).
\end{equation}
Applying this to all three blocks gives
\begin{equation}
    X^{\otimes 9} (|000\rangle \pm |111\rangle)^{\otimes 3}
=
\pm (|000\rangle \pm |111\rangle)^{\otimes 3}.
\end{equation}
Similarly, for $Z^{\otimes 3}$ one has
\begin{equation}
    Z^{\otimes 3} |000\rangle = |000\rangle,
\quad
Z^{\otimes 3} |111\rangle = -|111\rangle,
\end{equation}
and hence
\begin{equation}
  Z^{\otimes 3} (|000\rangle \pm |111\rangle)
=
|000\rangle \mp |111\rangle,
\end{equation}
which implies
\begin{equation}
    Z^{\otimes 9} (|000\rangle \pm |111\rangle)^{\otimes 3}
=
(|000\rangle \mp |111\rangle)^{\otimes 3}.
\end{equation}
Now consider the four possible choices of $(a,b)$:

\begin{enumerate}[label=\roman*.]
    \item \textit{$(a,b)=(0,0)$.---} Then $U_{\text{enc}}=I^{\otimes9}$; the statement is trivial.
    \item \textit{$(a,b)=(1,0)$.---} Then $U_{\text{enc}}=X^{\otimes9}$. Using the identities above,
    \begin{equation}    X^{\otimes9}|\bar{0}\rangle=|\bar{0}\rangle, \quad X^{\otimes9}|\bar{1}\rangle=-|\bar{1}\rangle,
    \end{equation}
    and thus $X^{\otimes9}|\bar{\psi}\rangle=c_{0}|\bar{0}\rangle-c_{1}|\bar{1}\rangle$, which is still a superposition of $|\bar{0}\rangle$ and $|\bar{1}\rangle$.
    \item \textit{$(a,b)=(0,1)$.---} Then $U_{\text{enc}}=Z^{\otimes9}$. From the identities,
    \begin{equation}   Z^{\otimes9}|\bar{0}\rangle=|\bar{1}\rangle, \quad Z^{\otimes9}|\bar{1}\rangle=|\bar{0}\rangle,
    \end{equation}
    thus
    \begin{equation}
        Z^{\otimes9}|\bar{\psi}\rangle=c_{0}|\bar{1}\rangle+c_{1}|\bar{0}\rangle,
    \end{equation}
    again in the span of $\{|\bar{0}\rangle, |\bar{1}\rangle\}$.
    \item \textit{$(a,b)=(1,1)$.---} Then $U_{\text{enc}}=(XZ)^{\otimes9}=X^{\otimes 9}Z^{\otimes 9}$. Combining the previous cases:
    \begin{equation}
(XZ)^{\otimes9}|\bar{0}\rangle=X^{\otimes9}Z^{\otimes9}|\bar{0}\rangle=X^{\otimes9}|\bar{1}\rangle=-|\bar{1}\rangle
    \end{equation}
    \begin{equation}
(XZ)^{\otimes9}|\bar{1}\rangle=X^{\otimes9}Z^{\otimes9}|\bar{1}\rangle=X^{\otimes9}|\bar{0}\rangle=|\bar{0}\rangle
    \end{equation}
    and therefore $(XZ)^{\otimes9}|\bar{\psi}\rangle=-c_{0}|\bar{1}\rangle+c_{1}|\bar{0}\rangle$.
\end{enumerate}

In all four cases, $U_{\text{enc}}(a,b)|\bar{\psi}\rangle$ is a linear combination of $|\bar{0}\rangle$ and $|\bar{1}\rangle$; no component orthogonal to the Shor code space is generated. Hence there exist amplitudes $c_{0}^{\prime}(a,b)$, $c_{1}^{\prime}(a,b)$ such that
\begin{equation}
U_{\text{enc}}(a,b)|\bar{\psi}\rangle=c_{0}^{\prime}(a,b)|\bar{0}\rangle+c_{1}^{\prime}(a,b)|\bar{1}\rangle,
\end{equation}
and the Shor code space is preserved by the encryption map. This proves the lemma.
\end{proof}

\section{Encoded states and CSS compatibility of the Steane code}
\label{A3}
The Steane code can be described using two linear classical codes, denoted by $C_1$ and $C_2$, with dimensions $k_1=4$ and $k_2=3$, respectively. We use generator matrices written as maps from binary column vectors to length-seven codewords:
\begin{eqnarray}
    \mathbb{G}[C_1]=\left[\begin{matrix}
        1&0&0&0\\
        0&1&0&0\\
        0&0&1&0\\
        0&0&0&1\\
        0&1&1&1\\
        1&0&1&1\\
        1&1&0&1
    \end{matrix}\right] \text{ and }
    \mathbb{G}[C_2]=\left[\begin{matrix}
        0&0&1\\
        0&1&0\\
        0&1&1\\
        1&0&0\\
        1&0&1\\
        1&1&0\\
        1&1&1
    \end{matrix}\right].
\end{eqnarray}
Thus, a word $x\in\mathbb{F}_2^4$ is encoded as $\mathbb{G}[C_1]x$, while a word $y\in\mathbb{F}_2^3$ is encoded as $\mathbb{G}[C_2]y$, with all operations performed over $\mathbb{F}_2$. For example, for $x=(1,0,1,0)^T$ and $y=(0,1,1)^T$, one obtains
\begin{equation}
\mathbb{G}[C_1]x=1010101,\qquad
\mathbb{G}[C_2]y=1100110.
\end{equation}
Multiplying all possible binary input words by the corresponding generator matrices gives
\begin{widetext}
\begin{eqnarray}
C_1 &=& \{0000000, 0001111, 0010110, 0011001, 0100101, 0101010, 0110011, 0111100,\nonumber\\
&&\quad 1000011, 1001100, 1010101, 1011010, 1100110, 1101001, 1110000,1111111\},\label{myeq15}\\
C_2&=&\{0000000, 1010101, 0110011, 1100110, 0001111, 1011010, 0111100, 1101001\}.\label{myeq16}
\end{eqnarray}
\end{widetext}
Since $C_2\subset C_1$, the CSS construction encodes $k_1-k_2=1$ logical qubit. The two logical computational basis states are
\begin{widetext}
\begin{eqnarray}
\ket{\bar{0}} &=& \frac{1}{\sqrt{8}} \Big[
\ket{0000000} + \ket{1010101} + \ket{0110011} + \ket{1100110}\nonumber\\
&&\quad + \ket{0001111} + \ket{1011010} + \ket{0111100} + \ket{1101001}
\Big], \label{eq:0L}\\
\ket{\bar{1}} &=& \frac{1}{\sqrt{8}} \Big[
\ket{1111111} + \ket{0101010} + \ket{1001100} + \ket{0011001}\nonumber\\
&&\quad + \ket{1110000} + \ket{0100101} + \ket{1000011} + \ket{0010110}
\Big]. \label{eq:1L}
\end{eqnarray}
\end{widetext}
These expressions follow from Eq. \eqref{eq:css-basis}: $\ket{\bar{0}}$ is the uniform superposition over $C_2$, while $\ket{\bar{1}}$ is the uniform superposition over the coset $C_2\oplus e$, where $e=1111111$.

We now verify explicitly the CSS compatibility criterion of Corollary~\ref{thm:css-compatibility}. Let $e=1111111$ be the all-ones word of length seven. From Eq. \eqref{myeq15}, $e\in C_1$. Hence, for every $x\in C_1$, linearity of $C_1$ implies $x\oplus e\in C_1$, so the condition associated with $X^{\otimes 7}$ is satisfied. From Eq. \eqref{myeq16}, every word in $C_2$ has even Hamming weight:
\begin{equation}
w(y)\equiv e\cdot y=0 \pmod 2,\qquad \forall\,y\in C_2 .
\end{equation}
Equivalently, $e\in C_2^\perp$, so the condition associated with $Z^{\otimes 7}$ is also satisfied. Therefore the Steane CSS code satisfies both conditions of Corollary~\ref{thm:css-compatibility}. Explicitly,
\begin{equation}
X^{\otimes 7}\ket{\bar{0}}=\ket{\bar{1}},\qquad
X^{\otimes 7}\ket{\bar{1}}=\ket{\bar{0}},
\end{equation}
and
\begin{equation}
Z^{\otimes 7}\ket{\bar{0}}=\ket{\bar{0}},\qquad
Z^{\otimes 7}\ket{\bar{1}}=-\ket{\bar{1}}.
\end{equation}
Thus, in the CSS/Hamming-code presentation used here, $X^{\otimes 7}$ and $Z^{\otimes 7}$ preserve the code space and act as logical Pauli representatives on the encoded qubit \cite{colorcode}.

\section{Compatibility of the $[[15,1,3]]$ triorthogonal code with Proposition~\ref{thm:stabilizer-compatibility}}\label{Aqrm}
\label{X}
In this Appendix, we present an explicit triorthogonal CSS code and verify that it satisfies the homomorphic-compatibility criterion stated in Proposition \ref{thm:stabilizer-compatibility} of the main text. According to Proposition~\ref{thm:stabilizer-compatibility}, a stabilizer code, encoding 1 logical qubit in $n$ physical qubits, is compatible with the transversal Pauli encryption scheme generated by $X^{\otimes n}$ and $Z^{\otimes n}$ if and only if these operators commute with all stabilizer generators.

We consider the standard $[[15,1,3]]$ triorthogonal Reed--Muller code \cite{steane1996quantumreedmullercodes}. In the triorthogonal construction reviewed in Sec.~\ref{sec4}, a binary matrix $G$ is decomposed into a set $G_1$ of odd-weight rows and a set $G_0$ of even-weight rows, where the rows in $G_0$ define the CSS stabilizer structure and the rows in $G_1$ determine the logical operators.

Let the columns be ordered by the nonzero elements of $\mathbb{F}_2^4$, i.e., the nonzero binary 4-vectors and consider the binary matrix
\setcounter{MaxMatrixCols}{15}
\begin{equation}
G=
\begin{pmatrix}
1&1&1&1&1&1&1&1&1&1&1&1&1&1&1\\
0&0&0&0&0&0&0&1&1&1&1&1&1&1&1\\
0&0&0&1&1&1&1&0&0&0&0&1&1&1&1\\
0&1&1&0&0&1&1&0&0&1&1&0&0&1&1\\
1&0&1&0&1&0&1&0&1&0&1&0&1&0&1
\end{pmatrix}.
\end{equation}
Denoting its rows by $r_0,r_1,r_2,r_3,r_4$, we observe that $r_0$ has odd Hamming weight $15$, whereas each of $r_1,r_2,r_3,r_4$ has even Hamming weight $8$. Therefore, in the notation of Sec.~\ref{sec4}, we identify
\begin{equation}
G_1=\{r_0\},\qquad G_0=\{r_1,r_2,r_3,r_4\},
\end{equation}
so that the resulting code encodes a single logical qubit, as expected for the $[[15,1,3]]$ code.

\subsection{Triorthogonality}

We first verify that the matrix $G$ is triorthogonal. By definition, this requires that every pair of distinct rows have even overlap and that every triple of distinct rows also have even overlap.

For the pairwise overlaps, one finds
\begin{equation}
w(r_0\wedge r_i)=8,\qquad i=1,2,3,4,
\end{equation}
and, for $i\neq j$ with $i,j\in\{1,2,3,4\}$,
\begin{equation}
w(r_i\wedge r_j)=4.
\end{equation}
Hence, every pairwise overlap has even parity.

For the triple overlaps, one similarly obtains
\begin{equation}
w(r_0\wedge r_i\wedge r_j)=4,\qquad i\neq j,
\end{equation}
and
\begin{equation}
w(r_i\wedge r_j\wedge r_k)=2,\qquad i,j,k\ \text{distinct in}\ \{1,2,3,4\}.
\end{equation}
Therefore, every triple overlap is also even, and the matrix $G$ is triorthogonal.

\subsection{CSS code associated with $G$}

Following the construction described in Sec.~\ref{sec4}, the even-weight rows in $G_0$ generate the $X$-type stabilizer sector of the CSS code, while the full matrix $G$ determines the logical structure of the encoded subspace.  Since $G_1$ contains a single odd-weight row, the corresponding CSS code encodes one logical qubit.

The $X$-type stabilizer generators are therefore
\begin{equation}
X(r_i)=\bigotimes_{j=1}^{15} X_j^{(r_i)_j},\qquad i=1,2,3,4.
\end{equation}
The $Z$-type stabilizers are generated by binary vectors $s\in G^\perp$, namely vectors orthogonal over $\mathbb{F}_2$ to every row of $G$, through
\begin{equation}
Z(s)=\bigotimes_{j=1}^{15} Z_j^{s_j}.
\end{equation}
This gives the standard CSS realization of the $[[15,1,3]]$ triorthogonal code.

\subsection{Compatibility with transversal Pauli encryption}

We now verify that this code satisfies Proposition~\ref{thm:stabilizer-compatibility}. In the notation of the main text, the encryption is implemented by the transversal physical operators $X^{\otimes 15}$ and $Z^{\otimes 15}$. Proposition~\ref{thm:stabilizer-compatibility} states that the homomorphic encryption scheme is compatible with the stabilizer code if and only if
\begin{equation}
[X^{\otimes 15},g_i]=0,\qquad [Z^{\otimes 15},g_i]=0
\end{equation}
for every stabilizer generator $g_i$.

\subsubsection{Commutation with $Z^{\otimes 15}$}

Let $v\in\mathbb{F}_2^{15}$ and define
\begin{equation}
X(v)=\bigotimes_{j=1}^{15} X_j^{v_j}.
\end{equation}
Then
\begin{equation}
Z^{\otimes 15}X(v)=(-1)^{w(v)}X(v)Z^{\otimes 15}.
\end{equation}
Since each row $r_i$ in $G_0$ has even weight,
\begin{equation}
w(r_i)=8,\qquad i=1,2,3,4,
\end{equation}
it follows immediately that
\begin{equation}
[Z^{\otimes 15},X(r_i)]=0,\qquad i=1,2,3,4.
\end{equation}
Moreover, $Z^{\otimes 15}$ trivially commutes with every $Z$-type stabilizer $Z(s)$, since both are tensor products of $Z$ operators and identities. Hence, $Z^{\otimes 15}$ commutes with all stabilizer generators.

\subsubsection{Commutation with $X^{\otimes 15}$}

Now let
\begin{equation}
Z(s)=\bigotimes_{j=1}^{15} Z_j^{s_j},
\end{equation}
with $s\in G^\perp$. Then
\begin{equation}
X^{\otimes 15}Z(s)=(-1)^{w(s)}Z(s)X^{\otimes 15}.
\end{equation}
Therefore, $X^{\otimes 15}$ commutes with $Z(s)$ whenever $w(s)$ is even.

Since $s\in G^\perp$, it is orthogonal to every row of $G$, and in particular to the all-ones row $r_0$. Thus,
\begin{equation}
0=s\cdot r_0=\sum_{j=1}^{15}s_j=w(s) \pmod 2.
\end{equation}
Hence, every vector $s\in G^\perp$ has even Hamming weight, which implies
\begin{equation}
[X^{\otimes 15},Z(s)]=0
\end{equation}
for every $Z$-type stabilizer generator. Since $X^{\otimes 15}$ also commutes trivially with all $X$-type stabilizers, it follows that $X^{\otimes 15}$ commutes with the entire stabilizer group.

\vspace{2mm}
We have shown that both transversal Pauli encryption operators, $X^{\otimes 15}$ and $Z^{\otimes 15}$, commute with all stabilizer generators of the $[[15,1,3]]$ triorthogonal code. Therefore, by Proposition~\ref{thm:stabilizer-compatibility}, this code is compatible with the homomorphic encryption scheme considered in the main text.
In particular, the encryption maps encoded states into the code subspace, so syndrome extraction and error correction can be performed directly on encrypted data without prior decryption. This is precisely the sense in which the code is homomorphically compatible for encrypted cloud storage in the ideal stabilizer-code model.

Furthermore, since this code is triorthogonal, it also belongs to the class of CSS codes admitting a transversal non-Clifford phase gate. With the convention used above, the codewords in the even row span of $G_0$ have Hamming weight $0$ modulo $8$, whereas the codewords in the logical coset $r_0+{\rm span}(G_0)$ have Hamming weight $7$ modulo $8$. Hence $T^{\otimes 15}$ acts on the encoded qubit as a logical phase gate equivalent to $\bar{T}^{\dagger}$ up to an overall phase. Equivalently, by composing with an appropriate logical Clifford correction, one obtains the desired logical $\bar{T}$ action, in agreement with the triorthogonal-code discussion in Sec.~\ref{sec4}. Thus, the $[[15,1,3]]$ tetrahedral 3D color code (Reed--Muller code) provides an explicit example satisfying both the triorthogonality conditions and the compatibility condition of Proposition~\ref{thm:stabilizer-compatibility}.

\subsection{Extension to larger tetrahedral 3D color codes}
The $[[15,1,3]]$ triorthogonal code considered above is the smallest member of a broader
family of tetrahedral 3D color codes constructed from punctured $3$-colexes
\cite{BombinMartinDelgado2007PRL,BombinMartinDelgado2007PRB}. The compatibility
argument used above extends to the tetrahedral family considered here once one reformulates
the Pauli-encryption condition in geometric terms. This result should be viewed as
complementary to the known fault-tolerant feature of suitable 3D color codes, namely
that transversal physical non-Clifford rotations can implement logical non-Clifford gates.

Consider a tetrahedral 3D color code defined on a punctured $3$-colex with $n$ physical qubits
placed on vertices \cite{BombinMartinDelgado2007PRL,BombinMartinDelgado2007PRB}. In the
stabilizer convention used here, a generating set may be chosen with $Z$-type operators
supported on faces and $X$-type operators supported on 3-cells. Thus, for each relevant face
$F$ and each relevant 3-cell $C$, we write
\begin{equation}
S_Z(F)=\prod_{v\in F} Z_v,
\qquad
S_X(C)=\prod_{v\in C} X_v.
\end{equation}
We now verify that the transversal Pauli operators
\begin{equation}
X^{\otimes n}=\prod_{v\in V} X_v,
\qquad
Z^{\otimes n}=\prod_{v\in V} Z_v
\end{equation}
commute with all stabilizer generators. By Proposition~\ref{thm:stabilizer-compatibility}, this will imply homomorphic compatibility
with transversal Pauli encryption for such tetrahedral codes.
A basic property of punctured $3$-colexes is that their vertices admit a bipartition into starred
and unstarred sets \cite{BombinMartinDelgado2007PRB},
\begin{equation}
V=V_\bullet \sqcup V_\star,
\end{equation}
with the parity property needed below: the stabilizer supports used in this construction contain
an even number of vertices. In particular, every face $F$ supporting a $Z$-type generator and
every 3-cell $C$ supporting an $X$-type generator has even cardinality:
\begin{equation}
w(F)\equiv 0 \pmod 2,
\qquad
w(C)\equiv 0 \pmod 2.
\end{equation}
We first consider commutation with $Z$-type stabilizers. Since $X$ anticommutes with $Z$ on
the same qubit, for any face $F$ one has
\begin{equation}
X^{\otimes n} S_Z(F)=(-1)^{w(F)} S_Z(F) X^{\otimes n}.
\end{equation}
Because $w(F)$ is even, it follows that
\begin{equation}
[X^{\otimes n},S_Z(F)]=0
\end{equation}
for every face $F$. Moreover, $X^{\otimes n}$ trivially commutes with every $X$-type stabilizer
$S_X(C)$.
Similarly, for any 3-cell $C$,
\begin{equation}
Z^{\otimes n} S_X(C)=(-1)^{w(C)} S_X(C) Z^{\otimes n}.
\end{equation}
Since $w(C)$ is even, we obtain
\begin{equation}
[Z^{\otimes n},S_X(C)]=0
\end{equation}
for every 3-cell $C$. Also, $Z^{\otimes n}$ trivially commutes with every $Z$-type stabilizer
$S_Z(F)$.
Hence both transversal Pauli operators, $X^{\otimes n}$ and $Z^{\otimes n}$, commute with all
stabilizer generators of the tetrahedral 3D color code. Therefore, by
Proposition~\ref{thm:stabilizer-compatibility}, the tetrahedral 3D color codes satisfying the above
standard colex parity conditions are compatible with the homomorphic encryption scheme
generated by transversal Pauli encryption.
For the usual tetrahedral patches encoding a single logical qubit, the global Pauli operators
$X^{\otimes n}$ and $Z^{\otimes n}$ may be identified, up to stabilizers and convention-dependent
logical labeling, with encoded Pauli representatives. The commutation proof itself, however,
only establishes the normalizer property required for code-space preservation under Pauli
masking; it does not establish transversal implementation of non-Clifford gates.
This shows that the compatibility established explicitly above for the $[[15,1,3]]$ tetrahedral 3D
color code is not an isolated feature of the minimal code, but a geometric parity property of the
tetrahedral family in this Pauli-encryption sense. In particular, encrypted encoded states remain inside the code subspace, so
syndrome extraction and error correction can be carried out directly on encrypted data without
prior decryption.

\section{Homomorphic encryption scheme using logical gates with Shor code}
\label{A4}
In this Appendix, we consider a scenario where a client wants to apply a $T$ gate to a single logical qubit in the cloud. Let the encoded logical state be $\ket{\bar{\psi}}$. The target is to prepare $\bar{T}\ket{\bar{\psi}}$. To protect it from noise, the client encodes the state using the Shor code, and to encrypt the state, the client further applies logical Pauli gates on it depending on encryption keys. The process of applying a logical $\bar{T}$ gate to the encoded encrypted state in the cloud is described below in detail, starting from a discussion of the encoding.
The logical representatives constructed below are used to show that suitable unitary extensions exist on the full physical Hilbert space. They should not be interpreted as efficient or fault-tolerant circuit syntheses of these gates.
\subsection{Stabilizer formalism of the Shor code}
We first recall the stabilizer description of the Shor code.\\
\textit{Stabilizer generators.---}
    The following independent commuting stabilizers generate a stabilizer group:
\begin{eqnarray}
    g_1 &=& Z_1\otimes Z_2\text{, }
g_2 = Z_2 \otimes  Z_3\text{, }
g_3= Z_4\otimes  Z_5,\\
g_4 &=& Z_5\otimes  Z_6\text{, }
g_5 = Z_7\otimes  Z_8\text{, }
g_6 = Z_8\otimes  Z_9,\\
g_7 &=& X_1\otimes  X_2\otimes  X_3\otimes  X_4\otimes  X_5\otimes  X_6,\\
g_8 &=& X_4\otimes  X_5\otimes  X_6\otimes  X_7\otimes  X_8\otimes  X_9,
\end{eqnarray}
where $Z_i\otimes Z_j$ denotes application of the Pauli $Z$ gate on the $i$th and $j$th qubits and identity on the other qubits. Similarly, $g_7$ and $g_8$ involve application of $X$ gates on the qubits denoted in the subscript of $X$ and identity on the unmentioned qubits. The stabilizer group generated by $\{g_i\}_{i=1}^8$ does not contain $-I$. Moreover, all generators are elements of the 9-qubit Pauli group, $G_9$. Since there are eight independent generators, the dimension of $V_S$, the vector space stabilized by the stabilizer group, is $2^{9-8}=2$.\\
\textit{Logical $\bar{Z}$ operator.---} Another independent operator, $\bar{Z} \in G_9$, is defined as $\bar{Z}=X^{\otimes 9}$, which commutes with all the generators, $\{g_i\}$. The action of this physical representative on the code space is the logical phase-flip action.\\
\textit{Logical $\bar{X}$ operator.---} In addition, another operator can be defined as $\bar{X}=Z^{\otimes 9}$, which belongs to $G_9$ and satisfies the relations $\bar{X}\bar{Z}=-\bar{Z}\bar{X}$ and $\bar{X}g_i=g_i\bar{X}$ for all $i \in \{1,\dots,8\}$. This physical representative plays the role of the logical bit-flip operator on the code space.\\
\textit{Code space.---} The states, $\ket{\bar{0}}$ and $\ket{\bar{1}}$, stabilized by $g_1$ $\cdot\cdot\cdot$, $g_8$ and, respectively, by  $\bar{Z}$ and $-\bar{Z}$, are found to be
\begin{eqnarray}
         |\bar{0}\rangle&=&(\ket{\boldsymbol{0}}+\ket{\boldsymbol{1}})^{\otimes 3}/\sqrt{8}\text{ and }\\
    |\bar{1}\rangle&=&(\ket{\boldsymbol{0}}-\ket{\boldsymbol{1}})^{\otimes 3}/\sqrt{8}.
    \end{eqnarray}
    Here we have used the notation $\ket{\boldsymbol{0}}=\ket{000}$ and $\ket{\boldsymbol{1}}=\ket{111}$. The set of states $\{\ket{\bar{0}},\ket{\bar{1}}\}$ is a basis of the code space, and its span forms the entire code space.

This stabilizer group therefore defines a code with parameters $[[9,1,3]]$: it encodes $k=1$ logical qubit into $n=9$ physical qubits, and the resulting two-dimensional code space is spanned by $\ket{\bar{0}}$ and $\ket{\bar{1}}$. This is the Shor code.
 \subsection{The physical gates encoding action of the logical gates}
 Now we focus only on the two-dimensional code space spanned by $\{\ket{\bar{0}},\ket{\bar{1}}\}$. The physical representatives $\bar{X}=Z^{\otimes 9}$ and $\bar{Z}=X^{\otimes 9}$ reproduce the logical bit-flip and phase-flip actions on this code space:
\begin{eqnarray}
\bar{X}\ket{\bar{0}}=\ket{\bar{1}}\text{ and }\bar{X}\ket{\bar{1}}=\ket{\bar{0}},\\  \bar{Z}\ket{\bar{0}}=\ket{\bar{0}}\text{ and }\bar{Z}\ket{\bar{1}}=-\ket{\bar{1}}.
\end{eqnarray}
With the phase-repetition convention used for the Shor code, the physical all-$X$ operator represents logical $\bar Z$, while the physical all-$Z$ operator represents logical $\bar X$.

For the cloud-computation protocol, we also need physical representatives of the logical gates $\bar{S}$ and $\bar{T}$ such that, on the code space,
\begin{eqnarray}
\bar{S}\ket{\bar{0}}=\ket{\bar{0}}\text{, }\bar{S}\ket{\bar{1}}=i\ket{\bar{1}}\text{ and }\label{ap1}\\
\bar{T}\ket{\bar{0}}=\ket{\bar{0}}\text{, }\bar{T}\ket{\bar{1}}=e^{i\pi/4}\ket{\bar{1}}\label{ap2}.
\end{eqnarray}
These equations specify only the logical action. They do not by themselves define a physical gate on the full $2^9$-dimensional Hilbert space. A valid deterministic physical implementation must be unitary on the whole Hilbert space, or else be realized as a specified completely positive trace-preserving operation with measurements and ancillas.

Let
\begin{equation}
P_1=\ket{\bar{1}}\bra{\bar{1}}
\end{equation}
be the projector onto the one-dimensional logical state $\ket{\bar{1}}$. The simplest unitary extensions of the desired logical gates are obtained by leaving the entire orthogonal complement of $\ket{\bar{1}}$ invariant:
\begin{eqnarray}
\bar{S}&=&I+(i-1)P_1,\\
\bar{T}&=&I+\left(e^{i\pi/4}-1\right)P_1.
\end{eqnarray}
Equivalently, choose any orthonormal basis of the full $2^9$-dimensional physical Hilbert space whose first two vectors are $\ket{\bar{0}}$ and $\ket{\bar{1}}$, and whose remaining vectors $\{\ket{\chi_j}\}_{j=1}^{2^9-2}$ span the orthogonal complement of the Shor code space. In this ordered basis, one possible unitary extension is
\begin{eqnarray}
\bar{S}&=&\operatorname{diag}(1,i,I_{2^9-2}),\\
\bar{T}&=&\operatorname{diag}(1,e^{i\pi/4},I_{2^9-2}).
\end{eqnarray}
These operators are unitary on the full physical Hilbert space and reproduce the desired logical action on the Shor code space. Thus, the non-Clifford logical action can be specified directly as a unitary extension on the physical Hilbert space, rather than through magic-state injection. Since the logical encryption acts as $\bar X^a\bar Z^b$ on the encoded qubit, the encrypted logical state remains within the code space. The unitary extensions $\bar S$ and $\bar T$ restrict to the desired logical $S$ and $T$ gates on this subspace, and hence are compatible with the logical-encryption part of the protocol. The action of these extensions on error spaces, and their synthesis into an efficient and fault-tolerant physical circuit, remain separate implementation questions.
\subsection{Homomorphic encryption scheme}
Once the gates $\bar{X}$, $\bar{Z}$, $\bar{S}$, and $\bar{T}$ have been specified, the application of the logical $\bar{T}$ gate to a logical qubit can be performed following the conventional gate-teleportation steps. The steps of the process are described below in detail.\\
\textit{Setup.---} Since we want to apply only a single logical $\bar{T}$ gate to a single logical qubit, the client generates one encoded logical Bell state, $\ket{\bar{\Phi}}_{s_pc_p}=(\ket{\bar{0}}\ket{\bar{0}}+\ket{\bar{1}}\ket{\bar{1}})/\sqrt{2}$. This state involves a total of 18 physical qubits. The first and second blocks of 9 qubits are denoted by $s_p$ and $c_p$, respectively, where the subscript $p$ indicates physical encoded blocks. Let the initial logical state on which the client wants to apply the logical $\bar{T}$ gate be encoded using the Shor code in the physical state $\ket{\bar{\psi}}_{w_p}$, with 9 physical qubits. Here the system of the 9 physical qubits is jointly denoted by $w_p$.\\
    \textit{Encryption.---} To encrypt the data, the client generates a pair of binary numbers $(a,b)$ and applies the gates $\bar{X}^a$ and $\bar{Z}^b$ on the state, $\ket{\bar{\psi}}_{w_p}$.
    The final encrypted state is $\ket{\bar{\psi}^{\text{enc}}}_{w_p}=\bar{X}^a\bar{Z}^b\ket{\bar{\psi}}_{w_p}$. At the end of the encryption process, the client sends this encrypted block and the Bell-pair block $s_p$ to the server, while keeping the block $c_p$.\\
  \textit{Evaluation.---}
The server receives in total 18 physical qubits, 9 of which are prepared in the encrypted state $\ket{\bar{\psi}^{\text{enc}}}_{w_p}$ and the remaining qubits form the system $s_p$ and are part of the state $\ket{\bar{\Phi}}_{s_pc_p}$.
Assuming that a unitary physical implementation of the logical gate $\bar{T}$ is available, the server now applies it on the encrypted state, $\ket{\bar{\psi}^{\text{enc}}}_{w_p}$, and swaps this state of the system of 9 qubits, $w_p$, with the state of the 9 physical qubits of $s_p$. The final joint state of the 27 physical qubits at the end of this process is
  \begin{equation}     \bar{T}\bar{X}^a\bar{Z}^b\ket{\bar{\psi}}_{s_p} \ket{\bar{\Phi}}_{w_pc_p}.
  \end{equation}
The server sends the two blocks in its possession, $w_p$ and $s_p$, back to the client; the block $c_p$ has remained with the client throughout the protocol.\\
  \textit{Decryption.---}
  By construction, on the code space, the gates $\bar{X}$, $\bar{Z}$, $\bar{S}$, and $\bar{T}$ act as the usual Pauli $X$, Pauli $Z$, $S$, and $T$ gates, respectively. Therefore, the gates follow the same commutation relations in this code space as shown in Figure \ref{fig2}. Using these commutation properties, and ignoring irrelevant global phases, the joint state of $w_p$, $s_p$, and $c_p$ that currently belongs to the client can be written as
\begin{eqnarray}&&\bar{T}\bar{X}^a\bar{Z}^b\ket{\bar{\psi}}_{s_p} \ket{\bar{\Phi}}_{w_pc_p}\nonumber\\
&&\simeq(\bar{S}^\dagger)^a\bar{X}^a\bar{Z}^{a\oplus b}\bar{T}\ket{\bar{\psi}}_{s_p} \ket{\bar{\Phi}}_{w_pc_p}\nonumber\\
&&=\frac{1}{2}\sum_{r_a,r_b} \ket{\Phi(\bar{S}^{a})_{r_ar_b}}_{s_pc_p}[\bar{X}^{a\oplus r_a}\bar{Z}^{a\oplus b\oplus r_b} \bar{T}]\ket{\bar{\psi}}_{w_p},\nonumber
  \end{eqnarray}
 where $\simeq$ denotes equality up to a global phase, and $r_a$ and $r_b$ are binary numbers.
 At this point, to decrypt the information, the client applies a measurement in the logical rotated Bell basis $\{\ket{\Phi(\bar S^{a})_{r_ar_b}}_{s_pc_p}\}_{r_a,r_b}$ on the joint system $s_pc_p$. This is an orthonormal logical Bell basis on the two encoded blocks. If the output of the measurement is $r_a$ and $r_b$, then the corresponding output state is
 \begin{equation}
 \ket{\bar{\psi}_{out}^{\text{enc}}}_{w_p}=\bar{X}^{a\oplus r_a}\bar{Z}^{a\oplus b\oplus r_b}\bar{T}\ket{\bar{\psi}}_{w_p}.
 \end{equation}
 Finally, knowing all the binary numbers $a$, $b$, $r_a$, and $r_b$, the client decrypts the state by applying the physical representatives of $\bar{Z}^{a\oplus b\oplus r_b}\bar{X}^{a\oplus r_a}$ on $\ket{\bar{\psi}_{out}^{\text{enc}}}_{w_p}$ and obtains the state $\bar{T}\ket{\bar{\psi}}_{w_p}$, which is the targeted state.

\bibliography{HE_QEC}

\begin{thebibliography}{60}%
\makeatletter
\providecommand \@ifxundefined [1]{%
 \@ifx{#1\undefined}
}%
\providecommand \@ifnum [1]{%
 \ifnum #1\expandafter \@firstoftwo
 \else \expandafter \@secondoftwo
 \fi
}%
\providecommand \@ifx [1]{%
 \ifx #1\expandafter \@firstoftwo
 \else \expandafter \@secondoftwo
 \fi
}%
\providecommand \natexlab [1]{#1}%
\providecommand \enquote  [1]{``#1''}%
\providecommand \bibnamefont  [1]{#1}%
\providecommand \bibfnamefont [1]{#1}%
\providecommand \citenamefont [1]{#1}%
\providecommand \href@noop [0]{\@secondoftwo}%
\providecommand \href [0]{\begingroup \@sanitize@url \@href}%
\providecommand \@href[1]{\@@startlink{#1}\@@href}%
\providecommand \@@href[1]{\endgroup#1\@@endlink}%
\providecommand \@sanitize@url [0]{\catcode `\\12\catcode `\$12\catcode
  `\&12\catcode `\#12\catcode `\^12\catcode `\_12\catcode `\%12\relax}%
\providecommand \@@startlink[1]{}%
\providecommand \@@endlink[0]{}%
\providecommand \url  [0]{\begingroup\@sanitize@url \@url }%
\providecommand \@url [1]{\endgroup\@href {#1}{\urlprefix }}%
\providecommand \urlprefix  [0]{URL }%
\providecommand \Eprint [0]{\href }%
\providecommand \doibase [0]{http://dx.doi.org/}%
\providecommand \selectlanguage [0]{\@gobble}%
\providecommand \bibinfo  [0]{\@secondoftwo}%
\providecommand \bibfield  [0]{\@secondoftwo}%
\providecommand \translation [1]{[#1]}%
\providecommand \BibitemOpen [0]{}%
\providecommand \bibitemStop [0]{}%
\providecommand \bibitemNoStop [0]{.\EOS\space}%
\providecommand \EOS [0]{\spacefactor3000\relax}%
\providecommand \BibitemShut  [1]{\csname bibitem#1\endcsname}%
\let\auto@bib@innerbib\@empty
\bibitem [{\citenamefont {Rohde}\ \emph {et~al.}(2012)\citenamefont {Rohde},
  \citenamefont {Fitzsimons},\ and\ \citenamefont
  {Gilchrist}}]{PhysRevLett109150501}%
  \BibitemOpen
  \bibfield  {author} {\bibinfo {author} {\bibfnamefont {P.~P.}\ \bibnamefont
  {Rohde}}, \bibinfo {author} {\bibfnamefont {J.~F.}\ \bibnamefont
  {Fitzsimons}}, \ and\ \bibinfo {author} {\bibfnamefont {A.}~\bibnamefont
  {Gilchrist}},\ }\bibfield  {title} {\enquote {\bibinfo {title} {Quantum walks
  with encrypted data},}\ }\href {\doibase 10.1103/PhysRevLett.109.150501}
  {\bibfield  {journal} {\bibinfo  {journal} {Phys. Rev. Lett.}\ }\textbf
  {\bibinfo {volume} {109}},\ \bibinfo {pages} {150501} (\bibinfo {year}
  {2012})}\BibitemShut {NoStop}%
\bibitem [{\citenamefont {Liang}(2013)}]{Liang_2013}%
  \BibitemOpen
  \bibfield  {author} {\bibinfo {author} {\bibfnamefont {M.}~\bibnamefont
  {Liang}},\ }\bibfield  {title} {\enquote {\bibinfo {title} {Symmetric quantum
  fully homomorphic encryption with perfect security},}\ }\href {\doibase
  10.1007/s11128-013-0626-5} {\bibfield  {journal} {\bibinfo  {journal}
  {Quantum Information Processing}\ }\textbf {\bibinfo {volume} {12}},\
  \bibinfo {pages} {3675} (\bibinfo {year} {2013})}\BibitemShut {NoStop}%
\bibitem [{\citenamefont {Tan}\ \emph {et~al.}(2016)\citenamefont {Tan},
  \citenamefont {Kettlewell}, \citenamefont {Ouyang}, \citenamefont {Chen},\
  and\ \citenamefont {Fitzsimons}}]{Tan_2016}%
  \BibitemOpen
  \bibfield  {author} {\bibinfo {author} {\bibfnamefont {S.-H.}\ \bibnamefont
  {Tan}}, \bibinfo {author} {\bibfnamefont {J.~A.}\ \bibnamefont {Kettlewell}},
  \bibinfo {author} {\bibfnamefont {Y.}~\bibnamefont {Ouyang}}, \bibinfo
  {author} {\bibfnamefont {L.}~\bibnamefont {Chen}}, \ and\ \bibinfo {author}
  {\bibfnamefont {J.~F.}\ \bibnamefont {Fitzsimons}},\ }\bibfield  {title}
  {\enquote {\bibinfo {title} {A quantum approach to homomorphic encryption},}\
  }\href {http://dx.doi.org/10.1038/srep33467} {\bibfield  {journal} {\bibinfo
  {journal} {Scientific Reports}\ }\textbf {\bibinfo {volume} {6}},\ \bibinfo
  {pages} {33467} (\bibinfo {year} {2016})}\BibitemShut {NoStop}%
\bibitem [{\citenamefont {Broadbent}\ and\ \citenamefont
  {Jeffery}(2015)}]{Broadbent_2015}%
  \BibitemOpen
  \bibfield  {author} {\bibinfo {author} {\bibfnamefont {A.}~\bibnamefont
  {Broadbent}}\ and\ \bibinfo {author} {\bibfnamefont {S.}~\bibnamefont
  {Jeffery}},\ }\bibfield  {title} {\enquote {\bibinfo {title} {Quantum
  homomorphic encryption for circuits of low {T}-gate complexity},}\ }in\ \href
  {\doibase 10.1007/978-3-662-48000-7_30} {\emph {\bibinfo {booktitle}
  {Advances in Cryptology -- CRYPTO 2015}}},\ \bibinfo {series} {Lecture Notes
  in Computer Science}, Vol.\ \bibinfo {volume} {9216}\ (\bibinfo  {publisher}
  {Springer},\ \bibinfo {address} {Berlin, Heidelberg},\ \bibinfo {year}
  {2015})\ pp.\ \bibinfo {pages} {609--629}\BibitemShut {NoStop}%
\bibitem [{\citenamefont {Dulek}\ \emph {et~al.}(2018)\citenamefont {Dulek},
  \citenamefont {Schaffner},\ and\ \citenamefont {Speelman}}]{Dulek_2018}%
  \BibitemOpen
  \bibfield  {author} {\bibinfo {author} {\bibfnamefont {Y.}~\bibnamefont
  {Dulek}}, \bibinfo {author} {\bibfnamefont {C.}~\bibnamefont {Schaffner}}, \
  and\ \bibinfo {author} {\bibfnamefont {F.}~\bibnamefont {Speelman}},\
  }\bibfield  {title} {\enquote {\bibinfo {title} {Quantum homomorphic
  encryption for polynomial-size circuits},}\ }\href {\doibase
  10.4086/toc.2018.v014a007} {\bibfield  {journal} {\bibinfo  {journal} {Theory
  of Computing}\ }\textbf {\bibinfo {volume} {14}},\ \bibinfo {pages} {1--45}
  (\bibinfo {year} {2018})}\BibitemShut {NoStop}%
\bibitem [{\citenamefont {Liang}(2020)}]{Liang30}%
  \BibitemOpen
  \bibfield  {author} {\bibinfo {author} {\bibfnamefont {M.}~\bibnamefont
  {Liang}},\ }\bibfield  {title} {\enquote {\bibinfo {title}
  {Teleportation-based quantum homomorphic encryption scheme with
  quasi-compactness and perfect security},}\ }\href
  {https://link.springer.com/article/10.1007/s11128-019-2529-6} {\bibfield
  {journal} {\bibinfo  {journal} {Quantum Information Processing}\ }\textbf
  {\bibinfo {volume} {19}},\ \bibinfo {pages} {28} (\bibinfo {year}
  {2020})}\BibitemShut {NoStop}%
\bibitem [{\citenamefont {Yu}\ \emph {et~al.}(2014)\citenamefont {Yu},
  \citenamefont {P\'erez-Delgado},\ and\ \citenamefont {Fitzsimons}}]{nogo1}%
  \BibitemOpen
  \bibfield  {author} {\bibinfo {author} {\bibfnamefont {L.}~\bibnamefont
  {Yu}}, \bibinfo {author} {\bibfnamefont {C.~A.}\ \bibnamefont
  {P\'erez-Delgado}}, \ and\ \bibinfo {author} {\bibfnamefont {J.~F.}\
  \bibnamefont {Fitzsimons}},\ }\bibfield  {title} {\enquote {\bibinfo {title}
  {Limitations on information-theoretically-secure quantum homomorphic
  encryption},}\ }\href {\doibase 10.1103/PhysRevA.90.050303} {\bibfield
  {journal} {\bibinfo  {journal} {Phys. Rev. A}\ }\textbf {\bibinfo {volume}
  {90}},\ \bibinfo {pages} {050303} (\bibinfo {year} {2014})}\BibitemShut
  {NoStop}%
\bibitem [{\citenamefont {Lai}\ and\ \citenamefont {Chung}(2018)}]{nogo2}%
  \BibitemOpen
  \bibfield  {author} {\bibinfo {author} {\bibfnamefont {C.-Y.}\ \bibnamefont
  {Lai}}\ and\ \bibinfo {author} {\bibfnamefont {K.-M.}\ \bibnamefont
  {Chung}},\ }\bibfield  {title} {\enquote {\bibinfo {title} {On
  statistically-secure quantum homomorphic encryption},}\ }\href
  {https://arxiv.org/abs/1705.00139} {\bibfield  {journal} {\bibinfo  {journal}
  {Quantum Inf. Comput.}\ }\textbf {\bibinfo {volume} {18}},\ \bibinfo {pages}
  {785--794} (\bibinfo {year} {2018})},\ \Eprint
  {http://arxiv.org/abs/1705.00139} {arXiv:1705.00139} \BibitemShut {NoStop}%
\bibitem [{\citenamefont {Fern\'andez}\ and\ \citenamefont
  {Martin-Delgado}(2024)}]{Grover}%
  \BibitemOpen
  \bibfield  {author} {\bibinfo {author} {\bibfnamefont {P.}~\bibnamefont
  {Fern\'andez}}\ and\ \bibinfo {author} {\bibfnamefont {M.~A.}\ \bibnamefont
  {Martin-Delgado}},\ }\bibfield  {title} {\enquote {\bibinfo {title}
  {Implementing the {G}rover algorithm in homomorphic encryption schemes},}\
  }\href {\doibase 10.1103/PhysRevResearch.6.043109} {\bibfield  {journal}
  {\bibinfo  {journal} {Phys. Rev. Res.}\ }\textbf {\bibinfo {volume} {6}},\
  \bibinfo {pages} {043109} (\bibinfo {year} {2024})}\BibitemShut {NoStop}%
\bibitem [{\citenamefont {Fernández}\ and\ \citenamefont
  {Martin-Delgado}()}]{BV}%
  \BibitemOpen
  \bibfield  {author} {\bibinfo {author} {\bibfnamefont {P.}~\bibnamefont
  {Fernández}}\ and\ \bibinfo {author} {\bibfnamefont {M.~A.}\ \bibnamefont
  {Martin-Delgado}},\ }\bibfield  {title} {\enquote {\bibinfo {title}
  {Homomorphic encryption of the k=2 {B}ernstein-{V}azirani algorithm},}\
  }\href {https://arxiv.org/abs/2303.17426} {\ }\Eprint
  {http://arxiv.org/abs/2303.17426} {arXiv:2303.17426} \BibitemShut {NoStop}%
\bibitem [{\citenamefont {Ortega}\ \emph {et~al.}(2025)\citenamefont {Ortega},
  \citenamefont {Fernández},\ and\ \citenamefont {Martin-Delgado}}]{SQW}%
  \BibitemOpen
  \bibfield  {author} {\bibinfo {author} {\bibfnamefont {S.~A.}\ \bibnamefont
  {Ortega}}, \bibinfo {author} {\bibfnamefont {P.}~\bibnamefont {Fernández}},
  \ and\ \bibinfo {author} {\bibfnamefont {M.~A.}\ \bibnamefont
  {Martin-Delgado}},\ }\bibfield  {title} {\enquote {\bibinfo {title}
  {Implementing semiclassical {S}zegedy walks in classical-quantum circuits for
  homomorphic encryption},}\ }\href {\doibase 10.1088/2632-072x/add3aa}
  {\bibfield  {journal} {\bibinfo  {journal} {Journal of Physics: Complexity}\
  }\textbf {\bibinfo {volume} {6}},\ \bibinfo {pages} {025010} (\bibinfo {year}
  {2025})}\BibitemShut {NoStop}%
\bibitem [{\citenamefont {Gong}\ \emph {et~al.}(2020)\citenamefont {Gong},
  \citenamefont {Du}, \citenamefont {Dong}, \citenamefont {Guo}, \citenamefont
  {Gani}, \citenamefont {Zhao},\ and\ \citenamefont {Qi}}]{Grover2}%
  \BibitemOpen
  \bibfield  {author} {\bibinfo {author} {\bibfnamefont {C.}~\bibnamefont
  {Gong}}, \bibinfo {author} {\bibfnamefont {J.}~\bibnamefont {Du}}, \bibinfo
  {author} {\bibfnamefont {Z.}~\bibnamefont {Dong}}, \bibinfo {author}
  {\bibfnamefont {Z.}~\bibnamefont {Guo}}, \bibinfo {author} {\bibfnamefont
  {A.}~\bibnamefont {Gani}}, \bibinfo {author} {\bibfnamefont {L.}~\bibnamefont
  {Zhao}}, \ and\ \bibinfo {author} {\bibfnamefont {H.}~\bibnamefont {Qi}},\
  }\bibfield  {title} {\enquote {\bibinfo {title} {Grover algorithm-based
  quantum homomorphic encryption ciphertext retrieval scheme in quantum cloud
  computing},}\ }\href
  {https://link.springer.com/article/10.1007/s11128-020-2603-0#citeas}
  {\bibfield  {journal} {\bibinfo  {journal} {Quantum Inf Process}\ }\textbf
  {\bibinfo {volume} {19}},\ \bibinfo {pages} {105} (\bibinfo {year}
  {2020})}\BibitemShut {NoStop}%
\bibitem [{\citenamefont {Zeuner}\ \emph {et~al.}()\citenamefont {Zeuner},
  \citenamefont {Pitsios}, \citenamefont {Tan}, \citenamefont {Sharma},
  \citenamefont {Fitzsimons}, \citenamefont {Osellame},\ and\ \citenamefont
  {Walther}}]{exp1}%
  \BibitemOpen
  \bibfield  {author} {\bibinfo {author} {\bibfnamefont {J.}~\bibnamefont
  {Zeuner}}, \bibinfo {author} {\bibfnamefont {I.}~\bibnamefont {Pitsios}},
  \bibinfo {author} {\bibfnamefont {S.-H.}\ \bibnamefont {Tan}}, \bibinfo
  {author} {\bibfnamefont {A.~N.}\ \bibnamefont {Sharma}}, \bibinfo {author}
  {\bibfnamefont {J.~F.}\ \bibnamefont {Fitzsimons}}, \bibinfo {author}
  {\bibfnamefont {R.}~\bibnamefont {Osellame}}, \ and\ \bibinfo {author}
  {\bibfnamefont {P.}~\bibnamefont {Walther}},\ }\bibfield  {title} {\enquote
  {\bibinfo {title} {Experimental quantum homomorphic encryption},}\ }\href
  {https://arxiv.org/abs/1803.10246} {\ }\Eprint
  {http://arxiv.org/abs/1803.10246} {arXiv:1803.10246} \BibitemShut {NoStop}%
\bibitem [{\citenamefont {Tham}\ \emph {et~al.}(2020)\citenamefont {Tham},
  \citenamefont {Ferretti}, \citenamefont {Bonsma-Fisher}, \citenamefont
  {Brodutch}, \citenamefont {Sanders}, \citenamefont {Steinberg},\ and\
  \citenamefont {Jeffery}}]{exp2}%
  \BibitemOpen
  \bibfield  {author} {\bibinfo {author} {\bibfnamefont {W.~K.}\ \bibnamefont
  {Tham}}, \bibinfo {author} {\bibfnamefont {H.}~\bibnamefont {Ferretti}},
  \bibinfo {author} {\bibfnamefont {K.}~\bibnamefont {Bonsma-Fisher}}, \bibinfo
  {author} {\bibfnamefont {A.}~\bibnamefont {Brodutch}}, \bibinfo {author}
  {\bibfnamefont {B.~C.}\ \bibnamefont {Sanders}}, \bibinfo {author}
  {\bibfnamefont {A.~M.}\ \bibnamefont {Steinberg}}, \ and\ \bibinfo {author}
  {\bibfnamefont {S.}~\bibnamefont {Jeffery}},\ }\bibfield  {title} {\enquote
  {\bibinfo {title} {Experimental demonstration of quantum fully homomorphic
  encryption with application in a two-party secure protocol},}\ }\href
  {\doibase 10.1103/PhysRevX.10.011038} {\bibfield  {journal} {\bibinfo
  {journal} {Phys. Rev. X}\ }\textbf {\bibinfo {volume} {10}},\ \bibinfo
  {pages} {011038} (\bibinfo {year} {2020})}\BibitemShut {NoStop}%
\bibitem [{\citenamefont {Nielsen}\ and\ \citenamefont
  {Chuang}(2000)}]{NielsenChuang2000}%
  \BibitemOpen
  \bibfield  {author} {\bibinfo {author} {\bibfnamefont {M.~A.}\ \bibnamefont
  {Nielsen}}\ and\ \bibinfo {author} {\bibfnamefont {I.~L.}\ \bibnamefont
  {Chuang}},\ }\href@noop {} {\emph {\bibinfo {title} {Quantum Computation and
  Quantum Information}}}\ (\bibinfo  {publisher} {Cambridge University Press},\
  \bibinfo {address} {Cambridge},\ \bibinfo {year} {2000})\BibitemShut
  {NoStop}%
\bibitem [{\citenamefont {Galindo}\ and\ \citenamefont
  {Martin-Delgado}(2002)}]{GalindoMartinDelgado2002RMP}%
  \BibitemOpen
  \bibfield  {author} {\bibinfo {author} {\bibfnamefont {A.}~\bibnamefont
  {Galindo}}\ and\ \bibinfo {author} {\bibfnamefont {M.~A.}\ \bibnamefont
  {Martin-Delgado}},\ }\bibfield  {title} {\enquote {\bibinfo {title}
  {Information and computation: Classical and quantum aspects},}\ }\href
  {\doibase 10.1103/RevModPhys.74.347} {\bibfield  {journal} {\bibinfo
  {journal} {Rev. Mod. Phys.}\ }\textbf {\bibinfo {volume} {74}},\ \bibinfo
  {pages} {347--423} (\bibinfo {year} {2002})}\BibitemShut {NoStop}%
\bibitem [{\citenamefont {DiVincenzo}(2000)}]{DiVincenzo_2000}%
  \BibitemOpen
  \bibfield  {author} {\bibinfo {author} {\bibfnamefont {D.~P.}\ \bibnamefont
  {DiVincenzo}},\ }\bibfield  {title} {\enquote {\bibinfo {title} {The physical
  implementation of quantum computation},}\ }\href {\doibase
  10.1002/1521-3978(200009)48:9/11<771::aid-prop771>3.0.co;2-e} {\bibfield
  {journal} {\bibinfo  {journal} {Fortschritte der Physik}\ }\textbf {\bibinfo
  {volume} {48}},\ \bibinfo {pages} {771} (\bibinfo {year} {2000})}\BibitemShut
  {NoStop}%
\bibitem [{\citenamefont {Devoret}\ \emph {et~al.}()\citenamefont {Devoret},
  \citenamefont {Wallraff},\ and\ \citenamefont {Martinis}}]{superqubits}%
  \BibitemOpen
  \bibfield  {author} {\bibinfo {author} {\bibfnamefont {M.~H.}\ \bibnamefont
  {Devoret}}, \bibinfo {author} {\bibfnamefont {A.}~\bibnamefont {Wallraff}}, \
  and\ \bibinfo {author} {\bibfnamefont {J.~M.}\ \bibnamefont {Martinis}},\
  }\bibfield  {title} {\enquote {\bibinfo {title} {Superconducting qubits: A
  short review},}\ }\href {https://arxiv.org/abs/cond-mat/0411174} {\ }\Eprint
  {http://arxiv.org/abs/cond-mat/0411174} {arXiv:cond-mat/0411174} \BibitemShut
  {NoStop}%
\bibitem [{\citenamefont {Garc\'{\i}a-Ripoll}\ \emph
  {et~al.}(2008)\citenamefont {Garc\'{\i}a-Ripoll}, \citenamefont {Solano},\
  and\ \citenamefont {Martin-Delgado}}]{PhysRevB77024522}%
  \BibitemOpen
  \bibfield  {author} {\bibinfo {author} {\bibfnamefont {J.~J.}\ \bibnamefont
  {Garc\'{\i}a-Ripoll}}, \bibinfo {author} {\bibfnamefont {E.}~\bibnamefont
  {Solano}}, \ and\ \bibinfo {author} {\bibfnamefont {M.~A.}\ \bibnamefont
  {Martin-Delgado}},\ }\bibfield  {title} {\enquote {\bibinfo {title} {Quantum
  simulation of {A}nderson and {K}ondo lattices with superconducting qubits},}\
  }\href {\doibase 10.1103/PhysRevB.77.024522} {\bibfield  {journal} {\bibinfo
  {journal} {Phys. Rev. B}\ }\textbf {\bibinfo {volume} {77}},\ \bibinfo
  {pages} {024522} (\bibinfo {year} {2008})}\BibitemShut {NoStop}%
\bibitem [{\citenamefont {Devoret}\ and\ \citenamefont
  {Schoelkopf}(2013)}]{qc}%
  \BibitemOpen
  \bibfield  {author} {\bibinfo {author} {\bibfnamefont {M.~H.}\ \bibnamefont
  {Devoret}}\ and\ \bibinfo {author} {\bibfnamefont {R.~J.}\ \bibnamefont
  {Schoelkopf}},\ }\bibfield  {title} {\enquote {\bibinfo {title}
  {Superconducting circuits for quantum information: An outlook},}\ }\href
  {\doibase 10.1126/science.1231930} {\bibfield  {journal} {\bibinfo  {journal}
  {Science}\ }\textbf {\bibinfo {volume} {339}},\ \bibinfo {pages} {1169}
  (\bibinfo {year} {2013})}\BibitemShut {NoStop}%
\bibitem [{\citenamefont {Wendin}(2017)}]{Wendin_2017}%
  \BibitemOpen
  \bibfield  {author} {\bibinfo {author} {\bibfnamefont {G.}~\bibnamefont
  {Wendin}},\ }\bibfield  {title} {\enquote {\bibinfo {title} {Quantum
  information processing with superconducting circuits: a review},}\ }\href
  {\doibase 10.1088/1361-6633/aa7e1a} {\bibfield  {journal} {\bibinfo
  {journal} {Reports on Progress in Physics}\ }\textbf {\bibinfo {volume}
  {80}},\ \bibinfo {pages} {106001} (\bibinfo {year} {2017})}\BibitemShut
  {NoStop}%
\bibitem [{\citenamefont {Gambetta}\ \emph {et~al.}(2017)\citenamefont
  {Gambetta}, \citenamefont {Chow},\ and\ \citenamefont
  {Steffen}}]{Gambetta_2017}%
  \BibitemOpen
  \bibfield  {author} {\bibinfo {author} {\bibfnamefont {J.~M.}\ \bibnamefont
  {Gambetta}}, \bibinfo {author} {\bibfnamefont {J.~M.}\ \bibnamefont {Chow}},
  \ and\ \bibinfo {author} {\bibfnamefont {M.}~\bibnamefont {Steffen}},\
  }\bibfield  {title} {\enquote {\bibinfo {title} {Building logical qubits in a
  superconducting quantum computing system},}\ }\href
  {http://dx.doi.org/10.1038/s41534-016-0004-0} {\bibfield  {journal} {\bibinfo
   {journal} {npj Quantum Information}\ }\textbf {\bibinfo {volume} {3}},\
  \bibinfo {pages} {2} (\bibinfo {year} {2017})}\BibitemShut {NoStop}%
\bibitem [{\citenamefont {Lacroix}\ \emph {et~al.}(2025)\citenamefont
  {Lacroix}, \citenamefont {Bourassa}, \citenamefont {Heras}, \citenamefont
  {Zhang}, \citenamefont {Bausch}, \citenamefont {Senior}, \citenamefont
  {Satzinger} \emph {et~al.}}]{GoogleColorCode2025}%
  \BibitemOpen
  \bibfield  {author} {\bibinfo {author} {\bibfnamefont {N.}~\bibnamefont
  {Lacroix}}, \bibinfo {author} {\bibfnamefont {A.}~\bibnamefont {Bourassa}},
  \bibinfo {author} {\bibfnamefont {F.~J.~H.}\ \bibnamefont {Heras}}, \bibinfo
  {author} {\bibfnamefont {L.~M.}\ \bibnamefont {Zhang}}, \bibinfo {author}
  {\bibfnamefont {J.}~\bibnamefont {Bausch}}, \bibinfo {author} {\bibfnamefont
  {A.~W.}\ \bibnamefont {Senior}}, \bibinfo {author} {\bibfnamefont {K.~J.}\
  \bibnamefont {Satzinger}},  \emph {et~al.},\ }\bibfield  {title} {\enquote
  {\bibinfo {title} {Scaling and logic in the colour code on a superconducting
  quantum processor},}\ }\href {\doibase 10.1038/s41586-025-09061-4} {\bibfield
   {journal} {\bibinfo  {journal} {Nature}\ }\textbf {\bibinfo {volume}
  {645}},\ \bibinfo {pages} {614--619} (\bibinfo {year} {2025})}\BibitemShut
  {NoStop}%
\bibitem [{\citenamefont {Cirac}\ and\ \citenamefont
  {Zoller}(1995)}]{PhysRevLett.74.4091}%
  \BibitemOpen
  \bibfield  {author} {\bibinfo {author} {\bibfnamefont {J.~I.}\ \bibnamefont
  {Cirac}}\ and\ \bibinfo {author} {\bibfnamefont {P.}~\bibnamefont {Zoller}},\
  }\bibfield  {title} {\enquote {\bibinfo {title} {Quantum computations with
  cold trapped ions},}\ }\href {\doibase 10.1103/PhysRevLett.74.4091}
  {\bibfield  {journal} {\bibinfo  {journal} {Phys. Rev. Lett.}\ }\textbf
  {\bibinfo {volume} {74}},\ \bibinfo {pages} {4091} (\bibinfo {year}
  {1995})}\BibitemShut {NoStop}%
\bibitem [{\citenamefont {Blinov}\ \emph {et~al.}(2004)\citenamefont {Blinov},
  \citenamefont {Leibfried}, \citenamefont {Monroe},\ and\ \citenamefont
  {Wineland}}]{trapped_ion}%
  \BibitemOpen
  \bibfield  {author} {\bibinfo {author} {\bibfnamefont {B.~B.}\ \bibnamefont
  {Blinov}}, \bibinfo {author} {\bibfnamefont {D.}~\bibnamefont {Leibfried}},
  \bibinfo {author} {\bibfnamefont {C.}~\bibnamefont {Monroe}}, \ and\ \bibinfo
  {author} {\bibfnamefont {D.~J.}\ \bibnamefont {Wineland}},\ }\bibfield
  {title} {\enquote {\bibinfo {title} {Quantum computing with trapped ion
  hyperfine qubits},}\ }\href
  {https://link.springer.com/article/10.1007/s11128-004-9417-3} {\bibfield
  {journal} {\bibinfo  {journal} {Quantum Information Processing}\ }\textbf
  {\bibinfo {volume} {3}},\ \bibinfo {pages} {45} (\bibinfo {year}
  {2004})}\BibitemShut {NoStop}%
\bibitem [{\citenamefont {et~al.}(2023{\natexlab{a}})}]{PhysRevX.13.041052}%
  \BibitemOpen
  \bibfield  {author} {\bibinfo {author} {\bibfnamefont {S.~A.~Moses}\
  \bibnamefont {et~al.}},\ }\bibfield  {title} {\enquote {\bibinfo {title} {A
  race-track trapped-ion quantum processor},}\ }\href {\doibase
  10.1103/PhysRevX.13.041052} {\bibfield  {journal} {\bibinfo  {journal} {Phys.
  Rev. X}\ }\textbf {\bibinfo {volume} {13}},\ \bibinfo {pages} {041052}
  (\bibinfo {year} {2023}{\natexlab{a}})}\BibitemShut {NoStop}%
\bibitem [{\citenamefont {Nigg}\ \emph {et~al.}(2014)\citenamefont {Nigg},
  \citenamefont {M{\"u}ller}, \citenamefont {Martinez}, \citenamefont
  {Schindler}, \citenamefont {Hennrich}, \citenamefont {Monz}, \citenamefont
  {Martin-Delgado},\ and\ \citenamefont {Blatt}}]{Nigg2014TopologicalQubit}%
  \BibitemOpen
  \bibfield  {author} {\bibinfo {author} {\bibfnamefont {D.}~\bibnamefont
  {Nigg}}, \bibinfo {author} {\bibfnamefont {M.}~\bibnamefont {M{\"u}ller}},
  \bibinfo {author} {\bibfnamefont {E.~A.}\ \bibnamefont {Martinez}}, \bibinfo
  {author} {\bibfnamefont {P.}~\bibnamefont {Schindler}}, \bibinfo {author}
  {\bibfnamefont {M.}~\bibnamefont {Hennrich}}, \bibinfo {author}
  {\bibfnamefont {T.}~\bibnamefont {Monz}}, \bibinfo {author} {\bibfnamefont
  {M.~A.}\ \bibnamefont {Martin-Delgado}}, \ and\ \bibinfo {author}
  {\bibfnamefont {R.}~\bibnamefont {Blatt}},\ }\bibfield  {title} {\enquote
  {\bibinfo {title} {Quantum computations on a topologically encoded qubit},}\
  }\href {\doibase 10.1126/science.1253742} {\bibfield  {journal} {\bibinfo
  {journal} {Science}\ }\textbf {\bibinfo {volume} {345}},\ \bibinfo {pages}
  {302--305} (\bibinfo {year} {2014})}\BibitemShut {NoStop}%
\bibitem [{\citenamefont {Bloch}\ \emph {et~al.}(2012)\citenamefont {Bloch},
  \citenamefont {Dalibard},\ and\ \citenamefont {Nascimbène}}]{cold_atom1}%
  \BibitemOpen
  \bibfield  {author} {\bibinfo {author} {\bibfnamefont {I.}~\bibnamefont
  {Bloch}}, \bibinfo {author} {\bibfnamefont {J.}~\bibnamefont {Dalibard}}, \
  and\ \bibinfo {author} {\bibfnamefont {S.}~\bibnamefont {Nascimbène}},\
  }\bibfield  {title} {\enquote {\bibinfo {title} {Quantum simulations with
  ultracold quantum gases},}\ }\href
  {https://www.nature.com/articles/nphys2259} {\bibfield  {journal} {\bibinfo
  {journal} {Nature Physics}\ }\textbf {\bibinfo {volume} {8}},\ \bibinfo
  {pages} {267} (\bibinfo {year} {2012})}\BibitemShut {NoStop}%
\bibitem [{\citenamefont {et~al.}(2023{\natexlab{b}})}]{Bluvstein_2023}%
  \BibitemOpen
  \bibfield  {author} {\bibinfo {author} {\bibfnamefont {D.~Bluvstein}\
  \bibnamefont {et~al.}},\ }\bibfield  {title} {\enquote {\bibinfo {title}
  {Logical quantum processor based on reconfigurable atom arrays},}\ }\href
  {\doibase 10.1038/s41586-023-06927-3} {\bibfield  {journal} {\bibinfo
  {journal} {Nature}\ }\textbf {\bibinfo {volume} {626}},\ \bibinfo {pages}
  {58} (\bibinfo {year} {2023}{\natexlab{b}})}\BibitemShut {NoStop}%
\bibitem [{\citenamefont {Evered}\ \emph {et~al.}(2025)\citenamefont {Evered}
  \emph {et~al.}}]{Evered2025NeutralAtomArchitecture}%
  \BibitemOpen
  \bibfield  {author} {\bibinfo {author} {\bibfnamefont {S.~J.}\ \bibnamefont
  {Evered}} \emph {et~al.},\ }\bibfield  {title} {\enquote {\bibinfo {title} {A
  fault-tolerant neutral-atom architecture for universal quantum
  computation},}\ }\href {\doibase 10.1038/s41586-025-09848-5} {\bibfield
  {journal} {\bibinfo  {journal} {Nature}\ } (\bibinfo {year} {2025}),\
  10.1038/s41586-025-09848-5}\BibitemShut {NoStop}%
\bibitem [{\citenamefont {Gottesman}()}]{QEC1}%
  \BibitemOpen
  \bibfield  {author} {\bibinfo {author} {\bibfnamefont {D.}~\bibnamefont
  {Gottesman}},\ }\bibfield  {title} {\enquote {\bibinfo {title} {An
  introduction to quantum error correction and fault-tolerant quantum
  computation},}\ }\href {https://arxiv.org/abs/0904.2557} {\ }\Eprint
  {http://arxiv.org/abs/0904.2557} {arXiv:0904.2557} \BibitemShut {NoStop}%
\bibitem [{\citenamefont {Roffe}(2019)}]{Roffe_2019}%
  \BibitemOpen
  \bibfield  {author} {\bibinfo {author} {\bibfnamefont {J.}~\bibnamefont
  {Roffe}},\ }\bibfield  {title} {\enquote {\bibinfo {title} {Quantum error
  correction: an introductory guide},}\ }\href
  {http://dx.doi.org/10.1080/00107514.2019.1667078} {\bibfield  {journal}
  {\bibinfo  {journal} {Contemporary Physics}\ }\textbf {\bibinfo {volume}
  {60}},\ \bibinfo {pages} {226} (\bibinfo {year} {2019})}\BibitemShut
  {NoStop}%
\bibitem [{\citenamefont {Gottesman}(1996)}]{stabilizer}%
  \BibitemOpen
  \bibfield  {author} {\bibinfo {author} {\bibfnamefont {D.}~\bibnamefont
  {Gottesman}},\ }\bibfield  {title} {\enquote {\bibinfo {title} {Class of
  quantum error-correcting codes saturating the quantum hamming bound},}\
  }\href {https://link.aps.org/doi/10.1103/PhysRevA.54.1862} {\bibfield
  {journal} {\bibinfo  {journal} {Phys. Rev. A}\ }\textbf {\bibinfo {volume}
  {54}},\ \bibinfo {pages} {1862} (\bibinfo {year} {1996})}\BibitemShut
  {NoStop}%
\bibitem [{\citenamefont {Shor}(1995)}]{Shor}%
  \BibitemOpen
  \bibfield  {author} {\bibinfo {author} {\bibfnamefont {P.~W.}\ \bibnamefont
  {Shor}},\ }\bibfield  {title} {\enquote {\bibinfo {title} {Scheme for
  reducing decoherence in quantum computer memory},}\ }\href
  {https://link.aps.org/doi/10.1103/PhysRevA.52.R2493} {\bibfield  {journal}
  {\bibinfo  {journal} {Phys. Rev. A}\ }\textbf {\bibinfo {volume} {52}},\
  \bibinfo {pages} {R2493} (\bibinfo {year} {1995})}\BibitemShut {NoStop}%
\bibitem [{\citenamefont {Laflamme}\ \emph {et~al.}(1996)\citenamefont
  {Laflamme}, \citenamefont {Miquel}, \citenamefont {Paz},\ and\ \citenamefont
  {Zurek}}]{fivequbit}%
  \BibitemOpen
  \bibfield  {author} {\bibinfo {author} {\bibfnamefont {R.}~\bibnamefont
  {Laflamme}}, \bibinfo {author} {\bibfnamefont {C.}~\bibnamefont {Miquel}},
  \bibinfo {author} {\bibfnamefont {J.~P.}\ \bibnamefont {Paz}}, \ and\
  \bibinfo {author} {\bibfnamefont {W.~H.}\ \bibnamefont {Zurek}},\ }\bibfield
  {title} {\enquote {\bibinfo {title} {Perfect quantum error correcting
  code},}\ }\href {https://link.aps.org/doi/10.1103/PhysRevLett.77.198}
  {\bibfield  {journal} {\bibinfo  {journal} {Phys. Rev. Lett.}\ }\textbf
  {\bibinfo {volume} {77}},\ \bibinfo {pages} {198} (\bibinfo {year}
  {1996})}\BibitemShut {NoStop}%
\bibitem [{\citenamefont {Knill}\ \emph {et~al.}(2001)\citenamefont {Knill},
  \citenamefont {Laflamme}, \citenamefont {Martinez},\ and\ \citenamefont
  {Negrevergne}}]{fivequbit_exp}%
  \BibitemOpen
  \bibfield  {author} {\bibinfo {author} {\bibfnamefont {E.}~\bibnamefont
  {Knill}}, \bibinfo {author} {\bibfnamefont {R.}~\bibnamefont {Laflamme}},
  \bibinfo {author} {\bibfnamefont {R.}~\bibnamefont {Martinez}}, \ and\
  \bibinfo {author} {\bibfnamefont {C.}~\bibnamefont {Negrevergne}},\
  }\bibfield  {title} {\enquote {\bibinfo {title} {Benchmarking quantum
  computers: The five-qubit error correcting code},}\ }\href {\doibase
  10.1103/PhysRevLett.86.5811} {\bibfield  {journal} {\bibinfo  {journal}
  {Phys. Rev. Lett.}\ }\textbf {\bibinfo {volume} {86}},\ \bibinfo {pages}
  {5811} (\bibinfo {year} {2001})}\BibitemShut {NoStop}%
\bibitem [{\citenamefont {Steane}(1996)}]{CSS1}%
  \BibitemOpen
  \bibfield  {author} {\bibinfo {author} {\bibfnamefont {A.}~\bibnamefont
  {Steane}},\ }\bibfield  {title} {\enquote {\bibinfo {title}
  {Multiple-particle interference and quantum error correction},}\ }\href
  {http://dx.doi.org/10.1098/rspa.1996.0136} {\bibfield  {journal} {\bibinfo
  {journal} {Proceedings of the Royal Society of London. Series A:
  Mathematical, Physical and Engineering Sciences}\ }\textbf {\bibinfo {volume}
  {452}},\ \bibinfo {pages} {2551} (\bibinfo {year} {1996})}\BibitemShut
  {NoStop}%
\bibitem [{\citenamefont {Calderbank}\ and\ \citenamefont {Shor}(1996)}]{CSS2}%
  \BibitemOpen
  \bibfield  {author} {\bibinfo {author} {\bibfnamefont {A.~R.}\ \bibnamefont
  {Calderbank}}\ and\ \bibinfo {author} {\bibfnamefont {P.~W.}\ \bibnamefont
  {Shor}},\ }\bibfield  {title} {\enquote {\bibinfo {title} {Good quantum
  error-correcting codes exist},}\ }\href
  {https://link.aps.org/doi/10.1103/PhysRevA.54.1098} {\bibfield  {journal}
  {\bibinfo  {journal} {Phys. Rev. A}\ }\textbf {\bibinfo {volume} {54}},\
  \bibinfo {pages} {1098} (\bibinfo {year} {1996})}\BibitemShut {NoStop}%
\bibitem [{\citenamefont {Bravyi}\ and\ \citenamefont {Kitaev}()}]{toriccode1}%
  \BibitemOpen
  \bibfield  {author} {\bibinfo {author} {\bibfnamefont {S.~B.}\ \bibnamefont
  {Bravyi}}\ and\ \bibinfo {author} {\bibfnamefont {A.~Y.}\ \bibnamefont
  {Kitaev}},\ }\bibfield  {title} {\enquote {\bibinfo {title} {Quantum codes on
  a lattice with boundary},}\ }\href {https://arxiv.org/abs/quant-ph/9811052}
  {\ }\Eprint {http://arxiv.org/abs/quant-ph/9811052} {arXiv:quant-ph/9811052}
  \BibitemShut {NoStop}%
\bibitem [{\citenamefont {Kitaev}(2003)}]{toriccode2}%
  \BibitemOpen
  \bibfield  {author} {\bibinfo {author} {\bibfnamefont {A.~Y.}\ \bibnamefont
  {Kitaev}},\ }\bibfield  {title} {\enquote {\bibinfo {title} {Fault-tolerant
  quantum computation by anyons},}\ }\href
  {https://www.sciencedirect.com/science/article/pii/S0003491602000180}
  {\bibfield  {journal} {\bibinfo  {journal} {Annals of Physics}\ }\textbf
  {\bibinfo {volume} {303}},\ \bibinfo {pages} {2} (\bibinfo {year}
  {2003})}\BibitemShut {NoStop}%
\bibitem [{\citenamefont {Xu}\ \emph {et~al.}()\citenamefont {Xu},
  \citenamefont {Zhong}, \citenamefont {Martin-Delgado}, \citenamefont {Song},\
  and\ \citenamefont {Liu}}]{toriccode3}%
  \BibitemOpen
  \bibfield  {author} {\bibinfo {author} {\bibfnamefont {J.-Z.}\ \bibnamefont
  {Xu}}, \bibinfo {author} {\bibfnamefont {Y.}~\bibnamefont {Zhong}}, \bibinfo
  {author} {\bibfnamefont {M.~A.}\ \bibnamefont {Martin-Delgado}}, \bibinfo
  {author} {\bibfnamefont {H.}~\bibnamefont {Song}}, \ and\ \bibinfo {author}
  {\bibfnamefont {K.}~\bibnamefont {Liu}},\ }\bibfield  {title} {\enquote
  {\bibinfo {title} {Phenomenological noise models and optimal thresholds of
  the 3{D} toric code},}\ }\href {https://arxiv.org/abs/2510.20489} {\ }\Eprint
  {http://arxiv.org/abs/2510.20489} {arXiv:2510.20489} \BibitemShut {NoStop}%
\bibitem [{\citenamefont {Bombin}\ and\ \citenamefont
  {Martin-Delgado}(2006)}]{colorcode}%
  \BibitemOpen
  \bibfield  {author} {\bibinfo {author} {\bibfnamefont {H.}~\bibnamefont
  {Bombin}}\ and\ \bibinfo {author} {\bibfnamefont {M.~A.}\ \bibnamefont
  {Martin-Delgado}},\ }\bibfield  {title} {\enquote {\bibinfo {title}
  {Topological quantum distillation},}\ }\href
  {https://link.aps.org/doi/10.1103/PhysRevLett.97.180501} {\bibfield
  {journal} {\bibinfo  {journal} {Phys. Rev. Lett.}\ }\textbf {\bibinfo
  {volume} {97}},\ \bibinfo {pages} {180501} (\bibinfo {year}
  {2006})}\BibitemShut {NoStop}%
\bibitem [{\citenamefont {Bombin}\ and\ \citenamefont
  {Martin-Delgado}(2007{\natexlab{a}})}]{BombinMartinDelgado2007PRL}%
  \BibitemOpen
  \bibfield  {author} {\bibinfo {author} {\bibfnamefont {H.}~\bibnamefont
  {Bombin}}\ and\ \bibinfo {author} {\bibfnamefont {M.~A.}\ \bibnamefont
  {Martin-Delgado}},\ }\bibfield  {title} {\enquote {\bibinfo {title}
  {Topological computation without braiding},}\ }\href {\doibase
  10.1103/PhysRevLett.98.160502} {\bibfield  {journal} {\bibinfo  {journal}
  {Phys. Rev. Lett.}\ }\textbf {\bibinfo {volume} {98}},\ \bibinfo {pages}
  {160502} (\bibinfo {year} {2007}{\natexlab{a}})}\BibitemShut {NoStop}%
\bibitem [{\citenamefont {Bombin}\ and\ \citenamefont
  {Martin-Delgado}(2007{\natexlab{b}})}]{BombinMartinDelgado2007PRB}%
  \BibitemOpen
  \bibfield  {author} {\bibinfo {author} {\bibfnamefont {H.}~\bibnamefont
  {Bombin}}\ and\ \bibinfo {author} {\bibfnamefont {M.~A.}\ \bibnamefont
  {Martin-Delgado}},\ }\bibfield  {title} {\enquote {\bibinfo {title} {Exact
  topological quantum order in ${D}=3$ and beyond: Branyons and brane-net
  condensates},}\ }\href {\doibase 10.1103/PhysRevB.75.075103} {\bibfield
  {journal} {\bibinfo  {journal} {Phys. Rev. B}\ }\textbf {\bibinfo {volume}
  {75}},\ \bibinfo {pages} {075103} (\bibinfo {year}
  {2007}{\natexlab{b}})}\BibitemShut {NoStop}%
\bibitem [{\citenamefont {Ouyang}\ \emph {et~al.}(2018)\citenamefont {Ouyang},
  \citenamefont {Tan},\ and\ \citenamefont
  {Fitzsimons}}]{OuyangTanFitzsimons2018}%
  \BibitemOpen
  \bibfield  {author} {\bibinfo {author} {\bibfnamefont {Y.}~\bibnamefont
  {Ouyang}}, \bibinfo {author} {\bibfnamefont {S.-H.}\ \bibnamefont {Tan}}, \
  and\ \bibinfo {author} {\bibfnamefont {J.~F.}\ \bibnamefont {Fitzsimons}},\
  }\bibfield  {title} {\enquote {\bibinfo {title} {Quantum homomorphic
  encryption from quantum codes},}\ }\href {\doibase
  10.1103/PhysRevA.98.042334} {\bibfield  {journal} {\bibinfo  {journal} {Phys.
  Rev. A}\ }\textbf {\bibinfo {volume} {98}},\ \bibinfo {pages} {042334}
  (\bibinfo {year} {2018})}\BibitemShut {NoStop}%
\bibitem [{\citenamefont {Ouyang}\ and\ \citenamefont
  {Rohde}(2022)}]{OuyangRohde2022}%
  \BibitemOpen
  \bibfield  {author} {\bibinfo {author} {\bibfnamefont {Y.}~\bibnamefont
  {Ouyang}}\ and\ \bibinfo {author} {\bibfnamefont {P.~P.}\ \bibnamefont
  {Rohde}},\ }\bibfield  {title} {\enquote {\bibinfo {title} {Permutational-key
  quantum homomorphic encryption with homomorphic quantum error-correction},}\
  }\href {\doibase 10.48550/arXiv.2204.10471} {\bibfield  {journal} {\bibinfo
  {journal} {arXiv preprint arXiv:2204.10471}\ } (\bibinfo {year} {2022}),\
  10.48550/arXiv.2204.10471}\BibitemShut {NoStop}%
\bibitem [{\citenamefont {Sohn}\ \emph {et~al.}(2024)\citenamefont {Sohn},
  \citenamefont {Kim}, \citenamefont {Bae},\ and\ \citenamefont
  {Lee}}]{Sohn2024ErrorCorrectable}%
  \BibitemOpen
  \bibfield  {author} {\bibinfo {author} {\bibfnamefont {I.}~\bibnamefont
  {Sohn}}, \bibinfo {author} {\bibfnamefont {B.}~\bibnamefont {Kim}}, \bibinfo
  {author} {\bibfnamefont {K.}~\bibnamefont {Bae}}, \ and\ \bibinfo {author}
  {\bibfnamefont {W.}~\bibnamefont {Lee}},\ }\bibfield  {title} {\enquote
  {\bibinfo {title} {Error correctable efficient quantum homomorphic
  encryption},}\ }\href {\doibase 10.48550/arXiv.2401.08059} {\bibfield
  {journal} {\bibinfo  {journal} {arXiv preprint arXiv:2401.08059}\ } (\bibinfo
  {year} {2024}),\ 10.48550/arXiv.2401.08059}\BibitemShut {NoStop}%
\bibitem [{\citenamefont {Steane}()}]{steane1996quantumreedmullercodes}%
  \BibitemOpen
  \bibfield  {author} {\bibinfo {author} {\bibfnamefont {A.}~\bibnamefont
  {Steane}},\ }\bibfield  {title} {\enquote {\bibinfo {title} {Quantum
  {R}eed-{M}uller codes},}\ }\href {https://arxiv.org/abs/quant-ph/9608026} {\
  }\Eprint {http://arxiv.org/abs/quant-ph/9608026} {arXiv:quant-ph/9608026}
  \BibitemShut {NoStop}%
\bibitem [{\citenamefont {Bombin}\ and\ \citenamefont
  {Martin-Delgado}(2007{\natexlab{c}})}]{BombinMartinDelgado2007PRA}%
  \BibitemOpen
  \bibfield  {author} {\bibinfo {author} {\bibfnamefont {H.}~\bibnamefont
  {Bombin}}\ and\ \bibinfo {author} {\bibfnamefont {M.~A.}\ \bibnamefont
  {Martin-Delgado}},\ }\bibfield  {title} {\enquote {\bibinfo {title} {Optimal
  resources for topological two-dimensional stabilizer codes: Comparative
  study},}\ }\href {\doibase 10.1103/PhysRevA.76.012305} {\bibfield  {journal}
  {\bibinfo  {journal} {Phys. Rev. A}\ }\textbf {\bibinfo {volume} {76}},\
  \bibinfo {pages} {012305} (\bibinfo {year} {2007}{\natexlab{c}})}\BibitemShut
  {NoStop}%
\bibitem [{\citenamefont {Horsman}\ \emph {et~al.}(2012)\citenamefont
  {Horsman}, \citenamefont {Fowler}, \citenamefont {Devitt},\ and\
  \citenamefont {Van~Meter}}]{Horsman2012}%
  \BibitemOpen
  \bibfield  {author} {\bibinfo {author} {\bibfnamefont {C.}~\bibnamefont
  {Horsman}}, \bibinfo {author} {\bibfnamefont {A.~G.}\ \bibnamefont {Fowler}},
  \bibinfo {author} {\bibfnamefont {S.}~\bibnamefont {Devitt}}, \ and\ \bibinfo
  {author} {\bibfnamefont {R.}~\bibnamefont {Van~Meter}},\ }\bibfield  {title}
  {\enquote {\bibinfo {title} {Surface code quantum computing by lattice
  surgery},}\ }\href {\doibase 10.1088/1367-2630/14/12/123011} {\bibfield
  {journal} {\bibinfo  {journal} {New J. Phys.}\ }\textbf {\bibinfo {volume}
  {14}},\ \bibinfo {pages} {123011} (\bibinfo {year} {2012})}\BibitemShut
  {NoStop}%
\bibitem [{\citenamefont {Vaidman}\ \emph {et~al.}(1996)\citenamefont
  {Vaidman}, \citenamefont {Goldenberg},\ and\ \citenamefont
  {Wiesner}}]{Vaidman1996}%
  \BibitemOpen
  \bibfield  {author} {\bibinfo {author} {\bibfnamefont {L.}~\bibnamefont
  {Vaidman}}, \bibinfo {author} {\bibfnamefont {L.}~\bibnamefont {Goldenberg}},
  \ and\ \bibinfo {author} {\bibfnamefont {S.}~\bibnamefont {Wiesner}},\
  }\bibfield  {title} {\enquote {\bibinfo {title} {Error prevention scheme with
  four particles},}\ }\href {\doibase 10.1103/PhysRevA.54.R1745} {\bibfield
  {journal} {\bibinfo  {journal} {Phys. Rev. A}\ }\textbf {\bibinfo {volume}
  {54}},\ \bibinfo {pages} {R1745} (\bibinfo {year} {1996})}\BibitemShut
  {NoStop}%
\bibitem [{\citenamefont {Grassl}\ \emph {et~al.}(1997)\citenamefont {Grassl},
  \citenamefont {Beth},\ and\ \citenamefont {Pellizzari}}]{Grassl1997}%
  \BibitemOpen
  \bibfield  {author} {\bibinfo {author} {\bibfnamefont {M.}~\bibnamefont
  {Grassl}}, \bibinfo {author} {\bibfnamefont {T.}~\bibnamefont {Beth}}, \ and\
  \bibinfo {author} {\bibfnamefont {T.}~\bibnamefont {Pellizzari}},\ }\bibfield
   {title} {\enquote {\bibinfo {title} {Codes for the quantum erasure
  channel},}\ }\href {\doibase 10.1103/PhysRevA.56.33} {\bibfield  {journal}
  {\bibinfo  {journal} {Phys. Rev. A}\ }\textbf {\bibinfo {volume} {56}},\
  \bibinfo {pages} {33} (\bibinfo {year} {1997})}\BibitemShut {NoStop}%
\bibitem [{\citenamefont {Campbell}(2016)}]{Campbell2016}%
  \BibitemOpen
  \bibfield  {author} {\bibinfo {author} {\bibfnamefont {E.~T.}\ \bibnamefont
  {Campbell}},\ }\href@noop {} {\enquote {\bibinfo {title} {The smallest
  interesting colour code},}\ }\bibinfo {howpublished}
  {\url{https://earltcampbell.com/2016/09/26/the-smallest-interesting-colour-code/}}
  (\bibinfo {year} {2016})\BibitemShut {NoStop}%
\bibitem [{\citenamefont {Vasmer}\ and\ \citenamefont
  {Kubica}(2022)}]{VasmerKubica2022}%
  \BibitemOpen
  \bibfield  {author} {\bibinfo {author} {\bibfnamefont {M.}~\bibnamefont
  {Vasmer}}\ and\ \bibinfo {author} {\bibfnamefont {A.}~\bibnamefont
  {Kubica}},\ }\bibfield  {title} {\enquote {\bibinfo {title} {Morphing quantum
  codes},}\ }\href {\doibase 10.1103/PRXQuantum.3.030319} {\bibfield  {journal}
  {\bibinfo  {journal} {PRX Quantum}\ }\textbf {\bibinfo {volume} {3}},\
  \bibinfo {pages} {030319} (\bibinfo {year} {2022})}\BibitemShut {NoStop}%
\bibitem [{\citenamefont {Gidney}(2021)}]{Gidney2021}%
  \BibitemOpen
  \bibfield  {author} {\bibinfo {author} {\bibfnamefont {C.}~\bibnamefont
  {Gidney}},\ }\bibfield  {title} {\enquote {\bibinfo {title} {Stim: a fast
  stabilizer circuit simulator},}\ }\href {\doibase 10.22331/q-2021-07-06-497}
  {\bibfield  {journal} {\bibinfo  {journal} {Quantum}\ }\textbf {\bibinfo
  {volume} {5}},\ \bibinfo {pages} {497} (\bibinfo {year} {2021})}\BibitemShut
  {NoStop}%
\bibitem [{\citenamefont {Bravyi}\ and\ \citenamefont
  {Haah}(2012)}]{PhysRevA.86.052329}%
  \BibitemOpen
  \bibfield  {author} {\bibinfo {author} {\bibfnamefont {S.}~\bibnamefont
  {Bravyi}}\ and\ \bibinfo {author} {\bibfnamefont {J.}~\bibnamefont {Haah}},\
  }\bibfield  {title} {\enquote {\bibinfo {title} {Magic-state distillation
  with low overhead},}\ }\href {\doibase 10.1103/PhysRevA.86.052329} {\bibfield
   {journal} {\bibinfo  {journal} {Phys. Rev. A}\ }\textbf {\bibinfo {volume}
  {86}},\ \bibinfo {pages} {052329} (\bibinfo {year} {2012})}\BibitemShut
  {NoStop}%
\bibitem [{\citenamefont {Camps-Moreno}\ \emph {et~al.}(2024)\citenamefont
  {Camps-Moreno}, \citenamefont {López}, \citenamefont {Matthews},
  \citenamefont {Ruano}, \citenamefont {San-Jos\'e},\ and\ \citenamefont
  {Soprunov}}]{10735332}%
  \BibitemOpen
  \bibfield  {author} {\bibinfo {author} {\bibfnamefont {E.}~\bibnamefont
  {Camps-Moreno}}, \bibinfo {author} {\bibfnamefont {H.~H.}\ \bibnamefont
  {López}}, \bibinfo {author} {\bibfnamefont {G.~L.}\ \bibnamefont
  {Matthews}}, \bibinfo {author} {\bibfnamefont {D.}~\bibnamefont {Ruano}},
  \bibinfo {author} {\bibfnamefont {R.}~\bibnamefont {San-Jos\'e}}, \ and\
  \bibinfo {author} {\bibfnamefont {I.}~\bibnamefont {Soprunov}},\ }\bibfield
  {title} {\enquote {\bibinfo {title} {Binary triorthogonal and {CSS-T} codes
  for quantum error correction},}\ }\href {\doibase
  10.1109/Allerton63246.2024.10735332} {\bibfield  {journal} {\bibinfo
  {journal} {2024 60th Annual Allerton Conference on Communication, Control,
  and Computing}\ ,\ \bibinfo {pages} {01}} (\bibinfo {year}
  {2024})}\BibitemShut {NoStop}%
\bibitem [{\citenamefont {Shi}\ \emph {et~al.}()\citenamefont {Shi},
  \citenamefont {Lu}, \citenamefont {Kim},\ and\ \citenamefont
  {Sole}}]{shi2024triorthogonalcodesselfdualcodes}%
  \BibitemOpen
  \bibfield  {author} {\bibinfo {author} {\bibfnamefont {M.}~\bibnamefont
  {Shi}}, \bibinfo {author} {\bibfnamefont {H.}~\bibnamefont {Lu}}, \bibinfo
  {author} {\bibfnamefont {J.-L.}\ \bibnamefont {Kim}}, \ and\ \bibinfo
  {author} {\bibfnamefont {P.}~\bibnamefont {Sole}},\ }\bibfield  {title}
  {\enquote {\bibinfo {title} {Triorthogonal codes and self-dual codes},}\
  }\href {https://arxiv.org/abs/2408.09685} {\bibinfo  {journal}
  {arXiv:2408.09685}\ }\BibitemShut {NoStop}%
\bibitem [{\citenamefont {Kubica}\ and\ \citenamefont
  {Beverland}(2015)}]{KubicaBeverland2015}%
  \BibitemOpen
\bibfield  {journal} {  }\bibfield  {author} {\bibinfo {author} {\bibfnamefont
  {A.}~\bibnamefont {Kubica}}\ and\ \bibinfo {author} {\bibfnamefont {M.~E.}\
  \bibnamefont {Beverland}},\ }\bibfield  {title} {\enquote {\bibinfo {title}
  {Universal transversal gates with color codes: A simplified approach},}\
  }\href {\doibase 10.1103/PhysRevA.91.032330} {\bibfield  {journal} {\bibinfo
  {journal} {Phys. Rev. A}\ }\textbf {\bibinfo {volume} {91}},\ \bibinfo
  {pages} {032330} (\bibinfo {year} {2015})}\BibitemShut {NoStop}%
\bibitem [{\citenamefont {Bennett}\ \emph {et~al.}(1993)\citenamefont
  {Bennett}, \citenamefont {Brassard}, \citenamefont {Cr\'epeau}, \citenamefont
  {Jozsa}, \citenamefont {Peres},\ and\ \citenamefont
  {Wootters}}]{PhysRevLett.70.1895}%
  \BibitemOpen
  \bibfield  {author} {\bibinfo {author} {\bibfnamefont {C.~H.}\ \bibnamefont
  {Bennett}}, \bibinfo {author} {\bibfnamefont {G.}~\bibnamefont {Brassard}},
  \bibinfo {author} {\bibfnamefont {C.}~\bibnamefont {Cr\'epeau}}, \bibinfo
  {author} {\bibfnamefont {R.}~\bibnamefont {Jozsa}}, \bibinfo {author}
  {\bibfnamefont {A.}~\bibnamefont {Peres}}, \ and\ \bibinfo {author}
  {\bibfnamefont {W.~K.}\ \bibnamefont {Wootters}},\ }\bibfield  {title}
  {\enquote {\bibinfo {title} {Teleporting an unknown quantum state via dual
  classical and {E}instein-{P}odolsky-{R}osen channels},}\ }\href {\doibase
  10.1103/PhysRevLett.70.1895} {\bibfield  {journal} {\bibinfo  {journal}
  {Phys. Rev. Lett.}\ }\textbf {\bibinfo {volume} {70}},\ \bibinfo {pages}
  {1895} (\bibinfo {year} {1993})}\BibitemShut {NoStop}%
\end{thebibliography}%
\end{document}